%% file: manuscript.tex
\documentclass[%
 reprint,
 amsmath,amssymb,
 aps,
 superscriptaddress,
]{revtex4-2}
\usepackage{graphicx}% Include figure files
\usepackage{bm}% bold math
\usepackage{subcaption}
\graphicspath{{figures/}{images/}{plots/}}
\usepackage{amsmath}
\usepackage{amssymb}
\usepackage{amsthm}
\usepackage{graphicx}
\usepackage{mathrsfs}
\usepackage{braket}
\usepackage{comment}
\usepackage[hidelinks]{hyperref}
\usepackage[capitalize]{cleveref}
\usepackage[dvipsnames]{xcolor}
\usepackage{soul,xcolor}
\usepackage{orcidlink}
\input{preamble}

\graphicspath{{figures/}}

\newcommand{\DeptMath}{Department of Mathematics, University of Michigan, Ann Arbor, MI 48109, USA}
\begin{document}

\title{Beyond Lindblad Approximation: Rigorous Guarantees for Thermal and Ground State Preservation under System-Bath Interactions}

\author{Ke Wang\orcidlink{0009-0001-3225-887X}}
\email{kwmath@umich.edu}
\affiliation{\DeptMath}
\author{Zhiyan Ding\orcidlink{0000-0001-8863-403X}}
\email{zyding@umich.edu}
\affiliation{\DeptMath}
%TC:ignore
\begin{abstract}
We establish rigorous guarantees for quantum thermal and ground state preparation via system-bath interactions beyond the Lindblad approximation regime. In contrast to existing analyses, we show that the system-bath coupling strength need not vanish with the target accuracy: even when the coupling strength is chosen independently of the desired precision, the induced quantum channel admits the target thermal and ground states as approximate fixed points with arbitrarily small error. Our analysis develops new techniques to control all orders of the Dyson expansion and the associated multidimensional operator Fourier transforms, thereby going beyond the leading-order Lindbladian description. Building on these estimates, we derive an improved end-to-end complexity analysis and provide rigorous guarantees for thermalization and cooling outside the conventional Lindblad approximation regime. Numerical simulations further demonstrate the robustness of the system-bath interaction framework over a broad range of coupling strengths, including regimes where the Lindblad approximation is no longer accurate.

% \zd{We can consider change the title to avoid a bad impression in the beginning. How about: Beyond Lindblad approximation: Rigorous Guarantees for Thermal and Ground State Preservation under System-Bath Interactions}
\end{abstract}
\maketitle
\section{Introduction}
Quantum thermal state and ground state preparation are fundamental primitives with broad applications in quantum many body physics, quantum chemistry, and materials science. 
Dissipative quantum algorithms provide a natural route to these tasks by engineering an open quantum systems dynamics whose fixed point is approximately the desired thermal or ground state~\cite{VerstraeteWolfIgnacioCirac2009,RoyChalkerGornyiEtAl2020,zhou2021symmetry,WangSnizhkoRomitoEtAl2023,LuLessaKimEtAl2022,Temme_2011}. % ,li2024dissipative Mozgunov2020, langbehn2025universal ,gilyen2024 Cubitt2023,
A central challenge is to ensure both rigorous guarantees on stationary state accuracy and efficient, physically accessible implementations.

Microscopically, open system dynamics arise from interactions between the system and the environment. In the weak coupling regime, these dynamics are commonly reduced, under Born-Markov and secular approximations, to a Lindblad master equation~\cite{Lindblad1976,GoriniKossakowskiSudarshan1976,breuer2002theory}. With this tractable structure,  earlier work explored this idea through tailoring interactions and the jump operator of the Lindblad equation to the target state or specific Hamiltonian~\cite{andersen2025thermalization,shtanko2021preparing,PhysRevA.104.012414,PhysRevA.101.012328,MiMichailidisShabaniEtAl2024,Lloyd2025Quasiparticle}. More recently, building on Lindbladian dynamics~\cite{TemmeKastoryanoRuskaiEtAl2010,KastoryanoTemme2013,BardetCapelGaoEtAl2023,kochanowski2024rapid,bergamaschi2026fastmixingquantumspin,DingLiLinZhang2024,tong2024fast,smid2025rapidmixingquantumgibbs,smid2025polynomial,rouz2024,rouze2024optimal,zhan2025rapidquantumgroundstate}, general frameworks for fault-tolerant quantum Gibbs sampling~\cite{ChenKastoryanoBrandaoEtAl2025,Ding_2025} and ground-state preparation~\cite{Ding2024} with rigorous complexity guarantees have been developed. These Lindbladian dynamics often rely on carefully engineered jump operators that enable both efficient simulation and rigorous preservation of the target state. Despite their mathematical tractability and efficiency in the fault-tolerant setting, practical implementations require highly structured quantum subroutines. Implementing these algorithms requires quantum oracles for complicated jump operators constructed from operator Fourier transforms of system operators in the Heisenberg picture. Such block encodings typically rely on linear combinations of unitaries, multiple ancilla registers, and controlled Hamiltonian evolutions~\cite{cleve_2017,li_et_2023,ChenKastoryanoBrandaoEtAl2023,PRXQuantum.5.020332}. 

System-bath interaction algorithms provide a more direct alternative~\cite{hahn2025provably,lloyd2025quantumthermal,ding2025endtoendefficientquantumthermal} by including the filter function into the interaction Hamiltonian coefficient. They closely mirror the microscopic interaction model and naturally fit the repeated interaction paradigm~\cite{strasberg2017quantum,PhysRevLett.126.130403,ramonescandell2025thermalstatepreparationrepeated,CICCARELLO20221,PocrnicSegalWiebe2025}.
The algorithms couple the system to a small ancillary bath, simulate the joint Hamiltonian evolution, and refresh the bath at each iteration. These algorithms avoid the explicit implementation of complicated jump operators, requiring only the forward Hamiltonian simulation and coupling through local system operators.

Despite their more direct implementation, existing rigorous analyses of interaction-based algorithms still rely largely on the analysis of the Lindblad dynamics~\cite{hahn2025provably,lloyd2025quantumthermal,ding2025endtoendefficientquantumthermal,slezak2026polynomialtime}. Specifically, in~\cite{ding2025endtoendefficientquantumthermal}, the induced discrete quantum channel admits the expansion with coupling parameter $\Gamma$,
\begin{equation}\label{eqn:Lindblad_approximation_weak}
\Phi_\Gamma
=
\underbrace{\overbrace{\rho
+
    e^{\Gamma^2 \mathcal{L}_{\mathrm{Lind}}}(\rho)}
    ^{\text{Previous works}}
    +
    \Gamma^4(\cdots)
}_{\text{Our analysis}} 
\end{equation}
where $\mathcal{L}_{\rm Lind}$ is the effective Lindbladian. Arbitrarily accurate fixed point preservation is therefore established only after replacing the full interaction channel with its Lindblad approximation in the weak coupling limit $\Gamma \to 0$, and verifying the resulting Lindblad dynamics approximately preserves the target state.
This follows the standard approach in {microscopic open system}~\cite{breuer2002theory} and {repeated interaction models~\cite{strasberg2017quantum,PhysRevLett.126.130403}, where thermalization and cooling are established by reducing the dynamics to Lindblad form.}
Moreover, taking $\Gamma\to0$ introduces an intrinsic algorithmic overhead: the dissipative effect of each interaction step is only $O(\Gamma^2)$, so reducing $\Gamma$ to control the Lindblad approximation error simultaneously increases the number of iterations required for convergence.

In this work, we revisit the system-bath interaction algorithm introduced in~\cite{ding2025endtoendefficientquantumthermal}, but analyze it in a fundamentally different interaction regime. Unlike~\cite{ding2025endtoendefficientquantumthermal}, whose rigorous guarantees rely on the limit $\Gamma\to 0$, we establish correctness and rigorous end-to-end complexity bounds while keeping the coupling strength $\Gamma$ finite and independent of the desired precision.

\emph{Conceptually,} our results show that accurate thermal and ground state preparation via system-bath interactions does not require taking the vanishing coupling limit in which the full interaction dynamics reduces to its leading-order Lindblad approximation. Instead, thermalization and cooling can persist under the full reduced system dynamics at finite coupling, beyond the regime in which only the leading-order Lindbladian description is retained. This provides a different perspective on dissipative quantum algorithms: rather than designing and analyzing them solely through an effective Lindblad generator, one can work directly with the underlying system-bath dynamics. From this viewpoint, higher-order terms in the interaction are not merely approximation errors to be suppressed; they form an intrinsic part of the reduced dynamics and can contribute favorably to state preservation and convergence.

\emph{From the complexity perspective,} working at finite coupling eliminates the need to decrease $\Gamma$ polynomially with the target accuracy, thereby removing a major source of overhead in analyses based on the Lindblad approximation. By combining finite coupling state preservation with quantitative control of the convergence dynamics, we improve the dependence of the end-to-end complexity on the target accuracy from $\epsilon^{-4}$ to $\epsilon^{-2}$ for thermal state preparation and from $\mathrm{poly}(1/\epsilon)$ to $\mathrm{polylog}(1/\epsilon)$ for ground state preparation. These results show that going beyond the Lindblad approximation is not only conceptually possible, but can also yield substantial algorithmic improvements.

The main technical challenge of our work is to control the full system-bath dynamics without truncating it at leading order. To this end, we develop new estimates for all orders of the Dyson expansion of the reduced quantum channel, together with tools for analyzing the associated multidimensional operator Fourier transforms. We further develop a perturbative framework that quantitatively controls the higher-order contributions relative to the leading dissipative dynamics, allowing the coupling strength to remain independent of the target accuracy. We complement the theoretical analysis with numerical experiments on transverse-field Ising, Hubbard, and ANNNI models. The simulations confirm the predicted efficiency improvements within the analytically controlled regime and further exhibit robust convergence at stronger couplings, including regimes beyond the scope of our present theoretical guarantees.

\section{Results}
We focus on the system-bath interaction algorithm in~\cite{ding2025endtoendefficientquantumthermal}, but we expect the analysis of other interaction-based algorithms~\cite{hahn2025provably,lloyd2025quantumthermal} can be extended beyond the Lindblad framework using a similar approach. 
Given a $N$-qubit system Hamiltonian $H$ and inverse temperature $0<\beta<\infty$, the thermal state is defined as $\sigma_\beta = e^{-\beta H}/\mathrm{Tr}(e^{-\beta H})$, with the ground state denoted as $\ketbra{\psi_0}$. The system-bath interaction algorithm~\cite{hahn2025provably,lloyd2025quantumthermal,ding2025endtoendefficientquantumthermal} aims to prepare the target state by iteratively applying a quantum channel $\Phi_{\Gamma}$ generated from a system-bath interaction evolution:
\begin{equation}\label{eqn:Phi_alpha}
\rho_{n+1} = \Phi_{\Gamma}(\rho_n) := \mathbb{E}\left[\mathrm{Tr}_E\left( U^\Gamma(T) \left( \rho_n \otimes \rho_E \right) U^{\Gamma}(T)^\dagger \right)\right]\,.
\end{equation}
Here $U^\Gamma(t) := \mathcal{T} \exp\left(-i \int_{-T}^t H_{\Gamma}(s)\mathrm{d}s\right)$ denotes the time evolution operator of the joint system-bath dynamics up to time $t$. The total Hamiltonian is given by
\begin{equation}\label{eqn:Ham_total}
H_{\Gamma}(t) = H + H_E + \Gamma f(t) \left( A_S \otimes B_E + A_S^\dagger \otimes B_E^\dagger \right)\,,
\end{equation}
where $H$ is the system Hamiltonian. The bath consists of a single ancilla qubit  with $H_E=-{\omega}Z/2$ and $B_E=(X_E-iY_E)/2=\ket{1}\bra{0}$, where the bath frequency $\omega$ is randomly sampled from a probability density $g(\omega)$ to be specified below. The system coupling operator $A_S$ generates transitions between different energy levels of the system and is sampled uniformly from a set  $\mc{A}=\{A^i,-A^i\}_i$ where $\{(A^i)^\dagger\}_i=\{A^i\}_i$ and $\|A^i\|\leq 1$. %For spin systems, we take $\mc{A}=\{\pm X_j,\pm Y_j,\pm Z_j\}_{j=1}^n$, while for fermionic systems, we set $\mc{A}=\{\pm c_j,\pm c_j^\dagger\}_{j=1}^n$.  
$\Gamma$ is the coupling strength and the filter function is chosen as $f(t)= \frac{1}{(2\pi)^{1/4} \sigma} \exp\left(-\frac{t^2}{4\sigma^2}\right)$. Both coupling strength $\Gamma$ and filter width $\sigma$ are user-specified parameters that control the accuracy of the algorithm. 

% $H_E,B_E$ are the bath Hamiltonian and operator. Here, $\Gamma$ is the coupling strength between the system and the bath, and $f(t)$ is a filter function.

 The algorithm approximates the target state by iteratively applying the quantum channel $\Phi_\Gamma$.
 Each iteration requires simulating the time-dependent Hamiltonian $H_{\Gamma}(t)$ over the interval $t \in [-T, T]$. Since the auxiliary terms $H_E$ and $A_S\otimes B_E$ are strictly local, their computational overhead relative to simulating $H$ is negligible.
 The total simulation cost is therefore dominated by the simulation time $2T$ at each iteration and the integer mixing time $\tau_{\rm mix}$ quantifying the number of iterations required for target state preparation. We define the main cost metric as the total simulation time $T_{\rm total} := 2T \tau_{\rm mix}$. In the following complexity analysis, we define the rescaled mixing time $t_{\rm mix} = \frac{\Gamma^2}{\sigma}\tau_{\rm mix}$.
 See~\cref{sec:preliminaries} for the formal  definition.
 We use the standard asymptotic notation, $f = \mathcal{O}(g)$ if $f\leq c g$ for some constant $c>0$, and  $f = \Theta(g)$ if both $f = \mathcal{O}(g)$ and $g = \mathcal{O}(f)$. The notations $ \widetilde{\mathcal{O}}$ and $\widetilde{\Theta}$ suppress polylogarithmic factors.  We refer to \textit{finite coupling} regime when $\Gamma$ is independent of $\epsilon$, and the full channel is no longer controlled solely by its leading order Lindblad approximation as in~\cite{ding2025endtoendefficientquantumthermal,hahn2025provably,lloyd2025quantumthermal,slezak2026polynomialtime}.

\begin{figure*}[!htp]
    \centering
    \resizebox{1.4\columnwidth}{!}{%
    \begin{tikzpicture}[
        node distance=7mm and 9mm,
        box/.style={rectangle, rounded corners=4pt,
                    draw=blue!55!cyan!75!black, line width=0.7pt,
                    fill=blue!5, minimum height=9mm,
                    minimum width=27mm, align=center, inner sep=3.5pt,
                    font=\small},
        widebox/.style={box, minimum width=33mm},
        smallbox/.style={box, minimum width=25mm, minimum height=7mm,
                         font=\scriptsize},
        lbl/.style={inner sep=1pt, font=\scriptsize},
        ar/.style={-{Latex[length=1.8mm,width=1.6mm]},
                   line width=0.7pt, draw=blue!55!cyan!75!black},
    ]

        % Channel expansion
        \node[box] (Dyson)
            {Dyson expansion\\of $\Phi_\Gamma$};

        % Three technical ingredients
        \node[widebox, right=of Dyson, yshift=11mm] (Tools1)
            {Operator Fourier\\transform};

        \node[widebox, right=of Dyson] (Tools2)
            {Unital property\\of $\Phi_\Gamma^\dagger$};

        \node[widebox, right=of Dyson, yshift=-13mm] (Tools3)
            { Orthogonal\\decomposition,\\Fourier transform};

        % Preservation results
        \node[box, right=11mm of Tools2, yshift=7mm] (ThermalPres)
            {\textbf{Thermal state}\\\textbf{preservation}\\Thm~\ref{thm:main_informal_thermal}};

        \node[box, right=11mm of Tools2, yshift=-13mm] (GroundPres)
            {\textbf{Ground state}\\\textbf{preservation}\\Thm.~\ref{thm:main_informal_ground}};

        % Final preparation results
        \node[box, right=of ThermalPres] (ThermalAlg)
            {Thermal state\\preparation\\Thm~\ref{thm:total_run_time}};

        \node[box, right=of GroundPres] (GroundAlg)
            {Ground state\\preparation\\Thm.~\ref{thm:ground_run_time}};

        % Additional inputs
        \node[smallbox, above=4mm of ThermalAlg] (Mixing)
            {Lindbladian\\mixing};

        \node[smallbox, below=4mm of GroundAlg] (Contract)
            {Contraction of\\$\Tr(\rho\hat{N})$};

        % Arrows from Dyson expansion
        \draw[ar] (Dyson.east) -- (Tools1.west);
        \draw[ar] (Dyson.east) -- (Tools2.west);
        \draw[ar] (Dyson.east) -- (Tools3.west);

        % Technical ingredients to preservation
        \draw[ar] (Tools1.east) -- (ThermalPres.west);
        \draw[ar] (Tools2.east) -- (ThermalPres.west);

        \draw[ar] (Tools2.east) -- (GroundPres.west);
        \draw[ar] (Tools3.east) -- (GroundPres.west);

        % Preservation to preparation
        \draw[ar] (ThermalPres.east) -- (ThermalAlg.west);
        \draw[ar] (GroundPres.east) -- (GroundAlg.west);

        % Extra inputs
        \draw[ar] (Mixing.south) --
            node[lbl, right=1pt] {perturbation}
            (ThermalAlg.north);

        \draw[ar] (Contract.north) -- (GroundAlg.south);

    \end{tikzpicture}%
    }
    \caption{Flowchart of the proof structure.}
    \label{fig:proof-flowchart}
\end{figure*}

In this work, we establish state preservation guarantees and end-to-end complexity analysis in the finite coupling regime, following the structure illustrated in~\cref{fig:proof-flowchart}. We emphasize that the fixed point guarantees in~\cref{thm:main_informal_thermal,thm:main_informal_ground} and the end-to-end complexity of ground state preparation are established with a coupling strength  independent of both the target accuracy $\epsilon$ and the system size $N$. For thermal state preparation,  the end-to-end complexity bound in~\cref{thm:total_run_time} requires the coupling parameter depending on the system size $N$ as a consequence of the perturbative framework. 

% We emphasize that the fixed-point guarantees are proved for a coupling strength $\Gamma$ that is independent of both the target accuracy $\epsilon$ and the system size $N$. For thermal state preparation, however, the end-to-end complexity bound in~\cref{thm:total_run_time} is derived within a perturbative framework and therefore requires the coupling parameter to depend on the system size $N$. In contrast, the end-to-end complexity guarantee for ground state preparation continues to hold with a finite coupling strength $\Gamma=\mathcal{O}(1)$. 

\subsection{State Preservation}
% {\it{Main Result}}-- 
We prove the following state preservation results in the finite coupling regime.

\begin{thm}[Informal:  Fixed-Point approximation of thermal state]\label{thm:main_informal_thermal} In the finite coupling regime, $\Gamma = \mathcal{O}(1)$, for any  accuracy $\epsilon>0$, for any inverse temperature $\beta>0$, the quantum channel $\Phi_\Gamma$ from the system-bath interaction model has an approximate Gibbs fixed point, 
\begin{equation}
    \left\|\rho_{\rm fix}(\Phi_{\Gamma})-\sigma_\beta\right\|_1\leq\mathcal{O}\left( \frac{\beta t_{\rm mix}}{\sigma} + \frac{\sigma^2t_{\rm mix}}{T} e^{-T^2/(4\sigma^2)}+\epsilon\right)
\end{equation}
In particular,  choosing the coupling strength $T = \widetilde{\Theta}(\sigma)$ and $\sigma = \mathcal{O}(\beta t_{\rm mix} /\epsilon)$ gives $\left\|\rho_{\rm fix}(\Phi_{\Gamma})-\sigma_\beta\right\|_1\leq \epsilon$.
\end{thm}
\begin{thm}[Informal: Fixed-Point approximation of ground state]\label{thm:main_informal_ground}
In the finite coupling regime, $\Gamma = \mathcal{O}(1)$, for any accuracy $\epsilon>0$,
for any $H$ with a unique ground state $\ket{\psi_0}$ and spectral gap $\Delta$, 
the quantum channel $\Phi_\Gamma$ from system-bath interaction model has an approximate ground state fixed point, 
\begin{equation}
\begin{aligned}
&\left\|\rho_{\rm fix}(\Phi_{\Gamma})-\ketbra{\psi_0}\right\|_1\\
&\leq \mathcal{O}\left(\sigma t_{\rm mix}e^{-\sigma\Delta} + \frac{\sigma^2 t_{\rm mix}}{T} e^{-T^2/(4\sigma^2)}+\epsilon\right)
\end{aligned}
\end{equation}
In particular,   choosing $T = \widetilde{\Theta}(\sigma)$, and $\sigma=\widetilde{\Theta}(\Delta^{-1}\mathrm{log}(t_{\rm mix}/\epsilon))$ gives $\left\|\rho_{\rm fix}(\Phi_{\Gamma})-{\ketbra{\psi_0}}\right\|_1\leq \epsilon$.
\end{thm}

% \zd{The preservation result tells us no information. What do you want the reader to learn after reading this result?}

Here, we omit logarithmic factors in $\|H\|$, $\beta$, and $1/\epsilon$ for simplicity. 
% The $t_{\rm mix}:=\Gamma^2\tau_{\rm mix}/\sigma$ defined in~\cref{sec:preliminaries}~\cref{eqn:rescaled_mixing_time} is the rescaled mixing time of $\Phi_{\Gamma}$. This is a standard characterization in the literature~\cite{ding2025endtoendefficientquantumthermal,slezak2026polynomialtime}, and removes the explicit $\sigma$ dependence. See \emph{Analysis overview} for details.
The full statements and proofs are provided in~\cref{sec:thermal_state_appendix} and~\cref{sec:ground_state_appendix}.

According to~\cref{thm:main_informal_thermal,thm:main_informal_ground},  
% as long as $t_{\rm mix}$ is upper bounded, 
choosing $\Gamma$ as a sufficiently small constant independent of all other parameters ensures that the fixed point of $\Phi_{\Gamma}$ can be made arbitrarily close to the target thermal or ground state by taking $\sigma$ sufficiently large. This contrasts with previous analyses on the weak coupling regime.
As indicated in~\cref{eqn:Lindblad_approximation_weak}, previous proofs first replace the full interaction-induced channel by its leading order Lindblad approximation, and then prove state preservation for the resulting Lindblad dynamics.
% proof that the fixed point of $\Phi_\Gamma$ approximates the target state requires two steps: replacing the full interaction induced channel by its Lindblad approximation, and the fixed point analyses of the resulting Lindblad dynamics.
A rigorous proof
that the fixed point of $\Phi_\Gamma$ approximates the target state to high accuracy therefore build on the vanishing higher order term $\mathcal{O}(\Gamma^4)$ with $\epsilon$. To guarantee an arbitrary target accuracy $\epsilon$, existing analyses choose a vanishing coupling strength of
order $\mathcal{O}(\mathrm{poly}(\epsilon))$~\cite{hahn2025provably,lloyd2025quantumthermal,ding2025endtoendefficientquantumthermal,slezak2026polynomialtime}.
% To achieve an arbitrary high accuracy $\epsilon$, we need to allow the above reduction to be valid by requiring the higher order term $\mathcal{O}(\Gamma^4)$ vanishes with $\epsilon$. This is typically established in the \textit{weak coupling regime}, where $\Gamma = \mathcal{O}(\mathrm{poly}(\epsilon))$.
% This leading order reduction requires $\Gamma = O(\mathrm{poly}(\epsilon))$.
% In contrast, previous works require $\Gamma = O(\mathrm{poly}(\epsilon))$, as in~\cite{hahn2025provably,lloyd2025quantumthermal,ding2025endtoendefficientquantumthermal,slezak2026polynomialtime}.
In particular,~\cite{ding2025endtoendefficientquantumthermal} imposes $\Gamma = \mathcal{O}(\epsilon)$ for both thermal and ground state preparation.

% On the other hand, this result can be combined with the perturbative framework for mixing time analysis. 

\subsection{End-to-end Complexity}
% {\it{Application: End-to-end complexity beyond Lindblad dynamics}}--
Building on the state preservation result in~\cref{thm:main_informal_thermal,thm:main_informal_ground}, we develop a novel perturbation framework for analyzing the mixing time of the system-bath interaction algorithm \emph{beyond the Lindblad approximation}.
This yields improved end-to-end complexity bounds for both thermal and ground state preparation. The coupling strength $\Gamma$ remains  independent of $\epsilon$, although it may scale with the system size.
For instance, in~\cref{thm:total_run_time}, we take $\Gamma$ scales as $\mathcal{O}(\mathrm{poly}(N^{-1}))$ to control the higher-order terms.

For the thermal state preparation, existing literature has studied spectral gaps of Lindblad dynamics for a broad class of physical models~\cite{rouz2024,tong2024fast,smid2025polynomial,smid2025rapidmixingquantumgibbs,rouze2024optimal,bergamaschi2026fastmixingquantumspin}, including high‑temperature local spin Hamiltonians~\cite{rouz2024,rouze2024optimal}, weakly interacting fermionic or spin systems at all temperatures~\cite{smid2025rapidmixingquantumgibbs}, and one dimensional local Hamiltonians at all temperatures~\cite{bergamaschi2026fastmixingquantumspin}.
For \emph{all these models}, we summarize their complexity in the following informal theorem.
\begin{thm}[Informal: Total simulation time bound {of thermal state}]\label{thm:total_run_time} In the regime when coupling parameter $\Gamma$ is not depend on the accuracy $\epsilon$, for
all the above models with $N$ qubit
and proper choice of parameter, the system-bath interaction algorithm can prepare an approximation of the thermal state $\sigma_\beta$ with total Hamiltonian simulation time $T_{\rm total}=\widetilde{\mathcal{O}}\left(\mathrm{poly}(N)\epsilon^{-2}\right)$. In particular, $\Gamma = \mathcal{O}(N^{-1/2})$.
\end{thm}
We put the detailed version of \cref{thm:total_run_time}
and its proof in \cref{sec:mixing_time_thermal}. The above theorem is a direct generalization of the end-to-end complexity results in~\cite{slezak2026polynomialtime} beyond the Lindblad dynamics. Because $\Gamma$ is not required to depend on $\epsilon$, the above result saves a factor of $\epsilon^{-2}$ in total runtime compared to~\cite[Theorem 7]{slezak2026polynomialtime}. The dependence of $\Gamma$ on $N$ in \cref{thm:total_run_time}  might be further improved by a refined analysis (see~\cref{re:improvement}).

For ground state preparation, the perturbative framework with finite $\Gamma$ yields generalized end‑to‑end complexity bounds for quadratic fermionic systems~\cite{ding2025endtoendefficientquantumthermal}.
{\begin{thm}[Informal: Total simulation time bound of ground state]\label{thm:ground_run_time}
In the finite coupling regime $\Gamma = {\mathcal{O}}(1)$, for $N$ qubit gapped quadratic fermionic systems, the system-bath interaction algorithm with proper choices of $A_S$ can prepare an $\epsilon$-approximation of the ground state with total simulation time $T_{\rm total} = \widetilde{\mathcal{O}}(\Delta^{-2}N\log(1/\epsilon))$.  In particular, we choose $\sigma = \widetilde{\Theta}(\Delta^{-1}\log(N/\epsilon)),  $ $T =\widetilde{\Theta}(\sigma)$. 
% For ground state preparation of $N$ qubit gapped quadratic fermionic systems, given precision $\epsilon>0$ and proper choice of $A_S$, we can choose $\Gamma = {\mathcal{O}}(1)$, $\sigma = \widetilde{\Theta}(\Delta^{-1}\log(N/\epsilon)),  $ such that
% the mixing time $t_{\rm mix} (\epsilon)= \widetilde{\mathcal{O}}(\|h\|N\log(1/\epsilon))$ and $T_{\rm total} = \widetilde{\mathcal{O}}(\Delta^{-2}N\log(1/\epsilon))$.
\end{thm}}
{This result generalizes the end-to-end complexity result for quadratic fermionic systems in~\cite{ding2025endtoendefficientquantumthermal} to the finite coupling regime, while also improving the dependence on $\epsilon$ from $\mathrm{poly}(1/\epsilon)$ to $\log(1/\epsilon)$.} A detailed version of \cref{thm:ground_run_time}, together with its proof, is given in~\cref{sec:ground_mixing_time_1}. We note that the construction of $A_S$ uses structural information about the system Hamiltonian, which may be difficult to access in practice. To address this issue, in~\cref{sec:ground_mixing_time_2} we present an alternative choice of $A_S$ that avoids using such structural information, but leads to worse dependence on $N$ in $\Gamma$, $t_{\rm mix}$, and $T_{\rm total}$; see \cref{thm:S16} for details.

\section{Methods}
Our analysis {extends the previous work that focuses on the leading order Lindbladian term by providing explicit control of higher order terms.} Our analysis consists of two components: 
\begin{itemize}
    \item Rigorous guarantees that the target thermal or ground state is approximately preserved by the discrete quantum channel $\Phi_\Gamma$ when $\Gamma=\mathcal{O}(1)$;
    \item A perturbative framework for the end-to-end complexity analysis of thermal  and ground state preparation when $\Gamma$ is independent of $\epsilon$.
\end{itemize}

% {We emphasize that  In contrast, the fixed point approximation results in \cref{thm:main_informal_thermal} and \cref{thm:main_informal_ground} hold for $\Gamma=\mathcal{O}(1)$, which can be chosen independently of $N$, with the dependence on $t_{\rm mix}$ left implicit.}

% bound on the total Hamiltonian simulation time for thermal state preparation.
% Both results are build beyond the Lindblad dynamics, i.e., when $\Gamma = \Theta(1)$.
% \zd{To continue.}

First, the discrete quantum channel $\Phi_\Gamma$ induced by the system-bath interaction admits a Dyson series expansion. Inspired by recent progress on error analysis for open quantum systems~\cite{liu2025error,huang2024unified}, the influence of the environment on the system can be fully characterized by multi-point bath correlation functions. This leads to an expansion of $\Phi_\Gamma$ in even powers of the coupling strength, $\Gamma^{2n}$. The leading order term, of order $\mathcal{O}(\Gamma^2)$, gives an effective Lindbladian that approximately preserves the target state, as shown in~\cite{ding2025endtoendefficientquantumthermal}. The main technical challenge is to control the higher-order terms $\mathcal{O}(\Gamma^{2n})$ when $\Gamma=\mathcal{O}(1)$. These higher-order terms involve multidimensional operator Fourier transforms of the form $G^\dagger_{k}(\omega)\rho G_{2n-k}(\omega)$; see~\cref{thm:dyson_Series_phi_alpha} for details. Here $G_{[\cdot]}(\omega)$ denotes a multidimensional operator Fourier transform with $k$ time variables. Our key observation is that \emph{a generalized detailed balance condition} holds for these higher-order terms. More precisely, for any $\omega>0$, we show that conjugating with the Gibbs state $\sigma^{-1}_\beta G_{k}(\omega) \sigma_\beta$ preserves the form of the multidimensional Fourier transforms with shifts and rescaling in the filter function, see~\cref{sec:thermal_state_appendix}~\cref{lm:diff}. The rigorous justification of this approximate detailed balance structure relies on a time domain framework for the detailed balance transformation in~\cref{lm:1}. This framework is designed to handle the multidimensional operator Fourier transform structure and differs substantially from previous Bohr-frequency-domain approximation techniques~\cite{ChenKastoryanoBrandaoEtAl2023,hahn2025provably}. Together with the unitality of the adjoint channel, $\Phi_\Gamma^\dagger(I)=I$, this approximate detailed balance structure shows that $\Phi_\Gamma$ approximately fixes the Gibbs state. This gives the fixed-point approximation result for thermal state preparation in~\cref{thm:main_informal_thermal}. For ground state preservation, we evaluate $G_{k}\ket{\psi_0}$ using the Bohr-frequency-domain approximation technique~\cite{ChenKastoryanoBrandaoEtAl2025,ding2025endtoendefficientquantumthermal}, where $\ket{\psi_0}$ denotes the ground state. Thanks to the Gaussian filter function $f$, we show that $G_{k}\ket{\psi_0}$ lies predominantly in the subspace spanned by $\ket{\psi_0}$, up to an error that is exponentially small in $\sigma^2\Delta^2$. This implies that the higher-order terms in the Dyson expansion also approximately preserve the ground state, with an error exponentially small in $\sigma^2\Delta^2$, as shown in~\cref{thm:main_informal_ground}.

Using the previous state preservation result,~\cref{thm:total_run_time,thm:ground_run_time} establish the end-to-end complexity by constructing a perturbative framework relate the mixing time of the leading Lindblad term with higher order Dyson series in $\Phi_\Gamma$ properly. To the best of our knowledge, existing literature focuses exclusively on the mixing time of Lindbladian or near Lindbladian dynamics; consequently, these techniques are not directly applicable to our setting where $\Gamma$ is independent of $\epsilon$. To bridge this gap, we decompose $\Phi_\Gamma$ into a second-order term and a higher-order remainder as~\cref{eqn:Lindblad_approximation_weak}.
\begin{equation}
\Phi_\Gamma
=
\rho
+
\underbrace{
    \overbrace{\Gamma^2 \mathcal{L}_{\mathrm{Lind}}(\rho)}
    ^{\text{Previous works}}
    +
    \Gamma^4(\cdots)
}_{\text{Our analysis}} .
\end{equation}
We note that the choice of $f$ makes the norm of $\mathcal{L}_{\mathrm{Lind}}$ scale as $\mathcal{O}(1/\sigma)$. Consequently, the contraction rate of the second-order term is of order $\mathcal{O}(\Gamma^2/\sigma)$, which is consistent with the definition of the rescaled mixing time $t_{\rm mix}$. For thermal state preparation, the second-order term corresponds to a {Kubo-Martin-Schwinger (KMS)} detailed balance Lindbladian. Leveraging a recent monotonic spectral gap result from~\cite{slezak2026polynomialtime}, we lower bound this gap by $C/\sigma$, where $C$ is a uniform constant. While the second-order term possesses the desired spectral gap and fixed point, our regime of interest includes higher-order terms that could potentially close the gap and slow down mixing. To address this, we develop a general perturbation framework for quantum channels, inspired by~\cite{ChenKastoryanoBrandaoEtAl2023,ding2025endtoendefficientquantumthermal,slezak2026polynomialtime}, that bounds the spectral gap stability under perturbations. By utilizing the time domain framework developed for our fixed point analysis, we directly control the higher-order terms within the KMS inner product structure. Specifically, we demonstrate that if the second order truncation has a spectral gap $\lambda_{\rm gap}$, choosing $\Gamma^2 = \mathcal{O} (\lambda_{\rm gap})$ suffices to preserve the gap up to a constant factor. This analysis establishes the mixing time bound in Theorem~\ref{thm:total_run_time}. {The $N$ dependence in $\lambda_{\rm gap}$ leads to the $N$ dependence in $\Gamma$ in~\cref{thm:total_run_time}. For ground state preparation, we focus on gapped quadratic fermionic systems and provide an generalization of the result in~\cite{ding2025endtoendefficientquantumthermal}. Follow the idea in~\cite{zhan2025rapidquantumgroundstate}, it is enough to control  $\Tr(\rho \Phi_\Gamma^\dagger(\hat{N}))$ where $\hat{N}$ measures the number of excitations above the ground state.   A key step is to bound the multidimensional Fourier transforms generated by the filter function. This allows us to improve the accuracy dependence from $\mathcal{O}(\poly(1/\epsilon))$ to $\mathcal{O}(\log(1/\epsilon))$.}

\section{Numerical results}\label{sec:numerics}
% {\it Numerical results} --
\begin{figure*}[t]
  \centering
  \begin{subfigure}[b]{0.32\textwidth}
    \centering
    \includegraphics[width=\linewidth]{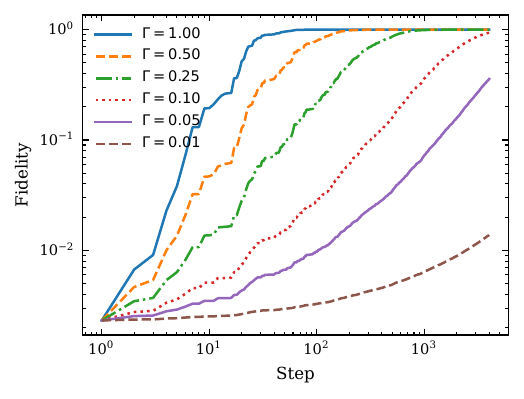}
    \caption{$\Gamma = \Theta(1)$}
    \label{Fig:TFIM_4_thermal_main}
  \end{subfigure}\hfill
  \begin{subfigure}[b]{0.32\textwidth}
    \centering
    \includegraphics[width=\linewidth]{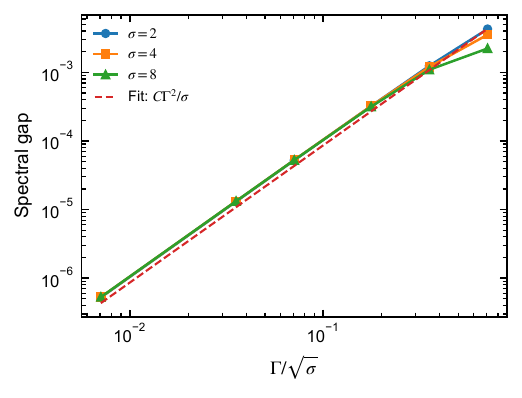}
    \caption{$\Gamma =  \Theta(1)$}
    \label{fig3:spectral_gap}
  \end{subfigure}\hfill
  \begin{subfigure}[b]{0.32\textwidth}
    \centering
    \includegraphics[width=\linewidth]{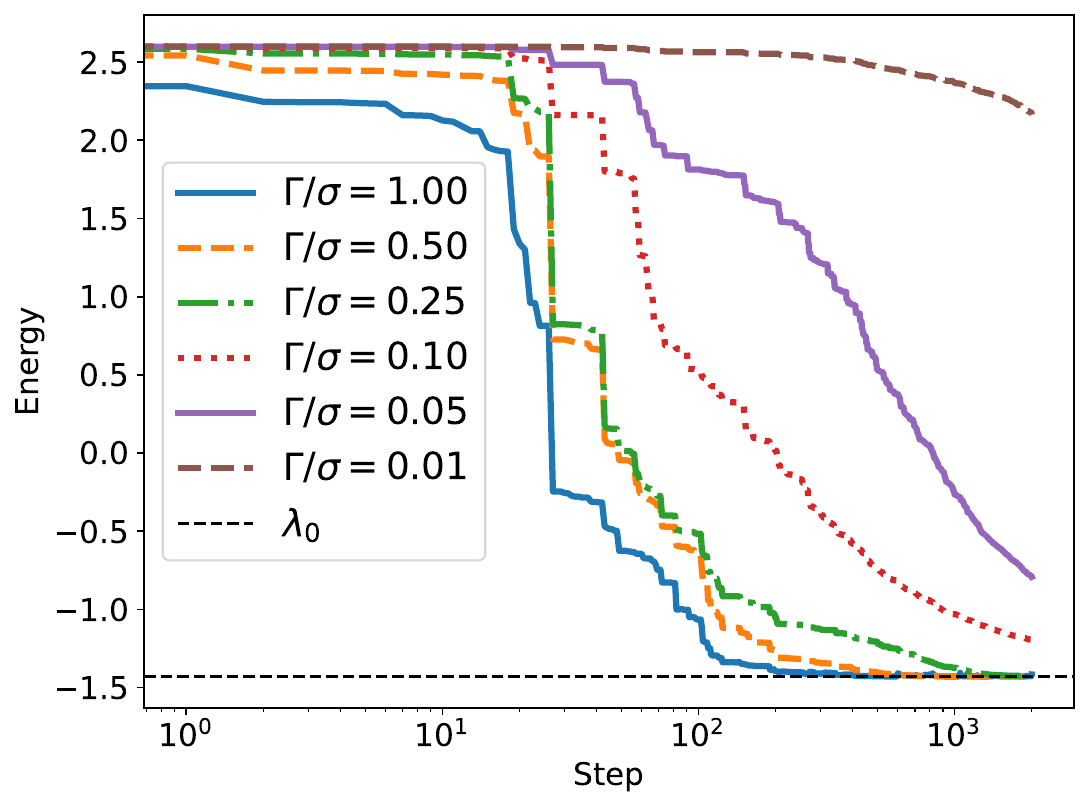}
    \caption{$\Gamma =  \Theta(\sigma)$}
    \label{fig:ANNNI_energy}
  \end{subfigure}
  \caption{ Fig(a)-(b): Thermal State Preparation in the finite coupling $\Gamma=\mathcal{O}(1)$ regime with frequency $\omega$ is sampled from Gaussian distribution $g(\omega) = \frac{\beta}{\sqrt{2\pi(2-\beta^2/(4\sigma^2))}}\exp\left(-\frac{(\beta\omega+1)^2}{2(2-\beta^2/(4\sigma^2))}\right)$: (a) Fidelity of TFIM with $L=4$ sites with $\beta=1$, $\sigma=2$, and $T=5\sigma$; (b) Spectral gap of $\Phi_{\Gamma}$ for the Hubbard model with $L = 2$ sites.
  % (a) Fidelity of thermal state preparation for TFIM with $L=4$ sites in the $\Gamma=\Theta(1)$ regime. Parameters are $\beta=1$, $\sigma=2$, and $T=5\sigma$. The .
  % (b) Spectral gap of $\Phi_{\Gamma}$ of thermal state preparation for the Hubbard model with $L = 2$ sites in the $\Gamma=\Theta(1)$ regime.
  The spectral gap scales as ${\Gamma^{2}/\sigma}$ and is independent of $\sigma$.
  (c) Evolution of energy for ground state preparation of the ANNNI model in the $\Gamma=\Theta(\sigma)$ regime
  , $\lambda_0$ is the ground state energy. The frequency $\omega$ is sampled uniformly from the interval $[0,5]$. }
  \label{fig:ABC}
\end{figure*}
We numerically implement the system bath interaction algorithm to validate our analysis results in the constant coupling regime $\Gamma =\mathcal{O}(1)$ and to probe its behavior in the strong coupling regime $\Gamma = \mathcal{O}(\sigma)$. As we discussed before, we expect the rescaled mixing time $t_{\rm mix}:=\Gamma^2\tau_{\rm mix}/\sigma$ to remain constant as $\sigma$ increases.
% We define $\alpha=\Gamma/\sqrt{\sigma}$ and characterize the mixing time $\tau_{\rm mix}$ as a function of $\alpha$ to verify the predicted $\mathcal{O}(\alpha^{-2})$ scaling.

We first consider the transverse field Ising model (TFIM) with up to $L = 8$ sites:
\begin{equation}\label{eqn:TFIM}
    H = -J\sum_{i=1}^{L-1} Z_i Z_{i+1} - g \sum_{i = 1}^L X_i\,,
\end{equation}
with $ J = 1$, $g = 1.2$.
We sample $A_S$ uniformly chosen from the single qubit Pauli set, with positive and negative signs assigned with equal probability. Starting from the initial state, we simulate the dynamics by iteratively applying the  discrete quantum channel $\Phi_\Gamma$ and track the evolution via the fidelity $F(\rho, \sigma_\beta) = \left(\mathrm{Tr}\sqrt{\sqrt{\sigma_\beta}\rho\sqrt{\sigma_\beta}}\right)^2$
recorded as a function of the number of channel applications along a single trajectory.
~\cref{Fig:TFIM_4_thermal_main} shows that, upon fixing the filter width $\sigma = 1$ and varying the coupling strength $\Gamma = \mathcal{O}(1)$, the iterates converge to $\sigma_\beta$, with a faster convergence for larger $\Gamma$. This observation is in quantitative agreement with the scaling predicted in~\cref{thm:main_informal_thermal}.

Next we consider the 1-D Hubbard model defined with up to $L=4$ spinful sites with open boundary conditions,
\begin{equation}\label{eqn:Hubbard}
    H = -t\sum_{j=1}^{L-1} \sum_{\sigma\in\{\uparrow,\downarrow\}} c_{j,\sigma}^\dag c_{j+1,\sigma}+U\sum_{j=1}^L (n_{j,\uparrow}-\frac{1}{2})(n_{j,\downarrow}-\frac{1}{2})\,.
\end{equation}
where the number operator $n_{j,\sigma}= c_{j,\sigma}^\dag c_{j,\sigma}$ and the dimension is $2^{2L}$. We choose $t = 1, U =-4$, and sample $A_S$ uniformly from $\mathcal{A} = \left\{\pm c_{j,\sigma}, \pm c_{j,\sigma}^\dag\right\}_{j=1,\cdots, L, \sigma = \uparrow,\downarrow}$.
We numerically calculate the quantum channel $\Phi_{\Gamma}$ as a superoperator and compute its fixed point and spectral gap. This enables us to quantify the accuracy of the fixed point in approximating the target state and to assess the spectral gap of the channel under various parameter settings, which directly characterize the accuracy of the fixed point and the mixing time. \cref{fig3:spectral_gap} shows that the spectral gap remains independent of $\sigma$ and scales quadratically as $\Gamma^2/\sigma$. The mixing time follows $\mathcal{O}(\sigma/\Gamma^2)$ scaling over the tested range, with no sensitive dependence on $\sigma$.

Finally, we consider the ground state preparation for the one-dimensional axial next-nearest-neighbor Ising (ANNNI) model~\cite{Selke1988} in the \textit{strong coupling} regime $\Gamma =\mathcal{O}(\sigma)$. In~\cite{zhan2025rapidquantumgroundstate}, it was shown that preparing the ground state of $H_{\mathrm{ANNNI}}$ via adiabatic evolution is challenging, as the effective spectral gap closes multiple times along the adiabatic path if the initial Hamiltonian is set to be $H_0=-Z$. In our work, we evaluate the performance of our ground state preparation algorithm on this model with the same choice of parameters as that in the case of TFIM or Hubbard models.  We continue to observe clear convergence of the energy toward the ground state energy versus the number of iterative steps, which is consistent with the state preservation result in~\cref{thm:main_informal_ground}. Moreover, the flexibility to choose $\Gamma = \mathcal{O}(\sigma)$ in ~\cref{fig:ANNNI_energy}  enables the use of larger couplings, leading to a larger spectral gap and faster mixing. Developing a corresponding theoretical understanding in this strong coupling regime remains an interesting direction for future work.

We note that the strong coupling regime $\Gamma = \mathcal{O}(\sigma)$ is not a special case of ANNNI model. In~\cref{sec:numerics_appendix}, we conduct more detailed numerical experiments for all three models in both constant coupling regime $\Gamma = \mathcal{O}(1)$ and the strong coupling regime $\Gamma = \mathcal{O}(\sigma)$.

\section{Conclusion and outlook}
In this work, we theoretically demonstrated that system-bath interaction algorithms can prepare thermal and ground states {beyond the Lindblad approximation regime.} We prove that, even when the interaction strength per iteration is chosen independently of the target accuracy, the fixed point of the induced quantum channel can still closely approximate the target state. Building on this preservation result, we improve previous end-to-end complexity for dissipative state preparation algorithms.
% {Since the algorithm requires only the simulation of the underlying system-bath dynamics, with the cost comparable to the Hamiltonian simulation and without postselection, it provides a natural framework for dissipative quantum algorithms.}
{Moreover, the techniques developed in this work for analyzing system-bath interaction models beyond the Lindblad approximation may have broader applications in dissipative state preparation algorithm design, and
in the study of open quantum systems, particularly in regimes where the coupling strength is not necessarily weak. This perspective may help guide algorithm design utilizing features arising from the higher order Dyson series. }
% Since open system dynamics also provide a realistic model of hardware noise~\cite{van2023probabilistic,wang2025non}, intrinsic hardware noise may be viewed as from an effective system-bath coupling. This connection enables both the analysis of quantum noise and the design of robust quantum algorithms for near-term quantum devices.
% ~\zd{We need to remove the noise discussion. Do we understand what this means concretely?}

Beyond our theoretical results, several promising directions for future research remain. First, our numerical experiments suggest that system-bath interaction models may perform even better in practice than our current bounds indicate, particularly in the stronger coupling regime where $\Gamma=\mathcal{O}(\sigma)$. Developing a theoretical understanding of this regime would further clarify the robustness and efficiency of these algorithms.
{Second, }the end-to-end complexity analysis for thermal state preparation in~\cref{thm:total_run_time} connects the previous Lindbladian mixing analysis based on the spectral gap with the microscopic system-bath interaction model. The current result requires the coupling strength $\Gamma$ to depend on the system size $N$. This dependence may be removed through a more refined analysis of higher-order terms in the Dyson series expansion, potentially by exploiting their quasi-local structure. The dependence on the initial condition, inherited from spectral gap based mixing bounds, may also be reduced using tools such as modified Logarithmic Sobolev inequalities or the oscillator norm. 
% which is likely can be improved to be independent of the system size with more careful analysis or quasi-local structure of the higher order Dyson series. Furthermore, the dependence on the initial condition, which is inherented from the Lindbladian technique based on spectral gap can potentially be removed using other techniques like Modified Logarithmic Sobolev inequalities or the oscillator norm. 
% is under the coupling strength that is of the same scale as the spectral gap of the leading Lindblad term, 
% is likely not tight, and sharpening it is an interesting direction for future work.
% (see the detailed discussion in~\cref{sec:application})
Third, our work varies the coupling parameter $\Gamma$ while keeping $\sigma$ large in order to ensure an accurate fixed-point approximation. A very recent work~\cite{chen2026overcominglambshiftsystembath} shows that, in the weak coupling regime $\Gamma\to 0$, one can choose $\sigma=\mathcal{O}(1)$ while still maintaining an accurate fixed-point approximation. This provides a complementary result to our work. An interesting direction for future investigation is whether the advantages of these two regimes can be combined to further improve the overall performance.

\vspace{1em}
\noindent\textbf{\large Acknowledgments\\}
This work was supported in part by the University of Michigan through a startup grant of Z.D. (Z.D., K.W.) and the Van Loo Postdoctoral Fellowship (K.W.). The authors thank Yongtao Zhan, Lin Lin, Daniel Stilck França, Chi-Fang Chen, Andre\'as Gily\'en for helpful discussions.

\vspace{1em}
\noindent\textbf{\large Data availability\\}
The code for the numerical experiments is openly available~\cite{wangding2025github}.
\bibliographystyle{apsrev4-2}
\bibliography{ref.bib}

\input{npj_supp}

\end{document}

%% file: preamble.tex
\usepackage{color, xcolor, colortbl}
\usepackage{graphicx}
\usepackage{amsthm,amsmath,amssymb,amsfonts}
\usepackage{dcolumn}% Align table columns on decimal point
\usepackage{url}
\usepackage{epstopdf}
\usepackage{enumitem}
\usepackage{algpseudocode}
\usepackage{bm}
\usepackage{appendix}
\usepackage{multirow}
\usepackage{braket}
\usepackage[english]{babel}
\usepackage{aliascnt}
\usepackage{hyperref}
\usepackage[capitalize]{cleveref}
\usepackage{algorithm}
\usepackage{xpatch}
\usepackage{tikz}
\usepackage{adjustbox}
\usepackage{xspace}
\usetikzlibrary{quantikz}
\usepackage{pifont}

\newcommand{\Or}{\mathcal{O}}

\newcommand{\diag}{\operatorname{diag}}

\newcommand{\poly}{\operatorname{poly}}
\newcommand{\Tr}{\operatorname{Tr}}

\DeclareMathOperator{\ad}{ad}

\newcommand{\mc}[1]{\mathcal{#1}}

\newcommand{\wt}[1]{\widetilde{#1}}

\newtheorem{thm}{Theorem}%[section]
\crefname{thm}{Theorem}{Theorems}
\Crefname{thm}{Theorem}{Theorems}

\newaliascnt{lem}{thm}
\newaliascnt{rem}{thm}
\newaliascnt{prop}{thm}
\newaliascnt{cor}{thm}
\newaliascnt{assumption}{thm}
\newaliascnt{defn}{thm}
\newaliascnt{notation}{thm}
\newaliascnt{fact}{thm}

\newtheorem{lem}[lem]{Lemma}
\newtheorem{rem}[rem]{Remark}

\newtheorem{cor}[cor]{Corollary}

\aliascntresetthe{lem}
\aliascntresetthe{rem}
\aliascntresetthe{prop}
\aliascntresetthe{cor}
\aliascntresetthe{assumption}
\aliascntresetthe{defn}
\aliascntresetthe{notation}
\aliascntresetthe{fact}

\crefname{lem}{Lemma}{Lemmas}
\Crefname{lem}{Lemma}{Lemmas}

\crefname{rem}{Remark}{Remarks}
\Crefname{rem}{Remark}{Remarks}

\crefname{prop}{Proposition}{Propositions}
\Crefname{prop}{Proposition}{Propositions}

\crefname{cor}{Corollary}{Corollaries}
\Crefname{cor}{Corollary}{Corollaries}

\crefname{assumption}{Assumption}{Assumptions}
\Crefname{assumption}{Assumption}{Assumptions}

\crefname{defn}{Definition}{Definitions}
\Crefname{defn}{Definition}{Definitions}

\crefname{notation}{Notation}{Notations}
\Crefname{notation}{Notation}{Notations}

\crefname{fact}{Fact}{Facts}
\Crefname{fact}{Fact}{Facts}

\usepackage{bbm}

\NewDocumentCommand{\ketbra}{mG{#1}}{\mathinner{|{#1}\rangle\!\langle{#2}|}}

\usepackage{silence}
\usetikzlibrary{fit}
\tikzset{%
  highlight/.style={rectangle,rounded corners,fill=blue!15,draw,fill opacity=0.3,thick,inner sep=0pt}
}

%% file: npj_supp.tex
\newpage
\clearpage
% \appendix
\thispagestyle{empty}
\onecolumngrid

\begin{center}
\textbf{\large Supplemental materials for \\
"Beyond Lindblad Dynamics: Rigorous Guarantees for Thermal and Ground State
Preservation under System-Bath Interactions"\\
}
\smallskip
Ke Wang$^1$, Zhiyan Ding$^1$\\
\smallskip
\small{\emph{$^1$ Department of Mathematics, University of Michigan, Ann Arbor, MI 48109, USA\\}}
(Dated: \today)
\end{center}
\setcounter{figure}{0}
\setcounter{table}{0}
\setcounter{page}{1}
\definecolor{RED}{RGB}{255,0,0}
\makeatletter
\renewcommand{\theequation}{S\arabic{equation}}
\renewcommand{\thefigure}{S\arabic{figure}}
\renewcommand{\thetable}{S\arabic{table}}
\renewcommand{\bibnumfmt}[1]{[S#1]}

\renewcommand{\thecor}{S\arabic{cor}}
% Prefix theorem-like counters with "S" in the sudpplement
\setcounter{thm}{0}
\setcounter{lem}{0}

\setcounter{defn}{0}
\renewcommand{\thethm}{S\arabic{thm}}
\renewcommand{\thelem}{S\arabic{lem}}
\renewcommand{\thedefn}{S\arabic{defn}}

\renewcommand{\theassumption}{S\arabic{assumption}}
% \tableofcontents

% enable numbered (prefixed "S") sections in the supplement
\setcounter{secnumdepth}{3}
\setcounter{section}{0}
\renewcommand{\thesection}{S\arabic{section}}
\renewcommand{\thesubsection}{S\arabic{section}.\arabic{subsection}}
\renewcommand{\thesubsubsection}{S\arabic{section}.\arabic{subsection}.\arabic{subsubsection}}
The supplemental materials is organized as follows. After introducing the notation and basic preliminaries in~\cref{sec:preliminaries}, we derive in~\cref{sec:dyson_expansion} the full Dyson series expansion of the induced quantum channel generated by the system-bath interaction algorithm. The proof of the thermal and ground state preservation in~\cref{thm:main_informal_thermal} and~\cref{thm:main_informal_ground}, are given in
\cref{sec:thermal_state_appendix} and \cref{sec:ground_state_appendix}.
\cref{sec:mixing_time_thermal} contains the perturbative mixing time analysis for thermal state preparation and completes the proof of \cref{thm:total_run_time}. \cref{sec:mixing_time_ground} proves the end-to-end complexity of the ground state preparation beyond the Lindblad approximation in~\cref{thm:ground_run_time}.
Additional numerical results are deferred to \cref{sec:numerics_appendix}.

\section{Notations and preliminaries}\label{sec:preliminaries}

We use $A^*$, $A^T$, and $A^\dagger$ to denote the complex conjugate, transpose, and Hermitian adjoint of a matrix $A \in \mathbb{C}^{N \times N}$, respectively.
The Schatten $p$-norm is denoted by $\|A\|_p := (\sum_i \sigma_i(A)^p)^{1/p}$, where $\sigma_i(A)$ are the singular values of $A$. This encompasses the trace norm ($\|A\|_1$), the Hilbert–Schmidt/Frobenius norm ($\|A\|_2$), and the operator norm ($\|A\|_\infty \equiv \|A\|$). For a superoperator $\Phi: \mathbb{C}^{N \times N} \to \mathbb{C}^{N \times N}$, the induced trace norm and Frobenius norm  are defined as
$$\|\Phi\|_{1\to 1} := \sup_{\|A\|_1=1} \|\Phi(A)\|_1 ,\quad  {\|\Phi\|_{2\to 2} := \sup_{\|A\|_2=1} \|\Phi(A)\|_2 }\,.$$
We use the standard asymptotic notations. We write $f=\Omega(g)$ if $g=\Or(f)$. The notations $\wt{\Or}$, $\wt{\Omega}$, $\wt{\Theta}$ are used to suppress subdominant polylogarithmic factors. If not specified, $f = \wt{\Or}(g)$ if $f = \Or(g\operatorname{polylog}(g))$; $f = \wt{\Omega}(g)$ if $f = \Omega(g\operatorname{polylog}(g))$; $f = \wt{\Theta}(g)$ if $f = \Theta(g\operatorname{polylog}(g))$.

Let the eigendecomposition of the Hamiltonian as $H=\sum_i \lambda_i \ket{\psi_i}\bra{\psi_i}$. The Bohr frequencies of $H$ are the energy differences $\lambda_i-\lambda_j$, denoted as $B(H)$.
For $\nu\in B(H)$, the component of the matrix $A$ at Bohr frequency $\nu$ is
\begin{equation}\label{eq:Anu}
A(\nu)=\sum_{\lambda_j-\lambda_i=\nu}\ket{\psi_j}\bra{\psi_j}A\ket{\psi_i}\bra{\psi_i}.
\end{equation}

Let $\sigma_\beta = \exp(-\beta H)/\Tr(\exp(-\beta H))$ denote the thermal state at the inverse temperature $\beta$. We use  the Kubo–Martin–Schwinger~(KMS) inner product
\[
\langle A, B \rangle_{1/2, \sigma_\beta} =\langle A, \sigma_\beta^{1/2} B \sigma_\beta^{1/2}  \rangle =  \mathrm{Tr}\left( A^\dagger \sigma_\beta^{1/2} B \sigma_\beta^{1/2} \right) \,.
\]
% \emph{Detailed balance condition and mixing time:}  We  the Kubo–Martin–Schwinger~(KMS) inner product on the space of operators as
A Lindbladian $\mathcal{L}$ satisfies the Kubo–Martin–Schwinger detailed balance condition if its adjoint $\mathcal{L}^\dagger$ is self-adjoint under this inner product.
In this case, the spectral gap is
\[
\lambda_{\rm gap}(\mathcal{L}) = \inf_{\substack{A \neq 0 \\ \mathrm{Tr}(A\sigma_\beta)=0}} \frac{-\langle A, \mathcal{L}^\dagger(A) \rangle_{1/2, \sigma_\beta}}{\langle A, A \rangle_{1/2, \sigma_\beta}} \,.
\]

Our main object is the discrete channel $\Phi_{\Gamma}$. For the quantum channel $\Phi_\Gamma$ with a unique fixed point $\rho_{\rm fix}(\Phi_\Gamma)$,
we define the integer mixing time of $\Phi_{\Gamma}$
\begin{equation}\label{def:mixing_time}
\tau_{{\rm mix}, \Phi_\Gamma}(\epsilon) = \min \left\{ t \in \mathbb{N} \middle| \sup_{\rho} \| \Phi_\Gamma^t(\rho) - \rho_{\rm fix}(\Phi_\Gamma) \|_1 \leq \epsilon \right\} \,.
\end{equation}
% . It quantifies the worst-case convergence speed, i.e. the minimum number of iterations required for an arbitrary initial state to become $\epsilon$-close to the target fixed point.
% \begin{defn}\label{def:mixing_time}
% Let $\Phi$ be a CPTP map with a unique fixed point $\rho_{\rm fix}(\Phi)$. For any $\epsilon > 0$, the integer mixing time $\tau_{{\rm mix}, \Phi}(\epsilon)$ is defined as
% \end{defn}
According to~\cite{ding2025endtoendefficientquantumthermal}, we define the rescaled mixing time as
\begin{equation}\label{eqn:rescaled_mixing_time}
t_{{\rm mix}, \Phi_\Gamma}(\epsilon) = {\frac{\Gamma^2}{\sigma}} \tau_{{\rm mix}, \Phi_\Gamma}(\epsilon)\,.
\end{equation}

\section{Dyson series expansion of \texorpdfstring{$\Phi_\Gamma$}{Phi alpha}}\label{sec:dyson_expansion}

In this section, we introduce the Dyson series expansion of $\Phi_\Gamma$ in~\cref{eqn:Phi_alpha}. We use the subindex $\{-1, 1\}$ to relabel the system and bath operator as
\begin{equation*}
    S_1 = A_S,\quad S_{-1} = A_S^\dag,\quad B_1 = B_E,\quad B_{-1} = B_E^\dag.
\end{equation*}
Define the system evolution operator $U_S(t)=\exp(-itH)$. The Dyson series expansion of $\Phi_\Gamma$ is summarized in the following theorem:
\begin{thm}\label{thm:dyson_Series_phi_alpha} Let $\rho_{n+1}=\Phi_\Gamma(\rho_n)$, then $\rho_{n+1}$ can be generated by the following three steps:
\begin{itemize}
\item
\[
\rho_{n+1/3}=U_S(T)\rho_n U_S^\dagger(T)
\]
\item \begin{equation}\label{eqn:middle_evolution}
    \begin{aligned}
        \rho_{n+2/3} &= \rho_{n+1/3}  + \mathbb{E}_{A_S}\left(\sum_n \Gamma^{2n} (-1)^n \sum_{k = 0}^{2n}(-1)^k \int \frac{g(\omega)}{1+ e^{\beta\omega}}{G}_{2n-k,A_S}^\dag(\omega) \rho_{n+1/3} {G}_{k,A_S}(\omega) \mathrm{d}\omega\right.\\
        & + \left.\sum_n \Gamma^{2n} (-1)^n \sum_{k = 0}^{2n}(-1)^k \int \frac{g(\omega)e^{\beta\omega}}{1+ e^{\beta\omega}}{F}_{2n-k,A_S}^\dag(\omega) \rho_{n+1/3}  {F}_{k,A_S}(\omega)\mathrm{d}\omega\right)=\Phi_{\rm part,\Gamma}(\rho_{n+1/3})\,;\\
    \end{aligned}
\end{equation}
\item
\[
\rho_{n+1}=U_S(T)\rho_{n+2/3} U_S^\dagger(T)\,.
\]
\end{itemize}
Here
\begin{equation}\label{eq:GF}
    \begin{aligned}
        {G}_{k,A_S}(\omega) & =\int_{-T<t_1\leq \cdots \leq t_k < T} A_S(t_1)A^\dagger_S(t_2)\cdots S_{(-1)^{k-1}}(t_k)e^{-i\omega\sum_{p=1}^k (-1)^p t_p} f(t_1)\cdots f(t_k)\mathrm{d}t_1 \cdots\mathrm{d}t_k \\
        {F}_{k,A_S}(\omega) & = \int_{-T<t_1\leq \cdots \leq t_k < T} A_S^\dag(t_1)A_S(t_2)\cdots S_{(-1)^{k}}(t_k)e^{i\omega\sum_{p=1}^k (-1)^p t_p} f(t_1)\cdots f(t_k)\mathrm{d}t_1 \cdots\mathrm{d}t_k \,.\\
    \end{aligned}
\end{equation}
Define $\gamma(\omega)=\frac{g(\omega)+g(-\omega)}{1+\exp(\beta\omega)}$. In the case when $A_S$ is uniformly sampled from a set of coupling operators $\mc{A}=\{A^i,-A^i\}_i$ with $\{(A^i)^\dagger\}_i=\{A^i\}_i$, we can rewrite~\eqref{eqn:middle_evolution} as
\begin{equation}\label{eqn:rho_n_1/3_update}
    \begin{aligned}
        \rho_{n+2/3} &= \rho_{n+1/3}  + \mathbb{E}_{A_S}\left(\sum_n \Gamma^{2n} (-1)^n \sum_{k = 0}^{2n}(-1)^k \int \gamma(\omega){G}_{2n-k,A_S}^\dag(\omega) \rho_{n+1/3}  {G}_{k,A_S}(\omega)\mathrm{d}\omega\right)\,.\\
    \end{aligned}
\end{equation}
\end{thm}

\begin{proof}
For simplicity, we fixed $A_S$. We use the subindex $\{-1, 1\}$ to relabel the system and bath operator as
\begin{equation*}
    S_1 = A_S, S_{-1} = A_S^\dag, B_1 = B_E, B_{-1} = B_E^\dag.
\end{equation*}
Recall~\cref{eqn:Ham_total}:
\[
H_{\Gamma}(t) = H + H_E + \Gamma f(t) \left( A_S \otimes B_E + A_S^\dagger \otimes B_E^\dagger \right),\quad H_E=-\omega Z/2,\quad B_E=\ket{1}\bra{0}\,.
\]
For the dynamics described by
\begin{equation}
    \begin{cases}
        \partial_t\rho(t) = -i[H_{\Gamma}(t), \rho(t)],\\
        \rho(-T) = \rho_n\otimes \rho_E,\\
        \rho_{n+1} = \mathbb{E}_{H_E, A_S} (\mathrm{Tr}_E \rho(T)),
    \end{cases}
\end{equation}
 the evolution operator of the time-dependent Hamiltonian can be expressed using the time-ordered exponential as $U^{\Gamma}(t) \coloneqq U^{\Gamma}(t; - T) = \mathcal{T}\left(\exp\left(-i\int_{-T}^t H_{\Gamma}(s)\mathrm{d}s\right)\right)$, and it satisfies
\begin{equation}
    \partial_t U^{\Gamma}(t; - T) = -i H_{\Gamma}(t) U^{\Gamma}(t; - T)\,.
\end{equation}
By Duhamel's expression, the Dyson series expansion can be written as $U^{\Gamma}(t; -T) = \sum_{n\geq 0} (-i\Gamma)^n U_n(t; - T)$, where
\begin{equation}
     U_n(t; - T) = \int_{-T}^t \int_{-T}^{s_1}\cdots \int_{-T}^{s_{n-1}} f(s_1)f(s_2)\cdots f(s_n) U_0(t; s_1) H_{S,B} U_0(s_1; s_2) \cdots H_{S,B} U_0(s_n;-T)\mathrm{d}s_n \cdots \mathrm{d}s_1\,.
\end{equation}
Here $H_{S,B} = S_1 \otimes B_1 + S_{-1}\otimes B_{-1}$ and $U_0(t;-T) = \exp(-i(t+T)(H+H_E))$. The Heisenberg
picture evolution of the bath operator follows $B_{(-1)^p}(t) = e^{(-1)^pit\omega}B_{(-1)^p}$.

After taking expectation in $\omega$ and tracing out the degree of freedom in the environment,
\begin{equation}\label{eq:Phia0}
    \begin{aligned}
        \Phi_\Gamma\rho &=\int g(\omega)\mathrm{Tr}_E[U^{\Gamma}(T)(\rho\otimes \rho_E)U^{\Gamma\dag}(T)](\omega)\mathrm{d}\omega \\
        & = \sum_{n,m\geq 0}(-i\Gamma)^n (i\Gamma)^m\int g(\omega)\mathrm{Tr}_E[U_n(T;-T)(\rho\otimes \rho_E)U^{\dag}_m(T;-T)](\omega)\mathrm{d}\omega\\
        & = \sum_{n,m\geq 0}(-1)^n(i\Gamma)^{n+m}\int \mathrm{d}\omega \int_{-T\leq s_n\leq \cdots \leq s_1 \leq T, -T\leq \tau_m\leq \cdots\leq \tau_1 \leq T}\mathrm{d}s_1 \cdots \mathrm{d}s_n \mathrm{d}\tau_1\cdots \mathrm{d}\tau_m g(\omega)\\
        &\sum_{\alpha_i,\beta_j\in\{-1,1\}}U_S(T)S_{\alpha_1}(s_1)\cdots S_{\alpha_n}(s_n)f(s_1)\cdots f(s_n) U_S(T)\rho U_S^\dag(T) S_{\beta_m}^\dag(\tau_m) \cdots S_{\beta_1}^\dag (\tau_1)U_S^\dag(T)\\
        & f(\tau_1)\cdots f(\tau_m)\mathrm{Tr}_E(B_{\alpha_1}(s_1)\cdots B_{\alpha_n}(s_n)\rho_E  B_{\beta_m}^\dag(\tau_m) \cdots B_{\beta_1}^\dag (\tau_1))\,, \\
    \end{aligned}
\end{equation}
where $U_S(T) = \exp(-iTH)$ and $S_{\alpha}(t)=\exp(iHt)S_{\alpha}\exp(-iHt)$. Notice that $B_\alpha^\dag(t) = B_{-\alpha}(t)$,
the multi-point bath correlation function follows
\[
\begin{aligned}
&\mathrm{Tr}\big[
  B_{\alpha_1}(s_1)\cdots B_{\alpha_n}(s_n)\,\rho_E\,
  B_{\beta_m}^\dag(\tau_m)\cdots B_{\beta_1}^\dag(\tau_1)
\big] \\[2pt]
&= \mathrm{Tr}\big[
  B_{-\beta_m}(\tau_m)\cdots B_{-\beta_1}(\tau_1)\,
  B_{\alpha_1}(s_1)\cdots B_{\alpha_n}(s_n)\,\rho_E
\big] \\[2pt]
&= e^{\,i\phi}\,
\mathrm{Tr}\big[
  B_{-\beta_m}\cdots B_{-\beta_1}\,B_{\alpha_1}\cdots B_{\alpha_n}\,\rho_E
\big],
\end{aligned}
\]
where \(\displaystyle \phi=\omega\!\left(\sum_{j=1}^n \alpha_j s_j-\sum_{k=1}^m \beta_k \tau_k\right).\) It can be rewritten as
\begin{equation}\label{eq:npointBCFs}
   (\cdots)
=
e^{\,i\phi}\times
\begin{cases}
\mathrm{Tr}(\ketbra{1}\rho_E)
= \dfrac{1}{e^{\beta\omega}+1},
& \!\!\!\!\text{if } m{+}n \text{ even and }
(-\beta_m,\ldots,-\beta_1,\alpha_1,\ldots,\alpha_n)=(1,-1)^{\otimes \frac{m+n}{2}},\\[10pt]
\mathrm{Tr}(\ketbra{0}\rho_E)
= \dfrac{e^{\beta\omega}}{e^{\beta\omega}+1},
& \!\!\!\!\text{if } m{+}n \text{ even and }
(-\beta_m,\ldots,-\beta_1,\alpha_1,\ldots,\alpha_n)=(-1,1)^{\otimes \frac{m+n}{2}},\\[10pt]
0, & \!\!\!\!\text{otherwise.}
\end{cases}
\end{equation}

Substituting the expression of the bath correlation functions back to \cref{eq:Phia0}, we have
\begin{equation}
    \begin{aligned}
        \Phi_\Gamma \rho_n &= U_S(2T)\rho_n U_S(2T)^\dag \\
        & + \sum_{n\geq 1}\Gamma^{2n} (-1)^n \sum_{k = 0}^{2n}(-1)^k \int \frac{g(\omega)}{1+ e^{\beta\omega}}U_S(T){G}_{2n-k,A_S}^\dag(\omega) U_S(T)\rho_n U_S^\dag(T){G}_{k,A_S}(\omega) U_S^\dag(T)\mathrm{d}\omega\\
        & + \sum_{n\geq 1} \Gamma^{2n}(-1)^n \sum_{k = 0}^{2n}(-1)^k \int \frac{g(\omega)e^{\beta\omega}}{1+ e^{\beta\omega}}U_S(T)F_{2n-k,A_S}^\dag(\omega) U_S(T)\rho_n U_S^\dag(T) {F}_{k,A_S}(\omega)U_S^\dag(T) \mathrm{d}\omega\,,\\
    \end{aligned}
\end{equation}
where ${G}_{k,A_S}(\omega)$ and ${F}_{k,A_S}(\omega)$ are defined in~\cref{eq:GF}.
Additionally, using the relation $F_{k, A_S^\dag}(-\omega) = G_{k, A_S}(\omega)$, we can further simplify the expression to~\cref{eqn:rho_n_1/3_update} in the case $\mathcal{A}^\dag = \mathcal{A}$.
This concludes the proof.
\end{proof}

\section{Thermal state preparation beyond Lindblad limit}\label{sec:thermal_state_appendix}
In this section, we introduce the rigorous version of~\cref{thm:main_informal_thermal} and provide the proof. First, recall~\cref{thm:dyson_Series_phi_alpha}~\eqref{eqn:rho_n_1/3_update},
\begin{itemize}
\item $\rho_{n+1/3}=U_S(T)\rho_n U_S^\dagger(T);$
\item $
 \rho_{n+2/3} = \rho_{n+1/3}  + \mathbb{E}_{A_S}\left(\sum_n \Gamma^{2n} (-1)^n \sum_{k = 0}^{2n}(-1)^k \int \gamma(\omega){G}_{2n-k,A_S}^\dag(\omega) \rho_{n+1/3}  {G}_{k,A_S}(\omega)\mathrm{d}\omega\right);
$
\item
$
\rho_{n+1}=U_S(T)\rho_{n+2/3} U_S^\dagger(T)\,.
$
\end{itemize}
We present the rigorous version of~\cref{thm:main_informal_thermal} in the following.
\begin{thm}[Thermal state, rigorous version]\label{thm:main_rigor_thermal} Assume $0\leq \beta<\infty$ and $\gamma(\omega)$ satisfies the property:
\begin{itemize}
    \item decay fast at infinity, i.e. $\lim_{|\omega|\rightarrow\infty}\gamma(\omega)=0$.
    \item $\gamma'(\omega),  (\widetilde{\gamma}(\omega))'\in L^1\,.$
\end{itemize}
  where $\widetilde{\gamma}(\omega) = \gamma(\omega)e^{\beta\omega}$.
  For any $\epsilon>0$,
  {\begin{equation}
    \begin{aligned}
        &\left\|\rho_{\rm fix}(\Phi_\Gamma)-\sigma_\beta\right\|_1\leq
        \mathcal{O}\left(t_{{\rm mix},\Phi_\Gamma}(\epsilon/2)\frac{\beta}{\sigma}\left( (\|\gamma'\|_{L^1}+\|\widetilde{\gamma}'\|_{L^1})\log(\sigma)\right)
        + t_{{\rm mix},\Phi_\Gamma}(\epsilon/2)\frac{\sigma^2}{T} e^{-T^2/(4\sigma^2)}\right)+ \epsilon/2\\
    \end{aligned}
    \end{equation}}
  {In particular, if  $\|\gamma'\|_{L^1}+\|\widetilde{\gamma}'\|_{L^1}=\mathcal{O}(1)$, then for $\Gamma=\mathcal{O}(1)$, the choice  $\sigma=\widetilde{\Theta}\left(\beta t_{{\rm mix},\Phi_\Gamma}(\epsilon/2)/\epsilon\right)$ and $T = \widetilde{\Theta}(\sigma \log^{1/2}(t_{\rm mix,\Phi_\Gamma}(\epsilon/2)\sigma/\epsilon))$ ensures that $\left\|\rho_{\rm fix}(\Phi_\Gamma)-\sigma_\beta\right\|_1\leq\epsilon$.  }
  % above expression is bounded by $\epsilon$.
  % $\sigma=\Omega(\beta)$, and $T=\Omega(\sigma)$. Then,
    % \begin{equation}
    % \begin{aligned}
    %     &\|\Phi_\Gamma \sigma_\beta - \sigma_\beta \|_1<\mathcal{O}\left(\frac{\beta\Gamma^2}{\sigma^2}\left( (\|\gamma'\|_{L^1}+\|\widetilde{\gamma}'\|_{L^1})\log(\sigma)\right)
    %     + \frac{\Gamma^2\sigma}{T} e^{-T^2/(4\sigma^2)}\right)\,.
    % \end{aligned}
    % \end{equation}
    % Furthermore, given any $\epsilon>0$ and $\|\gamma'\|_{L^1}+\|\widetilde{\gamma}'\|_{L^1}=\mathcal{O}(1)$,
    % In addition, for any initial state $\rho_0$,
    % \begin{equation}\label{eqn:state_close}
    % \left\|\Phi_\Gamma^{\tau_{\rm mix,{\Phi_\Gamma}}(\epsilon/4)}(\rho_0)-\sigma_\beta\right\|_1<\epsilon\,.
    % \end{equation}
    % Here, the rescaled mixing time $t_{{\rm mix},{\Phi_\Gamma}}(\epsilon):=\Gamma^2\tau_{\rm mix,{\Phi_\Gamma}}(\epsilon)/\sigma$ defined in~\cref{def:mixing_time}.
\end{thm}
% We note that the second part of the result~\eqref{eqn:state_close} follows from combining the first part with~\cite[Theorem 8]{ding2025endtoendefficientquantumthermal}.
Now, we provide the proof of the first part of~\cref{thm:main_rigor_thermal}. Letting $T\rightarrow \infty$, we define
\begin{equation}
\label{eq:tildeGk}
    \begin{aligned}
        \widetilde{G}_{k,A_S}(\omega) & \coloneqq\int_{-\infty<t_1\leq \cdots \leq t_k < \infty} A_S(t_1)A^\dagger_S(t_2)\cdots e^{-i\omega\sum_{p=1}^k (-1)^p t_p} f(t_1)\cdots f(t_k)\mathrm{d}t_1 \cdots\mathrm{d}t_k\,, \\
    \end{aligned}
\end{equation}
and the limiting CPTP map
\begin{equation}\label{eqn:Phi_alpha_limit}
    \begin{aligned}
        \widetilde \Phi_\Gamma& \coloneqq  U_S(2T)\rho_n U_S^\dag(2T)
         + \mathbb{E}_{A_S}\left(\sum_{n\geq 1} \Gamma^{2n}(-1)^n \sum_{k = 0}^{2n} (-1)^k \int \gamma(\omega) U_S(T) \widetilde{G}^\dag_{2n-k,A_S}(\omega)U_S(T)\rho_nU_S^\dag(T)\widetilde{G}_{k,A_S}(\omega) U_S^\dag(T)\mathrm{d}\omega\right)\,.
    \end{aligned}
\end{equation}
The distance between $\Phi_\Gamma$ and $\widetilde\Phi_\Gamma$ can be controlled in the following lemma:
\begin{lem}\label[lem]{lm:Tinf}
When $\Gamma=\mathcal{O}(1)$, we have
\begin{equation}
\|\Phi_\Gamma - \widetilde \Phi_\Gamma\|_{1\rightarrow 1} = \mathcal{O}\left(\frac{\Gamma^2\sigma}{T} e^{-T^2/(4\sigma^2)}\right)\,.
\end{equation}
\end{lem}
\begin{proof}[Proof of~\cref{lm:Tinf}]
For simplicity, we only consider the case when $\mathcal{A}=\{A_S,-A_S\}$ is fixed. The general case is almost the same. By the general expression of the dynamic map, we have
\begin{equation}
\begin{aligned}
&\|\Phi_\Gamma - \widetilde\Phi_\Gamma\|_{1\rightarrow 1}\\
\leq & \max_{\rho_n: \|\rho_n\| = 1}\sum_{n\geq 1}\Gamma^{2n}\sum_{k=0}^{2n} \int |\gamma(\omega)| \left\| \widetilde{G}^\dag_{2n-k,A_S}(\omega)U_S(T)\rho_nU_S^\dag(T)\widetilde{G}_{k,A_S}(\omega) - {G}^\dag_{2n-k,A_S}(\omega)U_S(T)\rho_nU_S^\dag(T) {G}_{k,A_S}(\omega)\right\|_1\mathrm{d}\omega\\
\leq &\sum_{n\geq 1}\Gamma^{2n}\sum_{k=0}^{2n} \int |\gamma(\omega)| \left(\left\| \widetilde{G}^\dag_{2n-k,A_S}(\omega)- {G}^\dag_{2n-k,A_S}(\omega)\right\|\left\|\widetilde{G}_{k,A_S}(\omega) \right\|{+\left\| {G}^\dag_{2n-k,A_S}(\omega)\right\|\left\|G_{k,A_S}(\omega)- \widetilde{G}_{k,A_S}(\omega)\right\|}\right)\mathrm{d}\omega
\\
% &+ \sum_{n\geq 1}\Gamma^{2n}\sum_{k=0}^{2n} \int |\gamma(\omega)| \mathrm{d}\omega\\
\end{aligned}
\end{equation}
It is sufficient to bound $\left\|G_{k,A_S}(\omega)- \widetilde{G}_{k,A_S}(\omega)\right\|$ and $\left\|G_{k,A_S}(\omega)\right\|$, $\left\|\widetilde G_{k,A_S}(\omega)\right\|$. For the first term, we have
\begin{equation}
\begin{aligned}
&\left\|G_{k,A_S}(\omega)- \widetilde{G}_{k,A_S}(\omega)\right\|\\
\leq & \left\|\left(\int_{-\infty<t_1\leq \cdots \leq t_k < \infty} -\int_{-T<t_1\leq \cdots \leq t_k < T} \right)A_S(t_1)A^\dagger_S(t_2)\cdots e^{-i\omega\sum_{p=1}^k (-1)^p t_p} f(t_1)\cdots f(t_k)\mathrm{d}t_1 \cdots\mathrm{d}t_k\right\|\\
\leq & \left\|\int_{t_1\leq \cdots \leq t_k, \max_i|t_i|\geq T }  A_S(t_1)A^\dagger_S(t_2)\cdots e^{-i\omega\sum_{p=1}^k (-1)^p t_p} f(t_1)\cdots f(t_k)\mathrm{d}t_1 \cdots\mathrm{d}t_k\right\|\\
\leq & \|A_S\|^k \int_{t_1\leq \cdots \leq t_k}\chi_{\max|t_i|\geq T}f(t_1)\cdots f(t_k)\mathrm{d}t_1 \cdots\mathrm{d}t_k =\frac{\|A_S\|^k }{(k-1)!}\left(\int_{\mathbb{R}}f(t)\mathrm{d}t\right)^{k-1} \int_{|t|\geq T}f(t)\mathrm{d}t\\
 =& \frac{\mathcal{O}(\sigma)}{(k-1)!}\frac{1}{T}e^{-T^2/(4\sigma^2)}
\end{aligned}
\end{equation}
where we use $\|A_S\|\leq 1$ in the last inequality.

For the second term, we have
\begin{equation}
\begin{aligned}
\|G_{k,A_S}(\omega)\|\leq\left\|\widetilde{G}_{k,A_S}(\omega)\right\|\leq \|A_S\|^k\int_{-\infty<t_1\leq \cdots \leq t_k<\infty}  \left|f(t_1)\right|\cdots \left|f(t_k)\right|\mathrm{d}t_1 \cdots\mathrm{d}t_k=\frac{\mathcal{O}(1)}{k!}
\end{aligned}
\end{equation}
Combining these two bounds and $\|\gamma\|_{L^1}=1$, we have
\begin{equation}
    \begin{aligned}
        &\|\Phi_\Gamma - \widetilde\Phi_\Gamma\|_{1\rightarrow 1}\leq \frac{\sigma}{T}e^{-T^2/(4\sigma^2)}\sum_{n\geq 1}\Gamma^{2n}\sum_{k=0}^{2n} \left(\frac{\left(\mathcal{O}(1)\right)}{(2n-k-1)!k!}+\frac{\left(\mathcal{O}(1)\right)}{(2n-k)!(k-1)!}\right)\\
         = &\mathcal{O}\left(\frac{\Gamma^2\sigma}{T}e^{-T^2/(4\sigma^2)}\sum_{n\geq 1}\frac{1}{(2n-2)!}\left(\mathcal{O}(\Gamma)\right)^{2n-2}\right)= \mathcal{O}\left(\frac{\Gamma^2\sigma}{T} e^{-T^2/(4\sigma^2)}\right)\,.
    \end{aligned}
\end{equation}
\end{proof}
Using $\widetilde\Phi_\Gamma$, we have
    \begin{equation}
        \|\Phi_\Gamma \sigma_\beta - \sigma_\beta \|_1\leq \|\widetilde\Phi_\Gamma \sigma_\beta - \Phi_\Gamma\sigma_\beta \|_1+ \|\widetilde\Phi_\Gamma \sigma_\beta - \sigma_\beta \|_1\leq \|\widetilde\Phi_\Gamma  - \Phi_\Gamma \|_{1\rightarrow 1}+ \|\widetilde\Phi_\Gamma \sigma_\beta - \sigma_\beta \|_1\,.
    \end{equation}
Thus, it suffices to show $\|\widetilde\Phi_\Gamma \sigma_\beta - \sigma_\beta \|_1$ is small. Instead of relying directly on the detailed-balance condition of the Lindblad equation, we make use of {unital property} of the adjoint of the channel, namely $\widetilde\Phi_\Gamma^{\dag} I = I$. This condition is significantly easier to incorporate into the arbitrary-order expansion.
\begin{lem}({Unital property})\label[lem]{lm:AvoidDB}
We have
\begin{equation}
\begin{aligned}
\mathbb{E}_{A_S}\left(\sum_{n} \Gamma^{2n}(-1)^{n}
\sum_{k=0}^{2n} (-1)^{k}\int\gamma\left((-1)^{k}\omega\right)
\widetilde{G}^{\dagger}_{2n-k,A_S}(\omega)\,
\widetilde{G}_{k,A_S}(\omega)\,\mathrm d\omega\right)  = 0\,.
\end{aligned}
\end{equation}
\end{lem}
\begin{proof}
See \cref{appd: AvoidDB}.
\end{proof}
In order to prove $\widetilde\Phi_\Gamma$ preserves the thermal state, with~\cref{lm:AvoidDB}, it suffices to bound the following terms for each $A_S$:
\begin{equation}\label{eqn:separation}
\begin{aligned}
&\left\|\sum_{n\geq 1} \Gamma^{2n} (-1)^n \sum_{k = 0}^{2n}(-1)^k \int \gamma(\omega)\widetilde{G}_{2n-k,A_S}^\dag(\omega) \sigma_\beta \widetilde{G}_{k,A_S}(\omega) \mathrm{d}\omega\right\|_1\\
\leq &\underbrace{\left\|\sum_{n\geq 1}  \Gamma^{2n} (-1)^n \left(\sum_{k = 0, \rm even}^{2n} \int \gamma(\omega)\widetilde{G}_{2n-k,A_S}^\dag(\omega) \sigma_\beta \widetilde{G}_{k,A_S}(\omega) \mathrm{d}\omega-\sigma_\beta\int \gamma(\omega)\widetilde{G}_{2n-k,A_S}^\dag(\omega)\widetilde{G}_{k,A_S}(\omega) \mathrm{d}\omega\right)\right\|_1}_{\text{Term 1}}\\
+&\underbrace{\left\|\sum_{n\geq 1} \Gamma^{2n} (-1)^n \left(\sum_{k = 0, \rm odd}^{2n}-\int \gamma(\omega)\widetilde{G}_{2n-k,A_S}^\dag(\omega) \sigma_\beta \widetilde{G}_{k,A_S}(\omega) \mathrm{d}\omega-\sigma_\beta\int \underbrace{\gamma(\omega)\exp(\beta\omega)}_{=\gamma(-\omega)}\widetilde{G}_{2n-k,A_S}^\dag(\omega)\widetilde{G}_{k,A_S}(\omega) \mathrm{d}\omega\right)\right\|_1}_{\text{Term 2}}\\
+&{\underbrace{\left\|\sum_{n} \Gamma^{2n}(-1)^{n}
\sum_{k=0}^{2n} (-1)^{k}\int\gamma\left((-1)^{k}\omega\right)
\widetilde{G}^{\dagger}_{2n-k,A_S}(\omega)\,
\widetilde{G}_{k,A_S}(\omega)\,\mathrm d\omega\right\|_1}_{=0\ \text{by \cref{lm:AvoidDB}}}}\,.
\end{aligned}
\end{equation}
The analysis of both Term 1 and Term 2 can be simplified to the bound the following two terms:
\begin{equation}\label{eqn:diff}
\begin{aligned}
&\sigma_\beta^{-1}\widetilde G_{2n-k,A_S}^{\dagger}(\omega)\sigma_\beta - \widetilde G_{2n-k,A_S}^{\dagger}({\omega}),\ \text{when}\ k \text{ is even}\\
&\sigma_\beta^{-1}\widetilde G_{k,A_S}^{\dagger}(\omega)\sigma_\beta -\exp(\beta\omega)\widetilde G_{k,A_S}^{\dagger}({\omega}),\ \text{when}\ k \text{ is odd}\,.
\end{aligned}
\end{equation}
{Specifically, we have the following lemma:
\begin{lem}\label[lem]{lm:diff} We have
\begin{equation}\label{eq:diff_two_terms}
\emph{Term}\ 1=\mathcal{O}\left(\frac{\beta\Gamma^2\log(\sigma)}{\sigma^2}\|\gamma'\|_{L^1}\right),\quad \emph{Term}\ 2=\mathcal{O}\left(\frac{\beta\Gamma^2\log(\sigma)}{\sigma^2}\|\gamma'\|_{L^1}\right)
\end{equation}
\end{lem}
}
\begin{proof}[Proof of~\cref{lm:diff}]
Term 1 can be rewritten as
\begin{equation}\label{eq:Term1}
\mathrm{Term}\ 1=\left\lVert\sigma_\beta
\sum_{n\geq 1} \Gamma^{2n} (-1)^n
\left(
\sum_{\substack{k=0 \\ \text{even}}}^{2n}
\int \gamma(\omega)\left(\sigma_\beta^{-1}\widetilde G_{2n-k,A_S}^{\dagger}(\omega)\sigma_\beta - \widetilde G_{2n-k,A_S}^{\dagger}({\omega})\right)\widetilde G_{k,A_S}(\omega)\mathrm d\omega
\right)
\right\rVert_{1}
\end{equation}
In order to analyze this term, we introduce the following lemma to obtain an explicit expression for this difference.
\begin{lem}\label[lem]{lm:1} We have
\begin{equation}
\label{eqlm:sigmaGk}
    \begin{aligned}
        &\sigma_\beta^{-1}\widetilde{G}_{k,A_S}^\dag(\omega) \sigma_\beta = \frac{1}{\pi^{k/4}}\int_{-\infty< \tau_1\leq \cdots\leq \tau_{k}<\infty} \cdots A_S(2\sigma \tau_{2})A^\dagger_S(2\sigma\tau_1) \\
        &\exp\left(i2\sigma\omega \sum_{p=1}^{k} (-1)^p\tau_p-\omega \beta\Lambda(k)-\sum_{p=1}^{k} \tau_p^2+\frac{k\beta^2}{4\sigma^2}-i\frac{\beta}{\sigma}\sum_{p=1}^{k}\tau_p\right) \mathrm{d}\tau_1 \cdots \mathrm{d}\tau_{k}\,.
    \end{aligned}
\end{equation}
where $\Lambda(k) = \begin{cases}
    0 & \text{When }k \text{ is even}\\
    -1 & \text{When }k \text{ is odd}
\end{cases}$.
\end{lem}
\begin{proof}
    See \cref{appd:lm1proof}.
\end{proof}

In the following part of the proof, we assume that $A_S$ is Hermitian to simplify the notation. The calculation for non-Hermitian $A_S$ is the same. By substituting~\cref{eqlm:sigmaGk} to Term 1, we have
\begin{equation}\label{eq:Term1bd0}
\begin{aligned}
&\text{Term}\ 1
% \leq  \|\sigma_\beta\|_1\sum_{n\geq 1} \frac{(2\Gamma)^{2n}}{((2\pi)^{1/4})^{2n}}\sum_{\substack{k=0 \\ \text{even}}}^{2n}\left\|\int_{\substack{-\infty< s_1\leq s_2 \leq \cdots \leq s_{2n-k}<\infty, \\-\infty < t_1 \leq t_2 \leq \cdots \leq t_k <\infty}}A_S(2\sigma s_{2n-k}) \cdots A_S(2\sigma  s_1)\right.\\
% &A_S(2\sigma  t_1) \cdots A_S(2\sigma t_k) e^{-\sum_{p=1}^{2n-k}s_p^2 - \sum_{p=1}^k t_p^2}\\
% &\left.\int \gamma(\omega) e^{i2\sigma\omega \sum_{p = 1}^{2n-k}(-1)^p s_p - i2\sigma\omega \sum_{p= 1}^k (-1)^pt_p}\mathrm{d}\omega\left(\exp\left(\frac{(2n-k)\beta^2}{4\sigma^2}-i\frac{\beta}{\sigma} \sum_{p = 1}^{2n-k}s_p\right)-1\right)\mathrm{d}s_1 \cdots\mathrm{d}s_{2n-k}\mathrm{d}t_1 \cdots \mathrm{d}t_k\right\|\\
\leq \|\sigma_\beta\|_1\sum_{n\geq 1} \frac{(2\Gamma)^{2n}}{((2\pi)^{1/4})^{2n}}\|A_S\|^{2n} \sum_{\substack{k=0 \\ \text{even}}}^{2n}
\int_{\substack{-\infty< s_1\leq s_2 \leq \cdots \leq s_{2n-k}<\infty, \\-\infty < t_1 \leq t_2 \leq \cdots \leq t_k <\infty}}
e^{-\sum_{p=1}^{2n-k}s_p^2 - \sum_{p=1}^k t_p^2}\\
&\left|\int \gamma(\omega) e^{i2\sigma\omega (\sum_{p = 1}^{2n-k}(-1)^p s_p -  \sum_{p= 1}^k (-1)^pt_p)}\mathrm{d}\omega\right|\left|e^{\frac{(2n-k)\beta^2}{4\sigma^2}-i\frac{\beta}{\sigma} \sum_{p = 1}^{2n-k}s_p}-1\right|\mathrm{d}s_1 \cdots\mathrm{d}s_{2n-k}\mathrm{d}t_1 \cdots \mathrm{d}t_k\,.
\end{aligned}
\end{equation}
The above bound can be obtained by
\begin{equation}
    \text{Term}1 \leq \sum_{n\geq 1} \frac{2^{2n}\Gamma^{2n}}{\sqrt{2\pi}^{n}}\|A_S\|^{2n} \sum_{\substack{k=0 \\ \text{even}}}^{2n}\left(\text{Term}1.1 + \text{Term}1.2\right)\,.
\end{equation}
Here
\begin{equation}\label{eq:finalbdTerm11}
\begin{aligned}
   \text{Term 1.1}& \coloneqq \int_{\substack{-\infty< s_1\leq s_2 \leq \cdots \leq s_{2n-k}<\infty, \\-\infty < t_1 \leq t_2 \leq \cdots \leq t_k <\infty}}
e^{-\sum_{p=1}^{2n-k}s_p^2 - \sum_{p=1}^k t_p^2+\frac{(2n-k)\beta^2}{4\sigma^2}}\\
&\left|\int \gamma(\omega) e^{i2\sigma\omega (\sum_{p = 1}^{2n-k}(-1)^p s_p -  \sum_{p= 1}^k (-1)^pt_p)}\mathrm{d}\omega\right|\left|e^{-i\frac{\beta}{\sigma} \sum_{p = 1}^{2n-k}s_p}-1\right|\mathrm{d}s_1 \cdots\mathrm{d}s_{2n-k}\mathrm{d}t_1 \cdots \mathrm{d}t_k\\
   &=\frac{\beta \pi^n}{\sigma^2}\exp\left(\frac{(2n-k)\beta^2}{4\sigma^2}\right) \mathcal{O}\left(\|\gamma'\|_{L^1}\log(\sqrt{2n}\sigma)+1\right)\,,
\end{aligned}
\end{equation}
and
\begin{equation}\label{eq:finalbdTerm12}
\begin{aligned}
\text{Term 1.2}&\coloneqq\left|e^{\frac{(2n-k)\beta^2}{4\sigma^2}}-1\right|\int_{\substack{-\infty< s_1\leq s_2 \leq \cdots \leq s_{2n-k}<\infty, \\-\infty < t_1 \leq t_2 \leq \cdots \leq t_k <\infty}}
e^{-\sum_{p=1}^{2n-k}s_p^2 - \sum_{p=1}^k t_p^2}\\
&\left|\int \gamma(\omega) e^{i2\sigma\omega (\sum_{p = 1}^{2n-k}(-1)^p s_p -  \sum_{p= 1}^k (-1)^pt_p)}\mathrm{d}\omega\right| \mathrm{d}s_1 \cdots\mathrm{d}s_{2n-k}\mathrm{d}t_1 \cdots \mathrm{d}t_k \\
&=\frac{\pi^n}{\sigma}\left|e^{\frac{(2n-k)\beta^2}{4\sigma^2}}-1\right|\mathcal{O}\left(\frac{\|\gamma'\|_{L^1}}{\sqrt{n}} \log(\sqrt{2n}\sigma)+1 \right)\,.\\
\end{aligned}
\end{equation}
The detailed calculation for the scaling of Term 1.1 and Term 1.2 can be found in~\cref{appd:term1.1-2}.

Therefore, by $|e^x-1|\leq xe^x$ for $x\geq0$ and denoting $r = e^{\beta^2/(2\sigma^2)}$, we have
\begin{equation}
\begin{aligned}
&\text{Term 1} \leq \frac{\beta}{\sigma^2}\sum_{n\geq 1}(n+1)(2\sqrt{2\pi}\Gamma^2r)^n\mathcal{O}\left(\|\gamma'\|_{L^1}\log(\sqrt{2n}\sigma)+1\right)\\
+&\frac{\beta^2}{2\sigma^3}\sum_{n\geq 1}\sqrt{n}(n+1)(2\sqrt{2\pi}\Gamma^2r)^n\mathcal{O}\left(\log\left(\sqrt{2n}\sigma\right) \|\gamma'\|_{L^1}+ \sqrt{n}\right)
\end{aligned}
\end{equation}
When $2\sqrt{2\pi}\Gamma^2r<1$, the summation is dominated by a convergent geometric series. Therefore, with   $\sigma \gg 1$,
we have
\begin{equation}\label{eq:Term1bd1}
\begin{aligned}
&\text{Term 1}
= \mathcal{O}\left(\frac{\beta\Gamma^2\log(\sigma)}{\sigma^2}\|\gamma'\|_{L^1}\right)\,.
\end{aligned}
\end{equation}

When $k$ is odd, the bound for Term 2 is obtained in exactly the same way, with an additional factor of $\exp(\beta\omega)$. In particular, the additional factor $e^{\beta\omega}$ is precisely the one appearing in~\cref{eqlm:sigmaGk} for odd indices. Similar to \cref{eq:Term1bd0} and \cref{eq:Term1bd1}, we have
\begin{equation}\label{eq:Term2bd1}
\begin{aligned}
&\text{Term 2}
\leq \sum_{n\geq 1} \frac{(2\Gamma)^{2n}}{((2\pi)^{1/4})^{2n}}\|A_S\|^{2n} \sum_{\substack{k=0 \\ \text{odd}}}^{2n}
\int_{\substack{-\infty< s_1\leq s_2 \leq \cdots \leq s_{2n-k}<\infty, \\-\infty < t_1 \leq t_2 \leq \cdots \leq t_k <\infty}}
e^{-\sum_{p=1}^{2n-k}s_p^2 - \sum_{p=1}^k t_p^2}\\
&\left|\int  \gamma(\omega)\exp(\beta\omega) e^{i2\sigma\omega (\sum_{p = 1}^{2n-k}(-1)^p s_p -  \sum_{p= 1}^k (-1)^pt_p)}\mathrm{d}\omega\right|\left|\exp\left(\frac{(2n-k)\beta^2}{4\sigma^2}-i\frac{\beta}{\sigma} \sum_{p = 1}^{2n-k}s_p\right)-1\right|\mathrm{d}s_1 \cdots\mathrm{d}s_{2n-k}\mathrm{d}t_1 \cdots \mathrm{d}t_k\\
&  =  \mathcal{O}\left(\frac{\beta\Gamma^2\log(\sigma)}{\sigma^2}\|\widetilde \gamma'\|_{L_1}\right)\,.\\
\end{aligned}
\end{equation}
\end{proof}

Plugging the bounds in~\cref{lm:diff} in~\eqref{eqn:separation}, we obtain the upper bound
    \begin{equation}
    \begin{aligned}
        &\|\Phi_\Gamma \sigma_\beta - \sigma_\beta \|_1<\mathcal{O}\left(\frac{\beta\Gamma^2}{\sigma^2}\left( (\|\gamma'\|_{L^1}+\|\widetilde{\gamma}'\|_{L^1})\log(\sigma)\right)
        + \frac{\Gamma^2\sigma}{T} e^{-T^2/(4\sigma^2)}\right)\,.
    \end{aligned}
    \end{equation} Furthermore, using the similar techniques as~\cite[Equation D2]{ding2025endtoendefficientquantumthermal}, we have
\begin{equation}\label{eqn:thermalbd}
\begin{aligned}
    \|\rho_{\rm fix}(\Phi_\Gamma) - \sigma_\beta\|_1 &\leq \|\Phi^{\tau_{\rm mix,\Phi_\Gamma}(\epsilon/2)}_\Gamma\sigma_\beta - \sigma_\beta\|_1 +\|\Phi^{\tau_{\rm mix,\Phi_\Gamma}(\epsilon/2)}_\Gamma\sigma_\beta - \rho_{\rm fix}(\Phi_\Gamma)\|_1\\
    &\leq \tau_{\rm mix,\Phi_\Gamma}(\epsilon/2)\|\Phi_\Gamma\sigma_\beta - \sigma_\beta\|_1 + \epsilon/2 \\
    &\leq \mathcal{O}\left(\frac{\Gamma^2\tau_{{\rm mix},\Phi_\Gamma}(\epsilon/2)}{\sigma}\frac{\beta}{\sigma}\left( (\|\gamma'\|_{L^1}+\|\widetilde{\gamma}'\|_{L^1})\log(\sigma)\right)
        + \frac{\Gamma^2\tau_{{\rm mix},\Phi_\Gamma}(\epsilon/2)}{\sigma}\frac{\sigma^2}{T} e^{-T^2/(4\sigma^2)}\right)+ \epsilon/2\\
        & \leq \mathcal{O}\left(t_{{\rm mix},\Phi_\Gamma}(\epsilon/2)\frac{\beta}{\sigma}\left( (\|\gamma'\|_{L^1}+\|\widetilde{\gamma}'\|_{L^1})\log(\sigma)\right)
        + t_{{\rm mix},\Phi_\Gamma}(\epsilon/2)\frac{\sigma^2}{T} e^{-T^2/(4\sigma^2)}\right)+ \epsilon/2
\end{aligned}
\end{equation}
It is bounded by $\epsilon$ by choosing  $\sigma = \widetilde{\Theta}(\beta t_{\rm mix,\Phi_\Gamma}(\epsilon/2)/\epsilon), T = \widetilde{\Theta}(\sigma \log^{1/2}(t_{\rm mix,\Phi_\Gamma}(\epsilon/2)\sigma/\epsilon)), \Gamma = \mathcal{O}(1)$.

\subsection{Proof of \texorpdfstring{\cref{lm:AvoidDB}}{Lg}}\label{appd: AvoidDB}
We note that the expression of $\widetilde G_{k,A_S}$ depends on the system Hamiltonian $H$ and the coupling operator $A_S$. {In this proof, we make this dependence explicit by writing $\widetilde{G}_{k,H,A_S}(\omega)  \coloneqq\widetilde{G}_{k,A_S}(\omega)$ with $\widetilde{G}_{k,A_S}(\omega)$ defined in~\cref{eq:tildeGk}.}
Then,
\[
\begin{aligned}
\widetilde{G}^\dagger_{k,H,A_S}(\omega) =&\int_{-\infty<t_1\leq \cdots \leq t_k < \infty} \cdots A_S(t_2)A^\dagger_S(t_1) e^{i\omega\sum_{p=1}^k (-1)^p t_p} f(t_1)\cdots f(t_k)\mathrm{d}t_1 \cdots\mathrm{d}t_k\\
=&\int_{-\infty<s_1\leq \cdots \leq s_k < \infty} \cdots A_S(-s_{k-1})A^\dagger_S(-s_k) e^{-i\omega(-1)^{k-1}\sum_{p=1}^k (-1)^p s_p} f(s_1)\cdots f(s_k)\mathrm{d}t_1\cdots\mathrm{d}t_k,\quad s_j=-t_{k-j+1}\\
=&\left\{\begin{aligned}
&\widetilde{G}_{k,-H,A^\dagger_S}(\omega),\quad k\ \rm{odd}\,,\\
&\widetilde{G}_{k,-H,A_S}(-\omega),\quad k\ \rm{even}\,.
\end{aligned}\right.
\end{aligned}
\]
Similarly, we have {$\widetilde{G}_{k,H,A_S}(\omega) = \widetilde{G}^\dagger_{k,-H,A^\dagger_S}(\omega)$ for odd $k$ and $\widetilde{G}_{k,H,A_S}(\omega) =\widetilde{G}^\dagger_{k,-H,A_S}(-\omega)$ for even $k$.}

Now, we generate evolution operator $\widetilde \Psi_\Gamma$, defined in analogy with $\widetilde \Phi_\Gamma$, by replacing Hamiltonian in $\widetilde G_k$  with $-H$ and inverse temperature with $-\beta$,
\[
    \begin{aligned}
        \widetilde\Psi_\Gamma \rho_0
        % &=\int g(\omega)\mathrm{Tr}_E[U^{\Gamma}(T)(\rho_0\otimes \rho_E)U^{\Gamma\dag}(T)](\omega)\mathrm{d}\omega \\
        & = U_S(2T)\rho_0 U_S^\dag(2T) \\
        & +\mathbb{E}_{A_S}\left(\sum_n \Gamma^{2n} (-1)^n \sum_{k = 0}^{2n}(-1)^k \int \underbrace{\frac{g(\omega)+g(-\omega)}{1+ e^{-\beta\omega}}}_{=\gamma(-\omega)}U_S(T)\widetilde{G}_{2n-k,-H,A_S}^\dag(\omega)U_S(T) \rho_0 U_S^\dag(T)  \widetilde{G}_{k,-H,A_S}(\omega) U_S^\dag(T)\mathrm{d}\omega\right)\\
    \end{aligned}
\]
It straightforward to see that $\widetilde\Psi_\Gamma$ is a quantum channel. This implies
\[
\widetilde{\Psi}_\Gamma^{\dagger}[I]=I\ \Rightarrow\ \mathbb{E}_{A_S}\left(\sum_n \Gamma^{2n} (-1)^n \sum_{k = 0}^{2n}(-1)^k \int \gamma(-\omega)\widetilde{G}_{2n-k,-H,A_S}(\omega)\widetilde{G}^\dagger_{k,-H,A_S}(\omega) \mathrm{d}\omega\right)=0\,.
\]
Because $\left\{(A_i)^\dagger\right\}=\left\{A_i\right\}$, we have
\begin{equation}
\begin{aligned}
&\mathbb{E}_{A_S}\left(\sum_{n} \Gamma^{2n}(-1)^{n}
\sum_{k=0}^{2n} (-1)^{k}
\int \gamma\big( (-1)^{k}\,\omega \big)\,
\widetilde{G}^{\dagger}_{2n-k,H,A_S}(\omega)\,
\widetilde{G}_{k,H,A_S}(\omega)\mathrm d\omega\right) \\
=& \mathbb{E}_{A_S}\left(\sum_{n} \Gamma^{2n}(-1)^{n}
\sum_{k=0,\text{even}}^{2n}(-1)^k
\int \gamma(\omega)\,
\widetilde{G}_{2n-k,-H,A_S}(-\omega)\,
\widetilde{G}^\dagger_{k,-H,A_S}(-\omega)\mathrm d\omega\right)\\
&+\underbrace{\mathbb{E}_{A_S}\left(\sum_{n} \Gamma^{2n}(-1)^{n}
\sum_{k=0,\text{odd}}^{2n}(-1)^k
\int \gamma(-\omega)\,
\widetilde{G}_{2n-k,-H,A^\dagger_S}(\omega)\,
\widetilde{G}^{\dagger}_{k,-H,A^\dagger_S}(\omega)\mathrm d\omega\right)}_{=\mathbb{E}_{A_S}\left(\sum_{n} \Gamma^{2n}(-1)^{n}
\sum_{k=0,\text{odd}}^{2n}(-1)^k
\int \gamma(-\omega)\,
\widetilde{G}_{2n-k,-H,A_S}(\omega)\,
\widetilde{G}^{\dagger}_{k,-H,A_S}(\omega)\mathrm d\omega\right)}\\
=&\mathbb{E}_{A_S}\left(\sum_n \Gamma^{2n} (-1)^n \sum_{k = 0}^{2n}(-1)^k \int \gamma(-\omega)\widetilde{G}_{2n-k,-H,A_S}(\omega)\widetilde{G}^\dagger_{k,-H,A_S}(\omega) \mathrm{d}\omega\right)=0\,.
\end{aligned}
\end{equation}
The proof is complete.

\subsection{Proof of \texorpdfstring{\cref{lm:1}}{Lg}}\label{appd:lm1proof}
We prove the stated expression for $\sigma_\beta^{-1}\widetilde G_{k,A_S}^\dag \sigma_\beta$.
% We show that it admits the same structural form as $\widetilde G_{k,A_S}^\dag$ after a $\beta$-dependent shift of the integration variable.
For notation simplicity, we assume $A_S = A_S^\dag$. The extension to the non-hermitian $A_S$ is straightforward.

We first calculate $\sigma_\beta^{-1}\widetilde G_{k,A_S} \sigma_\beta$, the expression of $\sigma_\beta^{-1}\widetilde G_{k,A_S}^\dag \sigma_\beta$ can be obtained by applying the complex conjugate on it.
We use the change of variable, $s_{p}=t_{p}-t_{p+1}$ for $1\leq p<k$ and $s_{k}=t_{k}$,
% \begin{equation}
% t_p = \sum_{q = p}^{k} s_q, \forall p = 1, \cdots, k\Leftrightarrow \left\{\begin{aligned}
% &s_{p}=t_{p}-t_{p+1},\quad 1\leq p<k\\
% &s_{k}=t_{k}
% \end{aligned}\right.\,.
% \end{equation}
to $\widetilde G_{k,A_S} $ in~\cref{eq:tildeGk}, then
\begin{equation}
\begin{aligned}
\widetilde{G}_{k,A_S}(\omega) & = \int_{-\infty<s_1, \cdots s_{k-1}\leq 0, -\infty<s_{k}<\infty} A_S\left(\sum_{p =1}^{k}s_p\right)\cdots A_S(s_{k}) e^{-i\omega \sum_{p = 1}^{k}(-1)^p \sum_{q = p}^{k} s_q}
f\left(\sum_{p =1}^{k}s_p\right)\cdots f(s_{k})\mathrm{d}s_1 \cdots \mathrm{d}s_{k}\,.
\end{aligned}
\end{equation}
Denote the adjoint map as $\ad_H(A) = [H, A]$. Since $\ad_H(A_S(t)) = -i\partial_t A_S(t)$, for the product appearing in $\widetilde{G}_{k,A_S}(\omega)$, we have
\begin{equation}
    \ad_H(A_S\left(\sum_{p =1}^{k}s_p\right)\cdots A_S(s_{k})) =-i\partial_{s_{k}}\left(A_S\left(\sum_{p =1}^{k}s_p\right)\cdots A_S(s_{k})\right).
\end{equation}
By integration by parts,
\begin{equation}
\begin{aligned}
&\ad_H^m\left(\widetilde{G}_{k,A_S}\right) = \int_{\substack{-\infty<s_1, \cdots s_{k-1}\leq 0\\ -\infty<s_{k}<\infty}}  A_S\left(\sum_{p =1}^{k}s_p\right)\cdots A_S(s_{k})
 i^m\partial_{s_{k}}^m\left(
e^{-i\omega \sum_{p = 1}^{k}(-1)^p \sum_{q = p}^{k} s_q}
\prod_{q=1}^kf\left(\sum_{p =q}^{k}s_p\right)\right) \mathrm{d}s_1 \cdots \mathrm{d}s_{k}\,.\\
\end{aligned}
\end{equation}
By Baker–Campbell–Hausdorff formula, we have $\sigma_\beta \widetilde{G}_{k,A_S} \sigma_\beta^{-1} = e^{-\beta \operatorname{ad}_H}(\widetilde{G}_{k,A_S})$, then
\begin{equation}
\begin{aligned}
&\sigma_\beta \widetilde{G}_{k,A_S} \sigma_\beta^{-1}
=\int_{\substack{-\infty<s_1, \cdots s_{k-1}\leq 0\\ -\infty<s_{k}<\infty}} A_S\left(\sum_{p =1}^{k}s_p\right)\cdots A_S(s_{k})\sum_{n\geq 0}\frac{1}{m!}(-\beta)^mi^m\partial_{s_{k}}^m\left(
e^{-i\omega \sum_{p = 1}^{k}(-1)^p \sum_{q = p}^{k} s_q}
\prod_{q=1}^kf\left(\sum_{p =q}^{k}s_p\right)\right) \mathrm{d}s_1 \cdots \mathrm{d}s_{k}\\
&= \int_{\substack{-\infty<s_1, \cdots s_{k-1}\leq 0\\ -\infty<s_{k}<\infty}}  A_S\left(\sum_{p =1}^{k}s_p\right)\cdots A_S(s_{k})\left(
e^{-i\omega \sum_{p=1}^{k-1} (-1)^p \sum_{q=p}^{k-1} s_q-i\omega (s_{k} - i\beta)\Lambda(k)} \prod_{q=1}^kf\left(\sum_{p =q}^{k}s_p- i\beta\right)\right) \mathrm{d}s_1 \cdots \mathrm{d}s_{k}\,.\\
\end{aligned}
\end{equation}
where the last equality follows from Taylor expansion in $s_k$ and $\Lambda(k)$ is defined in the statement of \cref{lm:1}.

Substituting the expression of $f(t)$ and rescaling the variable $s_q \rightarrow 2\sigma s_q$ for $q = 1, \cdots, k$, we have
\begin{equation}
\label{eq:Gtilde}
\begin{aligned}
&\sigma_\beta \widetilde{G}_{k,A_S} \sigma_\beta^{-1}=  \frac{1}{\pi^{k/4}}\int_{\substack{-\infty<s_1, \cdots s_{k-1}\leq 0\\ -\infty<s_{k}<\infty}}A_S\left(2\sigma\sum_{p =1}^{k}s_p\right)\cdots A_S(2\sigma s_{k})\\
&\exp\left(-i2\sigma\omega \sum_{p=1}^{k} (-1)^p \sum_{q=p}^{k} s_q-\omega \beta\Lambda(k)-\sum_{p=1}^{k} \left(\sum_{q = p}^{k} s_q\right)^2+\frac{k\beta^2}{4\sigma^2}+i\frac{\beta}{\sigma}\sum_{p=1}^k ps_p\right) \mathrm{d}s_1 \cdots \mathrm{d}s_{k}\,.
\end{aligned}
\end{equation}

% \begin{equation}
% \begin{aligned}
% &\sigma_\beta \widetilde{G}_{k,A_S} \sigma_\beta^{-1}= \int_{-\infty<s_1, \cdots s_{k-1}\leq 0, -\infty<s_{k}<\infty}  A_S\left(\sum_{p =1}^{k}s_p\right)\cdots A_S(s_{k})e^{-i\omega \sum_{p=1}^{k} (-1)^p \sum_{q=p}^{k} s_q-\omega \beta\Lambda(k)}\\
% & \exp\left(-\frac{1}{4\sigma^2}\sum_{p=1}^{k}\left(\sum_{q = p}^{k}s_q-i\beta\right)^2\right)\left(\frac{1}{(2\pi)^{1/4}\sigma}\right)^{k}\mathrm{d}s_1 \cdots \mathrm{d}s_{k}\\
% \end{aligned}
% \end{equation}
% By rescaling the variable $s_q \rightarrow 2\sigma s_q$ for $q = 1, \cdots, k$, we have
% \begin{equation}
% \label{eq:Gtilde}
% \begin{aligned}
% &\sigma_\beta \widetilde{G}_{k,A_S} \sigma_\beta^{-1}=  \frac{(2\sigma)^{k}}{((2\pi)^{1/4}\sigma)^{k}}\int_{-\infty< s_1, \cdots, s_{k-1}\leq 0, -\infty< s_{k}<\infty} A_S\left(2\sigma\sum_{p =1}^{k}s_p\right)\cdots A_S(2\sigma s_{k})\\
% &e^{-i2\sigma\omega \sum_{p=1}^{k} (-1)^p \sum_{q=p}^{k} s_q-\omega \beta\Lambda(k)}
% \exp\left(-\sum_{p=1}^{k} \left(\sum_{q = p}^{k} s_q\right)^2\right) \exp\left(\frac{k\beta^2}{4\sigma^2}+i\frac{\beta}{\sigma}(s_1 + 2s_2 + \cdots + ks_{k})\right) \mathrm{d}s_1 \cdots \mathrm{d}s_{k}\,.
% \end{aligned}
% \end{equation}
% Define the function
% \begin{equation}
% \Lambda(k) \coloneqq \sum_{p = 1}^{k}(-1)^p = \begin{cases}
% 0 & \text{When }k \text{ is even}\\
% -1 & \text{When }k \text{ is odd}
% \end{cases}\,,
% \end{equation}
By applying the conjugate transpose to~\cref{eq:Gtilde} and another change of variable $\tau_p = \sum^{k}_{q =p}s_q$,
 we obtain \cref{eqlm:sigmaGk}. This concludes the proof.
% \begin{equation}
% \begin{aligned}
% \sigma_\beta^{-1}\widetilde{G}_{k,A_S}^\dag \sigma_\beta &= \frac{1}{\pi^{k/4}}\int_{-\infty< s_1, \cdots, s_{k-1}\leq 0, -\infty< s_{k}<\infty} A_S(2\sigma s_{k})\cdots A_S\left(2\sigma\sum_{p =1}^{k}s_p\right) \\
% &e^{i2\sigma\omega \sum_{p=1}^{k} (-1)^p \sum_{q=p}^{k} s_q-\omega \beta\Lambda(k)}
% \exp\left(-\sum_{p=1}^{k} \left(\sum_{q = p}^{k} s_q\right)^2\right) \exp\left(\frac{k\beta^2}{4\sigma^2}-i\frac{\beta}{\sigma}(s_1 + 2s_2 + \cdots + ks_{k})\right) \mathrm{d}s_1 \cdots \mathrm{d}s_{k}\,.
% \end{aligned}
% \end{equation}
% With another change of variable $\tau_p = \sum^{k}_{q =p}s_q$, we obtain \cref{eqlm:sigmaGk}. This concludes the proof.

\subsection{Bounds of Term 1.1 and Term 1.2}\label{appd:term1.1-2}
We bound Term 1.1 and Term 1.2 by separating the integration domain into a neighborhood of the resonant hyperplane and its complement. We define $z(\boldsymbol{s}, \boldsymbol{t}) = \sum_{p = 1}^{2n-k}(-1)^p s_p-\sum_{p= 1}^k (-1)^pt_p$ where $\boldsymbol{s} = (s_1, \cdots, s_{2n-k})$ and $\boldsymbol{t} = (t_1,\cdots, t_{k})$. We also write $\hat{\gamma}(2\sigma z(\boldsymbol{s}, \boldsymbol{t})) \coloneqq \int \gamma(\omega) e^{i2\sigma\omega z(\boldsymbol{s}, \boldsymbol{t})}\mathrm{d}\omega$.
% so that the bound of Term 1 can be obtained by
% \begin{equation}
%     \text{Term}1 \leq \sum_{n\geq 1} \frac{(2\Gamma)^{2n}}{((2\pi)^{1/4})^{2n}}\|A_S\|^{2n} \sum_{\substack{k=0 \\ \text{even}}}^{2n}\left(\text{Term}1.1 + \text{Term}1.2\right)\,.
% \end{equation}

Using $|\sin(x)|\leq |x|$, we obtain
    \begin{equation}\label{eq:term11_start}
        \begin{aligned}
            \text{Term 1.1} &= {2}\exp\left(\frac{(2n-k)\beta^2}{4\sigma^2}\right)\int_{\substack{-\infty< s_1\leq s_2 \leq \cdots \leq s_{2n-k}<\infty, \\-\infty < t_1 \leq t_2 \leq \cdots \leq t_k <\infty}}
        e^{-\sum_{p=1}^{2n-k}s_p^2 - \sum_{p=1}^k t_p^2}
        \left|\hat{\gamma}\left(2\sigma z\right(\boldsymbol{s}, \boldsymbol{t}\left)\right)\right| \left|\sin\left(\frac{\beta}{2\sigma} \sum_{p = 1}^{2n-k}s_p\right)\right|\mathrm{d}\boldsymbol{s}\mathrm{d}\boldsymbol{t}\\
        &\leq \frac{\beta}{\sigma}\exp\left(\frac{(2n-k)\beta^2}{4\sigma^2}\right)\int_{\substack{-\infty< s_1\leq s_2 \leq \cdots \leq s_{2n-k}<\infty, \\-\infty < t_1 \leq t_2 \leq \cdots \leq t_k <\infty}}
        e^{-\sum_{p=1}^{2n-k}s_p^2 - \sum_{p=1}^k t_p^2}
        \left|\hat{\gamma}\left(2\sigma z\right(\boldsymbol{s}, \boldsymbol{t}\left)\right)\right| \left| \sum_{p = 1}^{2n-k}s_p\right|\mathrm{d}\boldsymbol{s}\mathrm{d}\boldsymbol{t}\,.\\
        \end{aligned}
    \end{equation}
    We split the integral in~\cref{eq:term11_start} into two integrals $\mathrm{I} $ and $\mathrm{II}$ such that
    \begin{equation}\label{eq:bdTerm11}
        \text{Term 1.1}\leq \frac{\beta}{\sigma}\exp\left(\frac{(2n-k)\beta^2}{4\sigma^2}\right) \left(\mathrm{I} + \mathrm{II}\right).
    \end{equation}
    Here $\mathrm{I}$ is the contribution from $|z(\boldsymbol{s}, \boldsymbol{t})|\geq\delta$,
    \begin{equation}
    \begin{aligned}
        \mathrm{I} &\coloneqq\int_{\substack{-\infty< s_1\leq s_2 \leq \cdots \leq s_{2n-k}<\infty,\\ -\infty < t_1 \leq t_2 \leq \cdots \leq t_k <\infty, |z(\boldsymbol{s}, \boldsymbol{t})|\geq\delta}}
        e^{-\sum_{p=1}^{2n-k}s_p^2 - \sum_{p=1}^k t_p^2}
        \left|\hat{\gamma}\left(2\sigma z\right(\boldsymbol{s}, \boldsymbol{t}\left)\right)\right| \left| \sum_{p = 1}^{2n-k}s_p\right|\mathrm{d}\boldsymbol{s}\mathrm{d}\boldsymbol{t},
    \end{aligned}
    \end{equation}
    and $\mathrm{II}$ counts the contribution from the complementary region,
    \begin{equation}
        \begin{aligned}
            \mathrm{II}& \coloneqq \int_{\substack{-\infty< s_1\leq s_2 \leq \cdots \leq s_{2n-k}<\infty,\\ -\infty < t_1 \leq t_2 \leq \cdots \leq t_k <\infty, |z(\boldsymbol{s}, \boldsymbol{t})| < \delta}} e^{-\sum_{p=1}^{2n-k}s_p^2 - \sum_{p=1}^k t_p^2}
        \left|\hat{\gamma}\left(2\sigma z\right(\boldsymbol{s}, \boldsymbol{t}\left)\right)\right| \left| \sum_{p = 1}^{2n-k}s_p\right|\mathrm{d}\boldsymbol{s}\mathrm{d}\boldsymbol{t}\,.\\
        \end{aligned}
    \end{equation}

    On the region $|z(\boldsymbol{s}, \boldsymbol{t})|>\delta$, integration by parts gives
    \begin{equation}\label{eq:inequalD}
            \left|\hat{\gamma}\left(2\sigma z\right(\boldsymbol{s}, \boldsymbol{t}\left)\right)\right|
            \leq  \frac{\|\gamma'\|_{L_1}}{2\sigma|z(\boldsymbol{s}, \boldsymbol{t})|}.
    \end{equation}
    Substituting it to the expression of $\mathrm{I}$, we have
    \begin{equation}
    \begin{aligned}
        \mathrm{I} &\leq \frac{\|\gamma'\|_{L_1}}{2\sigma}\int_{\substack{-\infty< s_1\leq s_2 \leq \cdots \leq s_{2n-k}<\infty,\\ -\infty < t_1 \leq t_2 \leq \cdots \leq t_k <\infty, |z(\boldsymbol{s}, \boldsymbol{t})|\geq\delta}}
        e^{-\sum_{p=1}^{2n-k}s_p^2 - \sum_{p=1}^k t_p^2}\frac{\left| \sum_{p = 1}^{2n-k}s_p\right|}{|z(\boldsymbol{s}, \boldsymbol{t})|}\mathrm{d}\boldsymbol{s}\mathrm{d}\boldsymbol{t}.
    \end{aligned}
    \end{equation}
    We denote the concatenated vector as $\boldsymbol{u} = (\boldsymbol{s}, \boldsymbol{t})$, then $z(\boldsymbol{s}, \boldsymbol{t}) = v\cdot \boldsymbol{u}, \sum_{p = 1}^{2n-k}s_p = c\cdot \boldsymbol{u}$ where
    \begin{equation}
        v \coloneqq (-1, 1, -1, \cdots, (-1)^{2n-k}, 1, -1,\cdots, (-1)^{k+1})\in\mathbb{R}^{2n},\quad c \coloneqq (\underbrace{1,\cdots,1}_{2n-k},0,\cdots,0)\in\mathbb{R}^{2n}.
    \end{equation}
    The normalized vector $v_1 = v/\sqrt{2n}$ and $c_1=c /\sqrt{2n-k}$. Moreover, the vector $v_2 = c_1 - (c_1\cdot v_1) v_1 = c_1 -\Lambda(2n-k)v_1\perp v_1$, $\|v_2\| = 1$ and
    \begin{equation}
        \frac{\left| \sum_{p = 1}^{2n-k}s_p\right|}{|z(\boldsymbol{s}, \boldsymbol{t})|}  = \frac{|c\cdot \boldsymbol{u}|}{|v\cdot\boldsymbol{u}|}= \frac{\sqrt{2n-k}|(\Lambda(2n-k)v_1 +v_2)\cdot \boldsymbol{u}|}{\sqrt{2n}|v_1\cdot\boldsymbol{u}|}\leq \sqrt{\frac{2n-k}{2n}}|\Lambda(2n-k)| + \sqrt{\frac{2n-k}{2n}}\frac{|v_2\cdot \boldsymbol{u}|}{|v_1\cdot \boldsymbol{u}|}.
    \end{equation}
    Therefore,
    \begin{equation}
    \begin{aligned}
        \mathrm{I} &\leq \frac{\|\gamma'\|_{L_1}}{2\sigma}\sqrt{\frac{2n-k}{2n}}|\Lambda(2n-k)|\int_{\substack{-\infty< s_1\leq s_2 \leq \cdots \leq s_{2n-k}<\infty,\\ -\infty < t_1 \leq t_2 \leq \cdots \leq t_k <\infty, |z(\boldsymbol{s}, \boldsymbol{t})|\geq\delta}}
        e^{-\sum_{p=1}^{2n-k}s_p^2 - \sum_{p=1}^k t_p^2}\mathrm{d}\boldsymbol{s}\mathrm{d}\boldsymbol{t}\\
        & + \frac{\|\gamma'\|_{L_1}}{2\sigma}\sqrt{\frac{2n-k}{2n}}\int_{\substack{-\infty< s_1\leq s_2 \leq \cdots \leq s_{2n-k}<\infty,\\ -\infty < t_1 \leq t_2 \leq \cdots \leq t_k <\infty, |z(\boldsymbol{s}, \boldsymbol{t})|\geq\delta}}
        e^{-\sum_{p=1}^{2n-k}s_p^2 - \sum_{p=1}^k t_p^2}\frac{|v_2\cdot \boldsymbol{u}|}{|v_1\cdot \boldsymbol{u}|}\mathrm{d}\boldsymbol{s}\mathrm{d}\boldsymbol{t}.\\
    \end{aligned}
    \end{equation}
    Since $v_1\perp v_2$ and $\|v_i\| = 1$ for $ i = 1,2$, we can complete them so that $\{v_i\}_{i=1}^{2n}$
    is a set of orthonormal basis. We denote $r_i = v_i\cdot \boldsymbol{u}$, then $\|\boldsymbol{u}\|^2 = \sum_i r_i^2$, and
    \begin{equation}
        \begin{aligned}
            \mathrm{I}         & \leq \frac{\|\gamma'\|_{L_1}}{2\sigma}\sqrt{\frac{2n-k}{2n}}|\Lambda(2n-k)|\int_{\mathbb{R}^{2n}}e^{-\|\boldsymbol{u}\|^2}\rm d \boldsymbol{u}
        +\frac{\|\gamma'\|_{L_1}}{2\sigma}\sqrt{\frac{2n-k}{2n}}\int_{|\sqrt{2n}r_1|>\delta}\frac{1}{|r_1|}e^{-r_1^2}\mathrm{d}r_1\int_{\mathbb{R}^{2n-1}} e^{-\sum_{p=2}^n r_p^2}|r_2|\mathrm{d}r_2\cdots\mathrm{d}r_{2n}.
        \end{aligned}
    \end{equation}
    Using the integral of the Gaussian function, $\int_\mathbb{R} e^{-x^2}\mathrm{d}x = \sqrt{\pi}$ and $\int_{\mathbb{R}}|x|e^{-x^2}\mathrm{d}x = 1$, we have
    \begin{equation}
        \mathrm{I}\leq \frac{\|\gamma'\|_{L_1}}{2\sigma}\sqrt{\frac{2n-k}{2n}}\pi^n+\frac{\|\gamma'\|_{L_1}}{2\sigma}\sqrt{\frac{2n-k}{2n}}\times \int_{|r_1|>\delta/\sqrt{2n}}\frac{1}{|r_1|}e^{-r_1^2}\mathrm{d}r_1\times \pi^{n-1}.
    \end{equation}
    For $0<a<1$, we have  $\int_{|r_1|>a}\frac{1}{|r_1|}e^{-r_1^2}\mathrm{d}r_1\leq \mathcal{O}(1+\log(\frac{1}{a}))$. Therefore, choosing $\delta = 1/\sigma$, we obtain
    \begin{equation}
        \mathrm{I} \leq \mathcal{O}(\frac{\|\gamma'\|_{L_1}}{\sigma}\pi^n(\log(\sqrt{2n}\sigma)+1)) = \mathcal{O}(\frac{\|\gamma'\|_{L_1}}{\sigma}\pi^n\log(\sqrt{2n}\sigma)).
    \end{equation}
    For $\mathrm{II}$, since $|\sum_{p = 1}^{2n-k}s_p|
    % \leq \sqrt{2n-k}|\Lambda(2n-k)v_1\cdot \boldsymbol{u} +v_2\cdot \boldsymbol{u}|
    \leq \sqrt{2n-k}|r_1| + \sqrt{2n-k}|r_2|$ and $|\hat{\gamma}(2\sigma z(\boldsymbol{s}, \boldsymbol{t}))|\leq \|\gamma\|_{L_1}=1$, we have
    \begin{equation}
    \begin{aligned}
        \mathrm{II} &\leq \sqrt{2n-k} \int_{|r_1|\leq\delta/\sqrt{2n}} e^{-\sum_{p=1}^{2n} r_p^2}(|r_1|+|r_2|)\mathrm{d}r_1\cdots\mathrm{d}r_{2n}.\\
    \end{aligned}
    \end{equation}
    With $\delta = 1/\sigma$, this yields
\begin{equation}
    \mathrm{II} \leq \sqrt{2n-k}(\frac{\delta^2}{2n}\pi^{n-\frac{1}{2}}+ \frac{2\delta}{\sqrt{2n}}\pi^{n-1}) = \mathcal{O}(\frac{1}{\sqrt{2n}\sigma}\pi^n) = \mathcal{O}(\frac{\pi^n}{\sigma})
\end{equation}
In the large $\sigma$ regime considered,
combining the bound for $\mathrm{I}$ and $\mathrm{II}$ with~\cref{eq:bdTerm11} yields the scaling of Term 1.1 stated in~\cref{eq:finalbdTerm11}.

We now turn to Term 1.2. Similarly, we decompose the integration into the summation of $\mathrm{III}$ and $\mathrm{IV}$ such that $\text{Term 1.2}\leq \left|e^{\frac{(2n-k)\beta^2}{4\sigma^2}}-1\right|(\mathrm{III} + \mathrm{IV})$. Here $\mathrm{III}$ and $\mathrm{IV}$ are the integral restricted to region $z(\boldsymbol{s}, \boldsymbol{t})\geq \delta$ and $z(\boldsymbol{s}, \boldsymbol{t})< \delta$ respectively.
% counts the contribution when $z(\boldsymbol{s}, \boldsymbol{t})\geq \delta$ and $\mathrm{IV}$ is the integral for the complementary area.
By~\cref{eq:inequalD}, we have
\begin{equation}
\begin{aligned}
    \mathrm{III} &\coloneqq \int_{\substack{-\infty< s_1\leq s_2 \leq \cdots \leq s_{2n-k}<\infty, \\-\infty < t_1 \leq t_2 \leq \cdots \leq t_k <\infty,|z(\boldsymbol{s}, \boldsymbol{t})|\geq\delta}}
        e^{-\sum_{p=1}^{2n-k}s_p^2 - \sum_{p=1}^k t_p^2}
        \left|\hat{\gamma}\left(2\sigma z\right(\boldsymbol{s}, \boldsymbol{t}\left)\right)\right|\mathrm{d}\boldsymbol{s}\mathrm{d}\boldsymbol{t}\\
        &\leq \frac{\|\gamma'\|_{L_1}}{2\sigma} \int_{\substack{-\infty< s_1\leq s_2 \leq \cdots \leq s_{2n-k}<\infty, \\-\infty < t_1 \leq t_2 \leq \cdots \leq t_k <\infty,|z(\boldsymbol{s}, \boldsymbol{t})|\geq\delta}}
        e^{-\sum_{p=1}^{2n-k}s_p^2 - \sum_{p=1}^k t_p^2}\frac{1}{|z(\boldsymbol{s}, \boldsymbol{t})|}\mathrm{d}\boldsymbol{s}\mathrm{d}\boldsymbol{t}=\mathcal{O}\left(\frac{\|\gamma'\|_{L_1}}{\sigma}\frac{\pi^n}{\sqrt{n}}\log(\sqrt{2n}\sigma)\right)
\end{aligned}
\end{equation}
where in the last line we use the same coordinates $r_i$ as above to approximate the integral and $\delta = 1/\sigma$.

When $z(\boldsymbol{s}, \boldsymbol{t})< \delta$, using again $|\hat{\gamma}(2\sigma z(\boldsymbol{s},\boldsymbol{t}))|\leq 1$, we have
\begin{equation}
\begin{aligned}
    \mathrm{IV} &\coloneqq \int_{\substack{-\infty< s_1\leq s_2 \leq \cdots \leq s_{2n-k}<\infty, \\-\infty < t_1 \leq t_2 \leq \cdots \leq t_k <\infty,|z(\boldsymbol{s}, \boldsymbol{t})|<\delta}}
        e^{-\sum_{p=1}^{2n-k}s_p^2 - \sum_{p=1}^k t_p^2}
        \left|\hat{\gamma}\left(2\sigma z\right(\boldsymbol{s}, \boldsymbol{t}\left)\right)\right|\mathrm{d}\boldsymbol{s}\mathrm{d}\boldsymbol{t}\\
        & \leq \int_{\substack{-\infty< s_1\leq s_2 \leq \cdots \leq s_{2n-k}<\infty, \\-\infty < t_1 \leq t_2 \leq \cdots \leq t_k <\infty,|z(\boldsymbol{s}, \boldsymbol{t})|<\delta}}
        e^{-\sum_{p=1}^{2n-k}s_p^2 - \sum_{p=1}^k t_p^2}\mathrm{d}\boldsymbol{s}\mathrm{d}\boldsymbol{t} = \mathcal{O}(\frac{\pi^n}{\sigma})\\
\end{aligned}
\end{equation}
Consequently, we conclude the proof with the scaling of Term 1.2 stated in~\cref{eq:finalbdTerm12}.

\section{Ground state preparation beyond Lindblad limit}\label{sec:ground_state_appendix}
For ground state preparation, we consider the zero temperature setting, where the environment is initialized in $\rho_E = \ketbra{0}$. Let the system Hamiltonian be decomposed as $H = \sum_j \lambda_j \ketbra{\psi_j}$ with $\lambda_0< \lambda_1 \leq\cdots\leq  \lambda_m$ and denote the spectral gap as $\Delta = \lambda_1 - \lambda_0$. We have the following result:
\begin{thm}\label{thm:main_rigor_ground}
    (Ground state) Assume $H$ has a spectral gap $\Delta$ and let $\ket{\psi_0}$ be the ground state of $H$. Then for any $\epsilon>0$,
    \begin{equation}
    \begin{aligned}
    &\|\rho_{\rm fix}(\Phi_\Gamma) - \ketbra{\psi_0}\|_1\leq \mathcal{O}\left(t_{\rm{mix},\Phi_{\Gamma}}(\epsilon/2)\frac{\sigma^2}{T}\exp(-\frac{T^2}{4\sigma^2})\right) + t_{\rm{mix},\Phi_{\Gamma}}(\epsilon/2)\sum_{n=N+1}^\infty\frac{\sigma}{\Gamma^2} \frac{\left(\mathcal{O}\left(\Gamma\|A_S\|\right)\right)^{2n}}{{(2n)!}}\\
&+t_{\rm{mix},\Phi_{\Gamma}}(\epsilon/2)\exp(-\frac{\sigma^2\Delta^2}{2N})\frac{\sigma}{\Gamma^2}\sum_{n=1}^{N}(2n+1)\frac{\mathcal{O}(\Gamma\|A_S\|\|H\|)^{2n}}{(2n)!}+\epsilon/2 ,
    \end{aligned}
    \end{equation}
    In particular, the choice
    \begin{equation}
\Gamma={\Theta}(1),\quad  \sigma=\widetilde{\Theta}\left(\Delta^{-1}\log\left(\frac{t_{\rm mix,\Phi_\Gamma}(\epsilon/2)}{\epsilon}\right)\right),\quad T=\widetilde{\Theta}\left(\sigma\log^{1/2}\left(\frac{t_{\rm mix,\Phi_\Gamma}(\epsilon/2)}{\epsilon}\right)\right)
    \end{equation}
    ensures that $\|\rho_{\rm fix}(\Phi_\Gamma) - \ketbra{\psi_0}\|_1 <\epsilon$. Here $\widetilde{\Theta}$ suppresses the logarithmic dependence on $\epsilon$, $\sigma$.
\end{thm}
% \zd{This is far from rigorous. Since this is the rigorous version, we need to specify all the dependence clearly, especially the dependence on mixing time. A good example is Theorem S3.}

For simplicity,
we assume that $A_S=A_S^\dagger$. The general case follows by applying the
same argument to the Hermitian and anti-Hermitian parts of $A_S$.
The evolution operator used for ground state preparation can be obtained by taking the inverse temperature $\beta\to\infty$
in \cref{eqn:middle_evolution,eqn:rho_n_1/3_update}. Equivalently,  one may replace the function
$\gamma(\omega)$ by $\zeta(\omega) = g(\omega)+g(-\omega)$, which is even in $\omega$. The resulting channel is
\begin{equation}\label{eq:Phialpha}
    \begin{aligned}
        \Phi_\Gamma\rho_n &= U_S(2T)\rho_n U_S(2T)^\dag \\
        & + \mathbb{E}_{A_S}\left(\sum_{n\geq 1} \Gamma^{2n} (-1)^n \sum_{k = 0}^{2n}(-1)^k \int_{-\infty}^0 \zeta(\omega)U_S(T){G}_{2n-k,A_S}^\dag(\omega) U_S(T)\rho_n U_S^\dag(T) {G}_{k,A_S}(\omega)U_S^\dag(T) \mathrm{d}\omega\right)\,.\\
    \end{aligned}
\end{equation}
Since \cref{lm:Tinf} provides the error bound for $\|\Phi_\Gamma - \widetilde \Phi_\Gamma\|_{1\rightarrow 1}$, we focus on approximating $\|\widetilde{\Phi}_\Gamma \ketbra{\psi_0}- \ketbra{\psi_0}\|$.

When evaluating $\widetilde{\Phi}_\Gamma(\ketbra{\psi_0})$, it suffices to simplify the form of $\widetilde G_{n,A_S}^\dag(\omega)\ket{\psi_0}$. Using the expression in~\cref{eq:tildeGk} and change the integration variables to $s_k \rightarrow s_{n-k}$, we have
\begin{equation}
    \widetilde{G}_{n,A_S}^\dag(\omega)\ket{\psi_0}
    = \int_{-\infty<s_n\leq \cdots\leq s_1<\infty} A_S(s_1) \cdots A_S(s_n) e^{i\omega\sum_{k = 1}^n(-1)^{k+n+1} s_k}f(s_1)\cdots f(s_n)\mathrm{d}s_1 \cdots\mathrm{d}s_n \ket{\psi_0}.
\end{equation}
Using the Bohr-frequency decomposition $A_S(t) = \sum_{\nu\in \mathcal \mathcal{B}(H)} e^{i\nu t} A_{S}(\nu)$ where $A_S(\nu)$ is defined in~\cref{eq:Anu},
\begin{equation}
    \widetilde{G}_{n,A_S}^\dag(\omega)\ket{\psi_0}
    = \sum_{\nu_1,\nu_2,\dots,\nu_n\in \mathcal{B}(H)}I_n(\boldsymbol{\nu},\omega) A_S(\nu_1)\cdots A_S(\nu_n)\ket{\psi_0},
\end{equation}
where for $\boldsymbol{\nu}\in\mathbb{R}^n$,
\begin{equation}
    I_n(\boldsymbol{\nu},\omega):=
    \int_{-\infty<s_n\leq\cdots\leq s_1<\infty}
    \prod_{\ell=1}^{n}f(s_\ell)
    \exp\left(
        i\sum_{\ell=1}^{n}
        \bigl(\nu_\ell+(-1)^{\ell+n+1}\omega\bigr)s_\ell
    \right)
    \mathrm{d}s_1 \cdots\mathrm{d}s_n .
\end{equation}
We split
\begin{equation}
    \widetilde{G}_{n,A_S}^\dag(\omega)\ket{\psi_0} = \widetilde{G}_{n,A_S,0}^\dag(\omega)\ket{\psi_0} + \widetilde{G}_{n,A_S,\perp}^\dag(\omega)\ket{\psi_0}
\end{equation}
where the first term contains the summands with
$\sum_i\nu_i<\Delta$, and the second contains those with $\sum_i\nu_i>\Delta$.
%  the summation over Bohr frequency to the case when $\sum_i\nu_i<\Delta$ and $\sum_i\nu_i>\Delta$. We define
% \begin{equation}
%     \widetilde{G}_{n,A_S,0}^\dag(\omega)\ket{\psi_0} \coloneqq \sum_{\substack{\nu_1,\nu_2,\dots,\nu_n\in \mathcal{B}(H)\\\sum_i\nu_i<\Delta}}\int_{-\infty<s_n\leq \cdots\leq s_1<\infty}\prod_{r=1}^nf(s_r)e^{i\sum^n_{k=1}(\nu_k+(-1)^{k+n+1}\omega) s_k}\mathrm{d}s_n\dots\mathrm{d}s_1 A_S(\nu_1)\cdots A_S(\nu_n)\ket{\psi_0},
% \end{equation}
% and
% \begin{equation}
%     \widetilde{G}_{n,A_S,\perp}^\dag(\omega)\ket{\psi_0} \coloneqq \sum_{\substack{\nu_1,\nu_2,\dots,\nu_n\in \mathcal{B}(H)\\\sum_i\nu_i\geq \Delta}}\int_{-\infty<s_n\leq \cdots\leq s_1<\infty}\prod_{r=1}^nf(s_r)e^{i\sum^n_{k=1}(\nu_k+(-1)^{k+n+1}\omega) s_k}\mathrm{d}s_n\dots\mathrm{d}s_1 A_S(\nu_1)\cdots A_S(\nu_n)\ket{\psi_0}.
% \end{equation}
% Hence, $\widetilde{G}_{n,A_S}^\dag(\omega)\ket{\psi_0} = \widetilde{G}_{n,A_S,0}^\dag(\omega)\ket{\psi_0} + \widetilde{G}_{n,A_S,\perp}^\dag(\omega)\ket{\psi_0}$.
For the first term, it can be rewritten as
\begin{equation}
    \begin{aligned}
        &\widetilde{G}_{n,A_S,0}^\dag(\omega)\ket{\psi_0}  = \sum_{\sum_{i=1}^n\nu_i < \Delta  }I_n(\boldsymbol{\nu},\omega)\sum_{\lambda_{E_{i}}-\lambda_{E_{i+1}} = \nu_{i}, i = 1,2, \cdots, n}\ketbra{\psi_{E_1}}A_S\ketbra{\psi_{E_2}} A_S \cdots \ketbra{\psi_{E_n}}A_S \ketbra{\psi_{E_{n+1}}}\ket{\psi_0}\\
        & = \sum_{\sum_{i=1}^n\nu_i < \Delta  }I_n(\boldsymbol{\nu},\omega)
        \sum_{\lambda_{E_{i}}-\lambda_{E_{i+1}} = \nu_{i}, i = 1,2, \cdots, n, {\lambda_{E_1}-\lambda_{E_{n+1}}} = \sum_i\nu_i<\Delta}\ketbra{\psi_{E_1}}A_S\ketbra{\psi_{E_2}} A_S \cdots \ketbra{\psi_{E_n}}A_S \ket{\psi_0}\,.\\
    \end{aligned}
\end{equation}
Since the Hamiltonian is gapped with $\Delta >0$, every non-ground eigenvalue lies at least $\Delta$ above the {ground state energy}. Therefore, if the accumulated energy input satisfies {$\lambda_{E_1} - \lambda_{E_{n+1}} = \lambda_{E_1}-\lambda_0 = \sum_{i = 1}^n \nu_i<\Delta$, we must have $\lambda_{E_1} = \lambda_0 $.} Thus,
\begin{equation}
    \widetilde{G}_{n,A_S,0}^\dag(\omega)\ket{\psi_0}  {= {d_{n}(\omega)}\ket{\psi_0}},
\end{equation}
where
{\[
\begin{aligned}
        {d_n(\omega)} &\coloneqq
        \sum_{\sum_i\nu_i=0  }I_n(\boldsymbol{\nu},\omega)
        \bra{\psi_0}A_S(\nu_1)\cdots A_S(\nu_n)\ket{\psi_0}, d_0(\omega)=1.
\end{aligned}
\]}

Next, we bound the second term. Given the energy condition $\sum_i\nu_i\geq\Delta$, {we note that $\bra{\psi_0}\widetilde G_{r,A_S,\perp}^\dagger(\omega)\ket{\psi_0}=0$.} The following estimation of the multivariable Fourier transform of $f$ is used repeatedly. The proof is left to~\cref{sec:MFourierProof}.
\begin{lem}[Multivariable Fourier transformation of $f$]
\label{lm:multiFourier}
    \begin{equation}
        |\int^\infty_{-\infty}\dots\int^{s_{n-1}}_{-\infty}\prod_{r=1}^nf(s_r) \exp\left(i\sum^n_{k=1}\alpha_k s_k\right)\mathrm{d}s_n\dots\mathrm{d}s_1| \leq \frac{{(2^{3/4}\pi^{1/4})^n}}{n!}\exp(-\frac{\sigma^2}{n}(\sum_{k=1}^n \alpha_k)^2).
    \end{equation}
\end{lem}
Since $\omega<0$, the summation of the frequencies in $I_n(\boldsymbol{\nu},\omega)$ is
$$  \sum_{k=1}^n \nu_k + (-1)^{k+n+1}\omega = \begin{cases}
    \sum_{k=1}^n \nu_k & n\text{ is even}\\
    \sum_{k=1}^n \nu_k - \omega & n\text{ is odd}\\
\end{cases}\geq \Delta.$$
By~\cref{lm:multiFourier}, $\widetilde{G}_{n,A_S,\perp}^\dag(\omega)(\ket{\psi_0})$ can be approximated as
\begin{equation}
    \begin{aligned}
        \|\widetilde{G}_{n,A_S,\perp}^\dag(\omega)(\ket{\psi_0})\|
        \leq \|A_S\|^n \sum_{\sum_i \nu_i\geq \Delta} \left|I_n(\boldsymbol{\nu},\omega) \right|\leq \frac{\left({2^{3/4}\pi^{1/4}}\|A_S\|\left|\mathcal{B}(H)\right|\right)^n}{n!}\exp\left(-\frac{\Delta^2\sigma^2}{n}\right).
    \end{aligned}
\end{equation}
To approximate the number of Bohr frequencies, we introduce the  coarse grained Hamiltonian $H_\eta$ such that $H_\eta$ has discrete eigenvalues in $[-\|H\|, \|H\|]$ with uniform gap $\eta$, the same ground state $\lambda_0$ and $\|H- H_\eta\|\leq \eta$.
With the Hamiltonian $H_\eta$, the number of Bohr frequencies $|B({H_\eta})|\leq 4\|H\|/\eta$, and
\begin{equation}\label{eq:IInbd}
    \begin{aligned}
        \|\widetilde{G}_{n,H_\eta, A_S,\perp}^\dag(\omega)(\ket{\psi_0})\|
        &\leq \frac{\left({2^{11/4}\pi^{1/4}}\|A_S\|||H||/\eta\right)^n}{n!}\exp\left(-\frac{\Delta^2\sigma^2}{n}\right)\\
    \end{aligned}
\end{equation}
We are ready to quantify $\|\widetilde{\Phi}_{\Gamma, H_\eta}\ketbra{\psi_0} - \ketbra{\psi_0}\|_1$. {First, we note that
\begin{equation}\label{eq:tracepreserving}
\begin{aligned}
1=\mathrm{Tr}\left(\widetilde{\Phi}_{\Gamma, H_\eta}(\ketbra{\psi_0})\right)=&\bra{\psi_0}\widetilde{\Phi}_{\Gamma, H_\eta}(\ketbra{\psi_0})\ket{\psi_0} \\
&+ \Tr\left(\sum^\infty_{n=1}\Gamma^{2n}(-1)^n \sum_{k=1}^{2n} (-1)^k \int_{-\infty}^{0} \zeta(\omega)\widetilde{G}_{2n-k,H_\eta, A_s,\perp}^\dag(\omega) \ketbra{\psi_0} \widetilde{G}_{k,H_\eta, A_s,\perp} (\omega)\mathrm{d}\omega\right)
\end{aligned}
\end{equation}
because the $G_{[\cdot],0}$ term only contributes to the projection on the ground state, and the $G_{[\cdot],\perp}$ term is orthogonal to the ground state.
}
We define
\begin{equation}
    d \coloneqq \sum_{n=1}^\infty \Gamma^{2n} (-1)^n \sum_{k = 0}^{2n} (-1)^k \int_{-\infty}^0 {\zeta(\omega)} d_{2n-k}(\omega) d_k^*(\omega)\mathrm d \omega = \bra{\psi_0}\widetilde{\Phi}_{\Gamma, H_\eta}(\ketbra{\psi_0})\ket{\psi_0} -1 \in\mathbb{R}.
\end{equation}
A rough bound of $d_{n}$ is provided as
\begin{equation}\label{eq:roughbd}
    |{d_{n}(\omega)}| \leq \|\widetilde{G}_{n,A_S}^\dag(\omega)\|\leq \|A_S\|^{n}\int_{-\infty< s_{n}\leq\cdots \leq s_1 <\infty}  \prod_{r=1}^{n}f(s_r)\mathrm{d}s_{n} \cdots \mathrm{d}s_1 =\frac{\mathcal{O}\left(\|A_S\|\right)^{n}}{n!}.
\end{equation}
Using \cref{eq:IInbd}, the lower order term introduces a dependence on the spectral gap. For higher order terms, we use the rough bound
in~\cref{eq:roughbd}. Consequently,
\begin{equation}\label{eqn:H_eta_bound_term_1}
    \begin{aligned}
        &{\left\|\sum^\infty_{n=1}\Gamma^{2n}(-1)^n \sum_{k=1}^{2n} (-1)^k \int_{-\infty}^{0} {\zeta(\omega)}d^*_{2n-k}(\omega) \ketbra{\psi_0} \widetilde{G}_{k,H_\eta, A_s,\perp}^\dag(\omega)\mathrm{d}\omega\right\|_1}\\
        &+ \left\|\sum^\infty_{n=1}\Gamma^{2n}(-1)^n \sum_{k=1}^{2n} (-1)^k \int_{-\infty}^{0} {\zeta(\omega)}\widetilde{G}_{2n-k,H_\eta, A_s,\perp}^\dag(\omega) \ketbra{\psi_0} d_{k}({\omega})\mathrm{d}\omega\right\|_1\\
        &+\left\|\sum^\infty_{n=1}\Gamma^{2n}(-1)^n \sum_{k=1}^{2n} (-1)^k \int_{-\infty}^{0} {\zeta(\omega)}\widetilde{G}_{2n-k,H_\eta, A_s,\perp}^\dag(\omega) \ketbra{\psi_0} \widetilde{G}_{k,H_\eta, A_s,\perp} (\omega)\mathrm{d}\omega\right\|_1\\
        &\leq \sum^N_{n=1} \frac{\left(\mathcal{O}\left(\Gamma\|A_S\|\|H\|\eta^{-1}\right)\right)^{2n}}{{(2n)!}} \exp\left(-\frac{\sigma^2\Delta^2}{2n}\right) + \sum_{n=N+1}^\infty \frac{\left(\mathcal{O}\left(\Gamma\|A_S\|\right)\right)^{2n}}{{(2n)!}} .
    \end{aligned}
\end{equation}
Furthermore, according to~\eqref{eq:tracepreserving}, we know that
\begin{equation}
    d = - \Tr(\sum^\infty_{n=1}\Gamma^{2n}(-1)^n \sum_{k=1}^{2n} (-1)^k \int_{-\infty}^{0} {\zeta(\omega)}\widetilde{G}_{2n-k,H_\eta, A_s,\perp}^\dag(\omega) \ketbra{\psi_0} \widetilde{G}_{k,H_\eta, A_s,\perp} (\omega)\mathrm{d}\omega).
\end{equation}
% Thus, for $d\in[-1,0]$, $|d|$ is upper bounded by the trace norm which is approximated in~\cref{eqn:H_eta_bound_term_1}.
{According to~\eqref{eqn:H_eta_bound_term_1},
\begin{equation}\label{eqn:estimation_d}
\begin{aligned}
\left|d\right|&\leq \left\|\sum^\infty_{n=1}\Gamma^{2n}(-1)^n \sum_{k=1}^{2n} (-1)^k \int_{-\infty}^{0} {\zeta(\omega)}\widetilde{G}_{2n-k,H_\eta, A_s,\perp}^\dag(\omega) \ketbra{\psi_0} \widetilde{G}_{k,H_\eta, A_s,\perp} (\omega)\mathrm{d}\omega\right\|_1\\
&\leq \sum^N_{n=1} \frac{\left(\mathcal{O}\left(\Gamma\|A_S\|\|H\|\eta^{-1}\right)\right)^{2n}}{{(2n)!}} \exp\left(-\frac{\sigma^2\Delta^2}{2n}\right) + \sum_{n=N+1}^\infty \frac{\left(\mathcal{O}\left(\Gamma\|A_S\|\right)\right)^{2n}}{{(2n)!}}
\end{aligned}
\end{equation}
Using~\eqref{eqn:H_eta_bound_term_1},~\eqref{eqn:estimation_d}, we obtain that
\[
\begin{aligned}
&\left\|\widetilde \Phi_{\Gamma, H_\eta}\ketbra{\psi_0} - \ketbra{\psi_0}\right\|_1\\
\leq &\underbrace{\left\|\widetilde \Phi_{\Gamma, H_\eta}\ketbra{\psi_0} - \bra{\psi_0}\widetilde{\Phi}_{\Gamma, H_\eta}(\ketbra{\psi_0})\ket{\psi_0}\ketbra{\psi_0}\right\|_1}_{\text{Bounded by }\eqref{eqn:H_eta_bound_term_1}}+\underbrace{\left|\bra{\psi_0}\widetilde{\Phi}_{\Gamma, H_\eta}(\ketbra{\psi_0})\ket{\psi_0}-1\right|}_{\text{Bounded by }\eqref{eqn:estimation_d}}\\
\leq &\sum^N_{n=1} \frac{\left(\mathcal{O}\left(\Gamma\|A_S\|\|H\|\eta^{-1}\right)\right)^{2n}}{{(2n)!}} \exp\left(-\frac{\sigma^2\Delta^2}{2n}\right) + \sum_{n=N+1}^\infty \frac{\left(\mathcal{O}\left(\Gamma\|A_S\|\right)\right)^{2n}}{{(2n)!}}
\end{aligned}
\]}
% Thus, $\left\|\widetilde \Phi_{\Gamma, H_\eta}\ketbra{\psi_0} - \ketbra{\psi_0}\right\|_1$ is of the same order as in~\cref{eqn:H_eta_bound_term_1}.

The error caused by replacing $H$ with $H_\eta$ can be quantified using $\|H-H_\eta\|\leq \eta$. To distinguish between the Hamiltonians, we denote the quantum channel in ~\cref{eq:Phialpha} as $\Phi_{\Gamma, H}$ with operator $G_{n, H, A_S}$. The quantum channel $\Phi_{\Gamma, H_\eta}$ is defined by replacing $G_{n, H, A_S}$ with $G_{n, H_\eta, A_S}$
The definitions of $\widetilde{\Phi}_{\Gamma, [\cdot]}, \widetilde{G}_{n,[\cdot],A_S}$ follow accordingly. We also define $U_{S,[\cdot]} = \exp(-i[\cdot] t)$.
Consequently,
\begin{equation}\label{eq:HvsHeta}
    \begin{aligned}
        &\|\widetilde \Phi_{\Gamma, H}\ketbra{\psi_0} - \widetilde \Phi_{\Gamma, H_\eta}\ketbra{\psi_0}\|_{1}  \\
        &\leq\mathbb{E}_{A_S}\left(\sum_{n=1}^\infty \Gamma^{2n} \sum_{k = 0}^{2n}\int_{-\infty}^0 {\zeta(\omega)}\|\widetilde {G}_{2n-k,H,A_S}^\dag(\omega) \ketbra{\psi_0}  \widetilde G_{k,H,A_S}(\omega)
        -\widetilde G_{2n-k,H_\eta,A_S}^\dag(\omega)\ketbra{\psi_0} \widetilde G_{k,H_\eta,A_S}(\omega) \|_1\mathrm{d}\omega\right)\\
    \end{aligned}
\end{equation}
We notice that
\begin{equation}
\begin{aligned}
    &\widetilde{G}_{2n-k,[\cdot],A_S}^\dag(\omega) \ketbra{\psi_0}  \widetilde{G}_{k,[\cdot],A_S}(\omega)
    = \int_{\substack{-\infty<s_{2n-k}\leq\cdots\leq s_1 <\infty\\ -\infty< s_{2n-k+1}\leq \cdots \leq s_{2n}<\infty}} \mathrm{d}s_1\cdots \mathrm{d}s_{2n}\\
    &A_{S,[\cdot]}(s_1)\cdots A_{S,[\cdot]}(s_{2n-k})\ketbra{\psi_0} A_{S,[\cdot]}(s_{2n-k+1})\cdots A_{S,[\cdot]}(s_{2n})e^{i\omega\sum_{p = 1}^{2n-k}(-1)^{p-k+1} s_k}\prod_{r = 1}^{2n}f(s_r)
\end{aligned}
\end{equation}
By Duhamel's principle, $\|U_{S, H}(t) - U_{S, H_\eta}(t)\|\leq |t|\eta$. Then $\|A_{S,H}(t) - A_{S, H_\eta}(t)\|\leq 2\eta|t|\|A\|$, and \cref{eq:HvsHeta} can be upper bounded through telescoping sum as
\begin{equation}
    \begin{aligned}
        &\|\widetilde \Phi_{\Gamma, H}\ketbra{\psi_0} - \widetilde \Phi_{\Gamma, H_\eta}\ketbra{\psi_0}\|_{1}
        \leq 2\eta\sum_{n= 1}^\infty(\|A_S\|\Gamma)^{2n} \sum_{k=0}^{2n}   \int_{\substack{-\infty <t_{2n}\leq \dots\leq t_{k+1}<\infty\\ -\infty < t_1\leq \cdots\leq t_k <\infty}}\left(\sum_{i=1}^{2n} |t_i|\right)\prod_{r=1}^{2n}f(t_r)\mathrm{d}t_1 \cdots \mathrm{d}t_{2n}\\
        &= 2\eta\sum_{n= 1}^\infty(\|A_S\|\Gamma)^{2n} \sum_{k=0}^{2n}  \frac{2n}{k!(2n-k)!}\left(\int_{-\infty}^{\infty}|t|f(t)\mathrm{d}t\right) \left(\int_{-\infty}^{\infty}f(t)\mathrm{d}t\right)^{2n-1}\\
        & \leq 2\eta \sum_{n= 1}^\infty \frac{(\Gamma\|A_S\|)^{2n}}{(2n-1)!}\sum_{k=0}^{2n}\binom{2n}{k}\left(\frac{4\sigma}{(2\pi)^{1/4}}\right)\left(2^{3/4}\pi^{1/4}\right)^{2n-1}
        = \mathcal{O}\left(\eta \Gamma^2 \sigma\|A_S\|^2  \cosh(\mathcal{O}(\Gamma\|A_S\|))\right)\,.
    \end{aligned}
\end{equation}

Combining the above bounds with Theorem~\ref{lm:Tinf}, we have
\begin{equation}\label{eqn:final_1}
\begin{aligned}
&\|\Phi_{\Gamma, H}\ketbra{\psi_0} - \ketbra{\psi_0} \|_1\leq   \|\Phi_{\Gamma, H} -  \widetilde{\Phi}_{\Gamma, H}\|_{1\to 1} + \|\widetilde{\Phi}_{\Gamma, H}\ketbra{\psi_0}-\widetilde{\Phi}_{\Gamma, H_\eta}\ketbra{\psi_0}\|_1
 +
\|\widetilde \Phi_{\Gamma, H_\eta}\ketbra{\psi_0}-\ketbra{\psi_0}\|_1\\
&\leq \mathcal{O}\left(\Gamma^2\frac{\sigma}{T} \exp(-\frac{T^2}{4\sigma^2})\right) + \sum_{n=N+1}^\infty \frac{\left(\mathcal{O}\left(\Gamma\|A_S\|\right)\right)^{2n}}{{(2n)!}}\\
&+\min_\eta\left( \mathcal{O}\left(\eta \Gamma^2 \sigma\|A_S\|^2  \cosh(\mathcal{O}(\Gamma\|A_S\|))\right)+\exp\left(-\frac{\sigma^2\Delta^2}{2N}\right)\sum^N_{n=1} \frac{\left(\mathcal{O}\left(\Gamma\|A_S\|\|H\|\eta^{-1}\right)\right)^{2n}}{{(2n)!}} \right)\\
\end{aligned}
\end{equation}
The unique positive minimum is denoted as $\eta_*$ satisfying $\mathcal{O}(\Gamma^2\sigma \|A_S\|^2\cosh(\mathcal{O}(\Gamma\|A_S\|))) = \exp(-\frac{\sigma^2\Delta^2}{2N})\sum_{n=1}^{N}\frac{\mathcal{O}(\Gamma\|A_S\|\|H\|)}{(2n-1)!}\eta_*^{-2n-1} $.
Therefore, \cref{eqn:final_1} is bounded by
\begin{equation}
    \mathcal{O}\left(\Gamma^2\frac{\sigma}{T} \exp(-\frac{T^2}{4\sigma^2})\right) + \sum_{n=N+1}^\infty \frac{\left(\mathcal{O}\left(\Gamma\|A_S\|\right)\right)^{2n}}{{(2n)!}}+\exp(-\frac{\sigma^2\Delta^2}{2N})\sum_{n=1}^{N}(2n+1)\frac{\mathcal{O}(\Gamma\|A_S\|\|H\|\eta_*^{-1})^{2n}}{(2n)!}
\end{equation}
{Similar to~\eqref{eqn:thermalbd}, the fixed point error is bounded by
\begin{equation}
\begin{aligned}
&\left\|\rho_{\rm fix}(\Phi_\Gamma)-\ketbra{\psi_0}\right\|_1\leq\tau_{\rm{mix},\Phi_{\Gamma}}(\epsilon/2)\|\Phi_\Gamma(\ketbra{\psi_0}) - \ketbra{\psi_0}\|_1+\epsilon/2\\
& \leq \mathcal{O}\left(t_{\rm{mix},\Phi_{\Gamma}}(\epsilon/2)\frac{\sigma^2}{T}\exp(-\frac{T^2}{4\sigma^2})\right) + t_{\rm{mix},\Phi_{\Gamma}}(\epsilon/2)\sum_{n=N+1}^\infty\frac{\sigma}{\Gamma^2} \frac{\left(\mathcal{O}\left(\Gamma\|A_S\|\right)\right)^{2n}}{{(2n)!}}\\
&+t_{\rm{mix},\Phi_{\Gamma}}(\epsilon/2)\exp(-\frac{\sigma^2\Delta^2}{2N})\frac{\sigma}{\Gamma^2}\sum_{n=1}^{N}(2n+1)\frac{\mathcal{O}(\Gamma\|A_S\|\|H\|\eta_*^{-1})^{2n}}{(2n)!}+\epsilon/2 \\
\end{aligned}
\end{equation}
This quantity is upper bounded by $\epsilon$ when choosing
\begin{equation}
    \Gamma = \mathcal{O}(1), \quad N = \Theta\left(\frac{\log(1/\epsilon')}{\log\log(1/\epsilon')}\right), \quad \sigma=\widetilde{\Theta}\left(\Delta^{-1}\log(\frac{t_{\rm mix,\Phi_\Gamma}(\epsilon/2)}{\Gamma^2\epsilon})\right),\quad T=\widetilde{\Theta}\left(\sigma\log^{1/2}\left(\frac{t_{\rm mix,\Phi_\Gamma}(\epsilon/2)}{\epsilon}\right)\right).
\end{equation}}

\subsection{Proof of~\cref{lm:multiFourier}}
\label{sec:MFourierProof}
    We begin with the change of variable $t_n=s_n$, $t_k = s_k - s_{k+1}$ for $1\leq k\leq n-1$,
    \begin{equation}\label{eqn:f_integral}
\begin{aligned}
&\int_{-\infty< s_n\leq\cdots\leq s_1 <\infty}\prod_{r=1}^nf(s_r) \exp\left(i\sum^n_{k=1}\alpha_k s_k\right)\mathrm{d}s_n\dots\mathrm{d}s_1
=\int_{\substack{0\leq t_1,\cdots,t_{n-1}<\infty\\-\infty<t_n<\infty}}\prod_{r=1}^nf\left(\sum^n_{k=r}t_k\right)\exp\left(i\sum^n_{k=1}\left(\sum^k_{j=1}\alpha_j\right)t_k \right)\mathrm{d}t_n\dots\mathrm{d}t_1\,.
\end{aligned}
\end{equation}
Substituting the expression of $f(t)$ back to the above expression,
\cref{eqn:f_integral} can be simplified by first calculating the inner integral with respect to $t_n$ over the real line,
\begin{equation}\label{eq:multiFourierStep1}
    \begin{aligned}
        \text{\cref{eqn:f_integral}}=&\frac{1}{(2\pi)^{n/4}\sigma^n}\int_0^\infty \cdots \int_0^\infty\int_{-\infty}^{\infty} \exp\left(-\frac{n}{4\sigma^2}\left(t_n + \frac{1}{n}\sum_{k = 1}^{n-1}\sum_{p = k}^{n-1} t_p\right)^2+i\left(\sum_{k = 1}^n \alpha_k\right) t_n\right)\mathrm{d}t_n \\
        &\exp\left(\frac{1}{4\sigma^2n}\left(\sum_{k =1}^{n-1}\sum_{p = k}^{n-1} t_p\right)^2 -\frac{1}{4\sigma^2} \sum_{k = 1}^{n-1} \left(\sum_{p = k}^{n-1}t_p\right)^2+i\sum_{p = 1}^{n-1} \left(\sum_{k = 1}^p \alpha_k\right)t_p\right) \mathrm{d}t_{n-1}\cdots \mathrm{d}t_1\\
        =& \frac{1}{(2\pi)^{n/4}\sigma^n}\frac{2\sigma\sqrt{\pi}}{\sqrt{n}}\exp\left(-\frac{\sigma^2}{n}\left(\sum_{k = 1}^n \alpha_k\right)^2\right) C(\boldsymbol{\alpha}).
    \end{aligned}
\end{equation}
where
{\begin{equation}
    C(\boldsymbol{\alpha}) \coloneqq \int_0^\infty \cdots \int_0^\infty \mathrm{d}t_{n-1}\cdots \mathrm{d}t_1
        \exp\left(\frac{1}{4\sigma^2n}(\sum_{k =1}^{n-1}\sum_{p = k}^{n-1} t_p)^2 -\frac{1}{4\sigma^2} \sum_{k = 1}^{n-1} (\sum_{p = k}^{n-1}t_p)^2+i\sum_{p = 1}^{n-1} \left(\sum_{k = 1}^p \alpha_k\right)t_p-\frac{i}{n}\sum_{k= 1}^n\alpha_k\sum_{k = 1}^{n-1}\sum_{p = k}^{n-1} t_p\right).
\end{equation}
Define the variables $\tau_k = \sum_{p=k}^{n-1}t_p$ for $1\leq k\leq p-1$, the averaged coefficient $\bar{\alpha}_n = \frac{1}{n}\sum_{k=1}^n \alpha_k$ and $\omega_p =  \alpha_p - \bar{\alpha}_n$ for $1\leq p \leq n - 1$, we can reformulate it as
\begin{equation}
    C(\boldsymbol{\alpha}) = \int_{0\leq \tau_{n-1}\leq \cdots\leq \tau_1<\infty} \exp\left( \frac{1-n}{4\sigma^2n}\sum_{k=1}^{n-1}\tau_k^2 + \frac{1}{2\sigma^2 n}\sum_{1\leq i<j\leq n-1}\tau_i\tau_j+i\sum_{p=1}^{n-1}\omega_p\tau_p\right)\mathrm{d}\tau_1 \cdots \mathrm{d}\tau_{n-1}
\end{equation}
We denote the vector $\boldsymbol{\tau} = (\tau_1, \cdots, \tau_{n-1})^\top$, then
\begin{equation}
    C(\boldsymbol{\alpha}) = \int_{0\leq \tau_{n-1}\leq \cdots\leq \tau_1<\infty} \exp\left(-\frac{1}{4\sigma^2 n} \boldsymbol{\tau}^\top (nI_{n-1} - 1 1^T)\boldsymbol{\tau} + i\boldsymbol{\omega}^\top \boldsymbol{\tau}\right)\mathrm{d}\boldsymbol{\tau}.
\end{equation}
Taking absolute values provides a uniform bound,
\begin{equation}
\begin{aligned}
     |C(\boldsymbol{\alpha})| &\leq  \int_{0\leq \tau_{n-1}\leq \cdots\leq \tau_1<\infty} \exp\left(-\frac{1}{4\sigma^2 } \boldsymbol{\tau}^\top (I_{n-1} -\frac{1}{n} 1 1^T)\boldsymbol{\tau} \right)\mathrm{d}\boldsymbol{\tau} \\
     &=\frac{1}{n!}\int_{\mathbb{R}^{n-1}}\exp\left(-\frac{1}{4\sigma^2 } \boldsymbol{\tau}^\top (I_{n-1} -\frac{1}{n} 1 1^T)\boldsymbol{\tau} \right)\mathrm{d}\boldsymbol{\tau}= \frac{\sqrt{n}(2\sigma\sqrt{\pi})^{n-1}}{n!},
\end{aligned}
\end{equation}
where the factor $\frac{1}{n!}$ arises from the symmetry of the integration domain over the correlated Gaussian variables.
Substituting this back into~\cref{eq:multiFourierStep1}completes the proof.}

\section{Mixing time analysis for thermal state preparation}
\label{sec:mixing_time_thermal}

In this section, we analyze the mixing time of the quantum channel $\Phi_\Gamma$ defined in~\eqref{eqn:Phi_alpha} and establish the rigorous end-to-end complexity guarantee for thermal state preparation. Following~\cref{thm:dyson_Series_phi_alpha}, let $\mathcal{U}_S(T)[\cdot] = U_S(T) [\cdot] U_S^\dag(T)$, we rewrite the quantum channel as
\begin{equation}
    \Phi_\Gamma=  \mathcal{U}_S(T)\circ \Phi_{\rm part, \Gamma} \circ \mathcal{U}_S(T),
\end{equation}
where the middle step $\Phi_{\rm part, \Gamma}$ is defined in~\cref{eqn:rho_n_1/3_update}. From~\cref{lm:Tinf}, the finite time channel is close to its infinite window approximation as $\|\Phi_\Gamma - \widetilde{\Phi}_\Gamma\|_{1\rightarrow 1} = \mathcal{O}(\frac{\Gamma^2\sigma}{T}e^{-T^2/(4\sigma^2)})$. It is therefore sufficient to analyze the channel in~\cref{eqn:Phi_alpha_limit},
\begin{equation}
    \widetilde{\Phi}_\Gamma=\mathcal{U}_S(T)\circ \widetilde{\Phi}_{\rm part, \Gamma} \circ \mathcal{U}_S(T)
\end{equation}
Here $\widetilde{\Phi}_{\rm part, \Gamma}$ is defined in with $T\to\infty$,
\begin{equation}
    \widetilde{\Phi}_{\rm part, \Gamma}  = I + \sum_{n\geq 1}\Gamma^{2n}\widetilde{\mathcal{N}}_n ,
\end{equation}
where
\begin{equation}
    \widetilde{\mathcal{N}}_n \rho = \mathbb{E}_{A_S}\left((-1)^n \sum_{k=0}^{2n}(-1)^k \int \gamma(\omega)\widetilde{G}_{2n-k,A_S}^\dag(\omega) \rho  \widetilde{G}_{k,A_S}(\omega)\mathrm{d}\omega\right),
\end{equation}
and $\widetilde{G}_{k,A_S}(\omega)$ in defined in~\cref{eq:tildeGk}.
We rewrite it as
\begin{equation}\label{eqn:Phi_alpha_limit_rewrite}
    \widetilde{\Phi}_{\rm part, \Gamma}   = I + \sum_{n\geq 1}\left(\frac{\Gamma^2}{\sigma}\right)^{n}\widetilde{\mathcal{M}}_n = I + \frac{\Gamma^2}{\sigma} (\widetilde{\mathcal{M}}_1 + \underbrace{\sum_{n\geq 2}\left(\frac{\Gamma^2}{\sigma}\right)^{n-1}\widetilde{\mathcal{M}}_n}_{\widetilde{\mathcal{M}}_{\geq 2}})  = I + \frac{\Gamma^2}{\sigma} \widetilde{\mathcal{M}},
\end{equation}
where $\widetilde{\mathcal{M}}_n = \sigma^n \widetilde{\mathcal{N}}_n$.
Here we rewrite it as the expansion in terms of $\mathcal{M}_n$ in order to make use of previous analysis results approximating $\mathcal{M}_1$ with KMS-detailed balanced Lindbladian with unitary drift and an extra Lamb shift term~\cite{ding2025endtoendefficientquantumthermal}.

The proof of the mixing time relies heavily on a perturbation analysis of the spectral gap of $\Phi_{\rm part,\Gamma}$. To establish the contraction of $\Phi_\Gamma$, we adopt the strategy of \cite{ding2025endtoendefficientquantumthermal} by utilizing the weighted Hilbert-Schmidt norm, $\|\sigma^{-1/4}_\beta \cdot \sigma^{-1/4}_\beta\|_2$.  Following the argument of~\cite{ding2025endtoendefficientquantumthermal}, we first observe that the weighted Hilbert-Schmidt norm is invariant under the action of $\mc{U}_S$, i.e., $\|\sigma^{-1/4}_\beta \mc{U}_S(\rho) \sigma^{-1/4}_\beta\|_2=\|\sigma^{-1/4}_\beta \rho \sigma^{-1/4}_\beta\|_2$. Consequently, to determine the contraction of $\Phi_\Gamma=\mc{U}_S \circ \Phi_{\rm part,\Gamma}\circ \mc{U}_S$, it suffices to analyze $\Phi_{\rm part,\Gamma}$.

In the case when ${\Gamma}$ is not necessarily small, $\Phi_{\rm part,\Gamma}$ can not be describe by the Lindblad dynamics. This prohibits us from directly applying the spectral gap analysis for Lindbladians. Instead, we see $\Phi_{\rm part,\Gamma}$ as a discrete perturbed channel of the Lindbladian dynamics generated by $-i[H_{{C}},\cdot]+\mc{L}_{\sigma, \rm KMS}$. Leveraging the analysis in~\cref{thm:main_rigor_thermal}, $\mathcal{M}_{{\geq 2}}$ preserves the thermal state. Thus, it suffices to study its effect of the contraction rate of the channel. One technical challenge is that, the leading order term in $\mathcal{M}_{{\geq 2}}$ corresponds to the fourth term in the Taylor expansion of $\Phi_{\rm part,\Gamma}$, which is positive and tends to shrink the spectral gap. However, thanks to the choice of filter function $f$ and inherent detailed balance condition that it satisfies, we can adopt similar idea as in the proof of~\cref{thm:main_rigor_thermal} to show that $\mathcal{M}_{{\geq 2}}$ has a small operator norm in the weighted Hilbert-Schmidt inner product, which scales as $\mathcal{O}({\Gamma^2})$. This implies that, as long as $\Gamma$ is small compared to the inverse spectral gap of $\mathcal{L}_{\sigma,\rm KMS}$, the gap perturbation caused by $\mathcal{M}_{{\geq 2}}$ can be controlled and the contraction rate of $\Phi_{\rm part,\Gamma}$ is still comparable with that of $\mathcal{L}_{\sigma,\rm KMS}$.

For a superoperator $\mathcal{M}$, we take a similarity transformation and decompose into the Hermitian and the anti-Hermitian parts and define the corresponding super operator as
\[
\begin{aligned}
\mathcal{K}(\sigma_\beta, \mathcal{M})=\sigma_{\beta}^{-1 / 4} \mathcal{M}\left[\sigma_{\beta}^{1 / 4} \cdot \sigma_{\beta}^{1 / 4}\right] \sigma_{\beta}^{-1 / 4} & =\mathcal{H}(\sigma_\beta, \mathcal{M})+\mathcal{A}(\sigma_\beta, \mathcal{M}), \\
\mathcal{K}(\sigma_\beta, \mathcal{M})^{\dagger}=\sigma_{\beta}^{1 / 4} \mathcal{M}^{\dagger}\left[\sigma_{\beta}^{-1 / 4} \cdot \sigma_{\beta}^{-1 / 4}\right] \sigma_{\beta}^{1 / 4} & =\mathcal{H}(\sigma_{\beta}, \mathcal{M})-\mathcal{A}(\sigma_{\beta}, \mathcal{M}).
\end{aligned}
\]
% We note that the above detailed balance condition allows a coherent term that commutes with $\sigma_\beta$ in $\mc{M}$. It is straightforward to check that if $\mc{M}$ satisfies the detailed balance with unitary drift and is traceless when applying density operator, then $\mc{H}(\sigma_\beta,\mc{M})(\sqrt{\sigma_\beta})=0$ and $\mc{M}(\sigma_\beta)=0$.
The structure of this section is as follows:
\begin{itemize}
\item In~\cref{sec:KMS_approximation}, we review the result in~\cite{ding2025endtoendefficientquantumthermal} and approximate leading order term in~\eqref{eqn:middle_evolution} using KMS-Lindbladian dynamics plus higher-order correction terms. The result also includes a choice of the function $g(\omega)$ in~\eqref{eqn:Phi_alpha} to ensure the approximation holds.
\item In~\cref{sec:spectral_gap_perturbation}, we present a novel perturbation theory in~\cref{thm:main_result_gap} for $\Phi_\Gamma$ and show that the spectral gap of $\Phi_\Gamma$ can be lower bounded by that of the approximated KMS-Lindbladian dynamics even beyond the Lindblad limit.
\item In~\cref{sec:application}, we combine the results in the previous two subsections to establish the mixing time bound for $\Phi_\Gamma$ and present the end-to-end complexity guarantee for thermal state preparation.
\item In~\cref{sec:proof_main_result_gap}, we provide the proof of~\cref{thm:main_result_gap}.
\end{itemize}

\subsection{Approximate \texorpdfstring{$\Phi_\Gamma$}{lg} using KMS-Lindbladian dynamics}\label{sec:KMS_approximation}
% Following~\cref{thm:dyson_Series_phi_alpha}, we define
% \begin{equation}\label{eqn:M1_def}
% \mathcal{M}_1[\rho]=-\mathbb{E}_{A_S}\left(\sum_{k = 0}^{2}(-1)^k {\Gamma^2}\int \gamma(\omega){G}_{2-k,A_S}^\dag(\omega) \rho  {G}_{k,A_S}(\omega)\mathrm{d}\omega\right)\,,
% \end{equation}
% \begin{equation}\label{eqn:M2_def}
% \mathcal{M}_2[\rho]=\mathbb{E}_{A_S}\left(\sum_{n=2}^\infty \Gamma^{2n} (-1)^n \sum_{k = 0}^{2n}(-1)^k \int \gamma(\omega){G}_{2n-k,A_S}^\dag(\omega) \rho  {G}_{k,A_S}(\omega)\mathrm{d}\omega\right)\,,
% \end{equation}
% Here $\mathcal{M}_1$ contains all $\mathcal{O}(\Gamma^2)$ terms, which can be approximated by the  Lindbladian derived in~\cite{ding2025endtoendefficientquantumthermal}, as shown in~\cref{thm:KMS_DBC_approx}, and $\mathcal{M}_2$ collects all higher-order correction terms beyond Lindblad dynamics, starting from $\mathcal{O}(\Gamma^4)$ term, i.e.
% \begin{equation}
%     \Phi_\Gamma \rho = \rho +  \mathcal{M}_1[\rho] +  \mathcal{M}_2[\rho]\,.
% \end{equation}

% Let $\|\cdot\|_p$ be the Schatten $p$-norm, $\|\cdot\|$ be the operator norm, $\|\cdot\|_{p\leftrightarrow p}$ be the induced Schatten norm for superoperators.

Given a jump operator $V$,  we define the associated Lindbladian dissipative operator as
\begin{equation}
\mathcal{D}_{V}(\rho)=V\rho V^\dagger- \frac12\{V^\dagger V,\rho\}\,.
\end{equation}
Notice that the leading term
\begin{equation}
\widetilde{\mathcal N}_1(\rho)
=
\mathbb E_{A_S}
\int_{-\infty}^{\infty}\gamma(\omega)
\left[
\widetilde G_{1,A_S}^\dagger(\omega)\rho \widetilde G_{1,A_S}(\omega)
-
\widetilde G_{2,A_S}^\dagger(\omega)\rho
-
\rho\,\widetilde G_{2,A_S}(\omega)
\right]d\omega .
\end{equation}
Since $\widetilde G_{2,A_S}(\omega)+\widetilde G_{2,A_S}^\dagger(\omega)
= \widetilde G_{1,A_S}(\omega)\widetilde G_{1,A_S}^\dagger(\omega),$ we can rewrite it as the Lindbladian with unitary drift
\[
\widetilde{\mathcal N}_1(\rho)
=
\mathbb E_{A_S}\int \gamma(\omega)
\mathcal D_{\widetilde G_{1,A_S}^\dagger(\omega)}(\rho)\,d\omega
-i\left[H_{\rm Lamb},\rho
\right],\]
where
\[
H_{\rm Lamb}
=
\mathbb E_{A_S}\int \gamma(\omega)
Q_{A_S}(\omega)\,d\omega ,\quad
Q_{A_S}(\omega)
=
i\left(
\widetilde G_{2,A_S}(\omega)
-
\frac12
\widetilde G_{1,A_S}(\omega)
\widetilde G_{1,A_S}^\dagger(\omega)
\right).
\]
Since $\widetilde{M}_1 = \sigma \widetilde{N}_1$, it is also of the Lindbladian form. However, since $g(\omega)$ is the distribution and $\gamma(\omega) = \frac{g(\omega)+ g(-\omega)}{1 + e^{\beta\omega}}$, the above Lindbladian does not satisfy the KMS detailed balance condition in general. The difference between them is  quantified in~\cite[Theorem 27]{ding2025endtoendefficientquantumthermal} and~\cite[Appendix B (B2)]{slezak2026polynomialtime}. Here we choose $x = 1/\beta$ in~\cite[Theorem 27]{ding2025endtoendefficientquantumthermal}.

\begin{thm}\label{thm:KMS_DBC_approx}
For any $\sigma>\beta$, we set
\begin{equation}\label{eqn:g_choice}
g(\omega)=\frac{1}{Z}\exp\left(-\frac{(\beta\omega+1)^2}{2(2-\frac{\beta^2}{4\sigma^2})}\right),\quad Z={\frac{1}{\beta}}\sqrt{2\pi\left(2-\frac{\beta^2}{4\sigma^2}\right)}\,.
\end{equation}
Then, there exists a Lindbladian $\mc{L}_{\sigma, \rm KMS}$ that satisfies KMS detailed balance condition (defined in~\cref{sec:preliminaries}) and a Hermitian operator ${H_{C}}$ such that
\[
\left\|\widetilde{\mc{M}}_1-\left(-i[{H_{C}},\cdot]+\mc{L}_{\sigma, \rm KMS}\right)\right\|_{1\leftrightarrow1}=\mc{O}\left({\frac{\beta^2}{\sigma}}\right)\,,
\]
and
\[
\left\|\sigma^{-1/4}_{\beta}{H_{C}} \sigma^{1/4}_{\beta}-\sigma^{1/4}_{\beta}{H_{C}} \sigma^{-1/4}_{\beta}\right\|=\mc{O}\left(\frac{\beta{^{2}}}{\sigma}\right)\,.
\]
Here, $\mc{L}_{\sigma, \rm KMS}$ takes the form of
\begin{equation}\label{eqn:KMS_Lindbladian}
\mc{L}_{\sigma, \rm KMS}[\rho]=\mathbb{E}_{A_S}\left(-i\left[B_{A_S},\rho\right]+\int^\infty_{-\infty}{\hat{\gamma}}(\omega)\mathcal{D}_{{\sqrt{\sigma}\widetilde{G}_{1,A_S}^\dagger(\omega)}}(\rho)\mathrm{d}\omega\right)\,,
\end{equation}
with ${\hat{\gamma}}(\omega)=g(\omega)$.
\end{thm}
In the weak coupling limit ($\Gamma\ll 1$), the channel $\Phi_{\rm part,\Gamma}$ is well-approximated by the Lindbladian dynamics generated by $-i[H_{\sigma,\text{Lamb}}, \cdot] + \mathcal{L}_{\sigma,\text{KMS}}$. This generator admits the thermal state $\sigma_\beta$ as its unique fixed point~\cite{ding2025endtoendefficientquantumthermal}, with a mixing time determined by the spectral gap of $\mathcal{L}_{\sigma,\text{KMS}}$. Consequently, the fixed point of $\Phi_\Gamma$ remains close to $\sigma_\beta$~\cite{ding2025endtoendefficientquantumthermal,scandi2025thermal,hahn2025,lloyd2025quantumthermal}, and its mixing time is comparable to that of $\mathcal{L}_{\sigma,\text{KMS}}$~\cite{ding2025endtoendefficientquantumthermal,slezak2026polynomialtime}. Now, we extend the analysis beyond the Lindblad limit, focusing on the regime where $\Gamma$ is not necessarily small.

\subsection{Perturbation result for the spectral gap}\label{sec:spectral_gap_perturbation}

According to~\cref{thm:main_rigor_thermal}, to achieve $\epsilon$-accuracy in trace distance, the total Hamiltonian simulation time is
\[
T_{\rm total}:= \tau_{\rm mix,{\Phi_\Gamma}}(\epsilon/4) \times 2T = \widetilde{\mathcal{O}}\left(t_{{\rm mix},{\Phi_\Gamma}}(\epsilon/4)\frac{\sigma^2}{\Gamma^2}\right) = \widetilde{\mathcal{O}}\left(\frac{\beta^{{4}}t^3_{{\rm mix},{\Phi_\Gamma}}(\epsilon/4)}{{\Gamma^2}\epsilon^2}\right)\,,
\]
where we use $\sigma=\widetilde{\Theta}(\beta^{{2}} t_{{\rm mix},{\Phi_\Gamma}}(\epsilon/4)/\epsilon)$ from~\cref{thm:main_rigor_thermal} in the last equality. Therefore, to characterize the end-to-end complexity of the system-bath interaction algorithm, it suffices to upper bound the rescaled mixing time $t_{{\rm mix},{\Phi_\Gamma}}(\epsilon)$ and lower bound the product ${\Gamma^2}$.

Now, we are ready to state our main perturbation result, which gives an upper bound on the rescaled mixing time $t_{{\rm mix},{\Phi_\Gamma}}(\epsilon)$ beyond the Lindblad limit
\begin{thm}\label{thm:main_result_gap} Given $0\leq \beta<\infty$. We choose $g(\omega)$ as in~\cref{thm:KMS_DBC_approx}~\cref{eqn:g_choice}. Let $\mc{L}_{\sigma, \rm KMS}$ be the Lindbladian defined in~\cref{eqn:KMS_Lindbladian}. If  ${\Gamma^2}=\mathcal{O}(\lambda_{\rm gap}(\mc{L}_{\sigma, \rm KMS}))$, $\sigma=\widetilde{\Omega}\left(\beta{^2} \lambda^{-1}_{\rm gap}(\mc{L}_{\sigma, \rm KMS})\right)$, and $T=\widetilde{\Omega}(\sigma)$, then for any $\epsilon>0$, we have
\[
t_{\rm mix,\Phi_\Gamma}(\epsilon)=\widetilde{\mathcal{O}}\left(\frac{1}{\lambda_{\rm gap}{(\mc{L}_{\sigma, \rm KMS})}}\log\left(\frac{\left\|\sigma^{-1/2}_\beta\right\|}{\epsilon}\right)\right)\,.
\]
\end{thm}

A detailed proof is provided in \cref{sec:proof_main_result_gap}. A potential concern regarding this theorem is the dependence of $\lambda_{\rm gap}$ on $\sigma$: if the gap decreases as $\sigma$ increases, the condition $\sigma=\widetilde{\Omega}(\beta{^2} \lambda^{-1}_{\rm gap})$ might fail. Fortunately,~\cite{ding2025endtoendefficientquantumthermal} demonstrated that the spectral gap of $\mathcal{L}_{\rm KMS}$ remains non-vanishing as $\sigma \to +\infty$ for simple physical models.~\cite{slezak2026polynomialtime} subsequently generalized this to a much broader class of models, proving that with the $g(\omega)$ defined in \cref{thm:KMS_DBC_approx}~\cref{eqn:g_choice}, the spectral gap is monotonically non-decreasing with respect to $\sigma$. We extend these results by combining this monotonicity with \cref{thm:main_result_gap} to demonstrate that the system-bath interaction algorithm achieves fast mixing times even beyond the Lindblad limit. Crucially, our analysis framework is highly general; it can be integrated with \textbf{all existing} spectral gap estimates for KMS generators to derive upper bounds on the mixing time beyond the Lindblad limit.

\subsection{Total Runtime Bounds Beyond the Lindblad Limit}\label{sec:application}
Combining~\cref{thm:main_result_gap} with the exising literature of spectral gap estimation for KMS generators~\cite{rouz2024,tong2024fast,smid2025polynomial,smid2025rapidmixingquantumgibbs,rouze2024optimal,bergamaschi2026fastmixingquantumspin}, we can derive explicit total runtime bounds for the system-bath interaction algorithm beyond the Lindblad limit for various physical models~\footnote{Some of the result here prove the rapid mixing property for KMS generators using oscillator norm technique instead of spectral gap. However, rapid mixing directly implies a lower bound for the spectral gap. Thus, we can directly apply their result in our case.}. We summarize the models and present our results below:
\begin{itemize}
\item High temperature local spin Hamiltonian~\cite{rouz2024,rouze2024optimal}:

Let $H = \sum_i h_i$ be a local Hamiltonian defined on a $D$-dimensional lattice, where each local term $h_i$ is supported on a ball of constant radius. Furthermore, each qubit $j$ is acted upon by only a constant number of terms $h_i$. Assume the system coupling operators are chosen as $\mc{A}=\{\pm X_j,\pm Y_j,\pm Z_j\}_{j=1}^N$. Then, there exists a constant $\beta_c$ that only depends on the locality of $H$ such that for any $0<\beta<\beta_c$ and $\sigma>\beta$, we have $\lambda_{\rm gap}(\mc{L}_{\sigma,\rm KMS})=\Omega(1/N)$~\footnote{Although it has not been explicitly proved in the literature, we expect similar results should also hold for local fermionic Hamiltonians using similar techniques as in~\cite{rouz2024,rouze2024optimal,tong2024fast,smid2025polynomial}.}.

\item Weak interaction fermionic system at all temperatures~\cite{smid2025rapidmixingquantumgibbs}:

Let a local fermionic Hamiltonian $H$ defined on a $D$-dimensional lattice of fermionic systems, $\Lambda = [0,L]^D$, given by
\begin{equation}\label{eqn:H_fermion}
H = H_0+ H_1= \sum_{ij} M_{i,j} c^\dagger_i c_j + \varepsilon \sum_j h_j, \quad \|h_j\|\leq 1,
\end{equation}
where $(M_{i,j})$ is a Hermitian matrix, and $c^\dagger_j$ and $c_j$ are the creation and annihilation operators at site $j$. The terms $\{h_j\}$ are local fermionic perturbations and are parity preserving, meaning that each $h_j$ contains an even number of creation and annihilation operators. We further assume that $H_0$ is $(1,l)$-geometrically local and $\sum_j h_j$ are $(r_0,l)$-geometrically local. Specifically, each term in $H_0$ is a product of fermionic operators acting on a set of sites whose Manhattan diameter is at most $1$, and each $h_j$ is a product of fermionic operators acting on a set of sites whose Manhattan diameter is at most $r_0$. In addition, each site $i$ appears in at most $l$ non-trivial $c^\dagger_i c_j$ and $h_j$ terms. Assume the system coupling operators are chosen as $\mc{A}=\{\pm c^\dagger_j,c_j\}_{j=1}^N$. For any fixed $\beta>0$, there exists a constant $\varepsilon_c=\Omega(1)$ that only depends on $r_0,l,D,\beta$ such that for any $0\leq \varepsilon<\varepsilon_c$ and $\sigma>\beta$, we have $\lambda_{\rm gap}(\mc{L}_{\sigma,\rm KMS})=\Omega(1/N)$.

\item Weak interaction spin system at all temperatures~\cite{tong2024fast,smid2025polynomial,smid2025rapidmixingquantumgibbs}:
Let a local Hamiltonian $H$ over a $D$-dimensional lattice of spin systems $\Lambda= [0,L]^D$ with the following form (the system size is $N=(L+1)^D$):
\begin{equation}\label{eqn:H}
H=H_0+H_1=-\sum_{i} Z_i+ \varepsilon \sum_{j} h_j, \quad \|h_j\|\leq 1.
\end{equation}
Here, $H_0 = -\sum_{i} Z_i$ is referred to as the non-interacting term because its indices do not overlap. The choice of $Z_i$ as the non-interacting term is made for convenience, and can be substituted with other simple, gapped local terms that also have non-overlapping indices. We assume the interacting term $H_1$ is an $(r_0,l)$-geometrically local Hamiltonian. A specific example of the Hamiltonian in \cref{eqn:H} is the $D$-dimensional TFIM model, which is a $(2,2D+1)$-local Hamiltonian. Assume the system coupling operators are chosen as $\mc{A}=\{\pm X_j,\pm Y_j,\pm Z_j\}_{j=1}^N$. For any fixed $\beta>0$, there exists a constant $\varepsilon_c=\Omega(1)$ that only depends on $r_0,l,D,\beta$ such that for any $0\leq \varepsilon<\varepsilon_c$ and $\sigma>\beta$, we have $\lambda_{\rm gap}(\mc{L}_{\sigma,\rm KMS})=\Omega(1/N)$.

\item 1D local Hamiltonian at all temperatures~\cite{bergamaschi2026fastmixingquantumspin}:
Let a 1D local Hamiltonian $H=\sum_{j=1}^{n} h_{j}$, where $\{h_j\}$ satisfies the same assumption as~\eqref{eqn:H} but is defined on a 1D lattice. Assume the system coupling operators are chosen as $\mc{A}=\{\pm X_j,\pm Y_j,\pm Z_j\}_{j=1}^N$. Then, for any $0\leq \beta<\infty$ and $\sigma>\beta$, we have $\lambda_{\rm gap}(\mc{L}_{\sigma,\rm KMS})=\Omega(1/N)$.

\end{itemize}

Using the above spectral gap estimations, {$\log\left(\frac{\|\sigma_\beta^{-1/2}\|}{\epsilon}\right) = \mathcal{O}(N)$ in~\cite{slezak2026polynomialtime}} and~\cref{thm:main_result_gap}, it is straightforward to derive the following total runtime bounds for the system-bath interaction algorithm beyond the Lindblad {dynamics}:
\begin{cor}[Total runtime bound beyond Lindblad limit]\label{cor:total_run_time} For all the above models with proper choice of $\beta$, by choosing {$\Gamma=\Theta(\sqrt{1/N})$}, $\sigma=\widetilde{\Theta}(N^2/\epsilon)$, and $T=\widetilde{\Theta}(\sigma)$, we have $t_{\rm mix,\Phi_\Gamma}(\epsilon)=\widetilde{\mathcal{O}}(N^2)$. Therefore, the system-bath interaction algorithm can prepare an $\epsilon$-approximation of the thermal state $\sigma_\beta$ with total Hamiltonian simulation time
\[
T_{\rm total}=\widetilde{\mathcal{O}}\left(\frac{N^7}{\epsilon^2}\right)\,.
\]
\end{cor}
\begin{rem}\label{re:improvement}
The above theorem is a direct generalization of the end-to-end complexity results in~\cite{slezak2026polynomialtime} beyond the Lindblad limit. Because $\Gamma$ is not required to be small, the above result saves a factor of $N^3/\epsilon^2$ in total runtime compared to the previous best result in~\cite[Theorem 7]{slezak2026polynomialtime}. On the other hand, we believe that the dependence on $N$ in the total runtime is far from optimal and can be further improved in the case of different models. There are several directions of the improvement:
\begin{itemize}
  \item Relaxing the requirement of ${\Gamma^2}=\mathcal{O}(\lambda_{\rm gap})$ in~\cref{thm:main_result_gap}. Currently, \cref{thm:main_result_gap} requires ${\Gamma^2}=\mathcal{O}(1/N)$. In our perturbation framework, this strict condition is necessary to control the gap perturbation caused by the higher-order term $\mathcal{M}_{{\geq}2}$. Although $\mathcal{M}_{{\geq}2}$ preserves the thermal state, its dominant term tends to shrink the spectral gap. Since the norm of $\mathcal{M}_{{\geq}2}$ scales as $\mathcal{O}({\Gamma^2})$, ${\Gamma^2}$ must be sufficiently small relative to the inverse gap to prevent the gap from closing. However, our perturbation analysis does not currently exploit the fact that the terms constituting $\mathcal{M}_{{\geq}2}$ are quasi-local when $A_S$ is local and $H$ is local Hamiltonian. When $\sigma$ is not large, these terms produce only constant perturbations to the local gaps. Existing spectral gap analyses often rely on deriving a global gap from local gaps~\cite{kastoryano2016quantum,bergamaschi2026fastmixingquantumspin}; this approach appears compatible with the quasi-local structure of $\mathcal{M}_{{\geq}2}$. If such a derivation can be extended to this context, it would suffice to bound the local gap perturbation. In that case, the effect of $\mathcal{M}_{{\geq}2}$ would be controlled by ${\Gamma^2}/N$, meaning the relaxed condition ${\Gamma^2}=\mathcal{O}(1)$ should suffice to guarantee a non-vanishing global gap.
  \item Relaxing the dependence of $\log\left(\left\|\sigma^{-1/2}_\beta\right\|/\epsilon\right)$. In~\cref{cor:total_run_time}, the factor of $N^3$ arises from the dependence on the initial condition, a known limitation of mixing time analysis based solely on the spectral gap. For certain systems where the Davies generator can be efficiently approximated, it might be possible to improve this dependence to $\log(N)$ using techniques like Modified Logarithmic Sobolev Inequalities (MLSI)~\cite{TemmeKastoryanoRuskaiEtAl2010,BardetCapelGaoEtAl2023}. The main technical challenge is to extend the entropy decay analysis from continuous Lindbladian dynamics to the discrete channel $\Phi_\Gamma$, which prohibits us from directly taking derivative of the entropy functional. Another possible way to improve the dependence on the initial condition is to use the oscillator norm technique developed in~\cite{smid2025rapidmixingquantumgibbs,rouze2024optimal,zhan2025rapidquantumgroundstate}. This technique can directly boundthe convergence in observables without going throughthe spectral gap, thus avoidingthe dependence onthe initial condition. However, even inthe Lindblad limit regime, the forward unitary map $\mc{U}_S$ and nonvanishing coherent part $-i[H_{\rm Lamb},\rho]$ prohibits us from directly applying this technique. We leave this as an interesting future direction to investigate.

  \item Relaxing the interaction strength requirement to strong coupling regime. In our analysis, we require ${\Gamma}=\Theta(1)$. Although this goes beyond the Lindblad limit, the interaction strength is still relatively weak since ${\Gamma/\sigma}\ll 1$. Surprisingly, our numerical results suggest that even in the strong coupling regime where ${\Gamma/\sigma}=\Omega(1)$, the fixed point of $\Phi_\Gamma$ remains close to the thermal state, and the mixing time still decreases with ${\frac{\sigma}{\Gamma^2}}$ factor. It would be an interesting future direction to theoretically understand the performance of the system-bath interaction algorithm in this strong coupling regime. This can potentially further reduce the total runtime.
\end{itemize}
\end{rem}
\subsection{Proof of \texorpdfstring{\cref{thm:main_result_gap}}{lg}}\label{sec:proof_main_result_gap}

% Recall~\eqref{eqn:M2_def}, we define
% \[
% {\widetilde{\mc{M}}_{2}}(\rho) = \mathbb{E}_{A_S}\left(\sum_{n\geq 2} \Gamma^{2n} (-1)^n \sum_{k = 0}^{2n}(-1)^k \int \gamma(\omega)\widetilde{G}_{2n-k,A_S}^\dag(\omega) \rho  \widetilde{G}_{k,A_S}(\omega)\mathrm{d}\omega\right)
% \]
% with $\widetilde{G}_{k,A_S}(\omega)$ in~\cref{eq:tildeGk}.
We first provide an upper bound for the norm of $\widetilde{\mc{M}}_{{\geq 2}}${, which is defined in~\eqref{eqn:Phi_alpha_limit_rewrite}:}
\begin{lem}\label{lem:M2_bound}
Assuming {$\Gamma = \mathcal{O}(1)$}, we have
\begin{equation}\label{eqn:M2_bound}
\|\widetilde{\mc{M}}_{{\geq 2}}\|_{2\leftrightarrow 2} =\mathcal{O}\left(\Gamma^2\right)\,,
\end{equation}
\begin{equation}\label{eqn:M2_bound_1}
\|\sigma^{-1/4}_\beta\widetilde{\mc{M}}_{{\geq 2}}[\sigma^{1/4}_\beta\cdot \sigma^{1/4}_\beta]\sigma^{-1/4}_\beta\|_{2\leftrightarrow 2} =\mathcal{O}\left(\Gamma^2\right)\,,
\end{equation}
\end{lem}
\begin{proof}[Proof of~\cref{lem:M2_bound}] We first prove~\eqref{eqn:M2_bound}. We fix $n,k$ in $\widetilde{\Phi}_{\rm part, \Gamma}$ and take an arbitrary matrix $\rho$ (not necessary density operator). Following {\cref{lm:1} when $\sigma_0=I$}, we have
\[
\begin{aligned}
&\int \gamma(\omega)\widetilde{G}_{2n-k,A_S}^\dag(\omega) \rho  \widetilde{G}_{k,A_S}(\omega)\mathrm{d}\omega\\
=&{\frac{1}{\pi^{n/2}}}\int_{\substack{-\infty< s_1\leq s_2 \leq \cdots \leq s_{2n-k}<\infty,\\ -\infty < t_1 \leq t_2 \leq \cdots \leq t_k <\infty}}A_S(2\sigma s_{2n-k}) \cdots A_S(2\sigma  s_1)\rho A_S(2\sigma  t_1) \cdots A_S(2\sigma t_k) e^{-\sum_{p=1}^{2n-k}s_p^2 - \sum_{p=1}^k t_p^2}\\
&\cdot\int \gamma(\omega) e^{i2\sigma\omega \sum_{p = 1}^{2n-k}(-1)^p s_p - i2\sigma\omega \sum_{p= 1}^k (-1)^pt_p}\mathrm{d}\omega\mathrm{d}s_1 \cdots\mathrm{d}s_{2n-k}\mathrm{d}t_1 \cdots \mathrm{d}t_k
\end{aligned}
\]
Because $\|BAC\|_2\leq \|B\|\|C\|\|A\|_2$, we have
\[
\begin{aligned}
&\left\|\int \gamma(\omega)\widetilde{G}_{2n-k,A_S}^\dag(\omega) \rho  \widetilde{G}_{k,A_S}(\omega)\mathrm{d}\omega\right\|_2
\leq \frac{1}{\pi^{n/2}}\left\|A_S\right\|^{2n}\|\rho\|_2\int_{\substack{-\infty< s_1\leq s_2 \leq \cdots \leq s_{2n-k}<\infty,\\ -\infty < t_1 \leq t_2 \leq \cdots \leq t_k <\infty}} e^{-\sum_{p=1}^{2n-k}s_p^2 - \sum_{p=1}^k t_p^2}\\
&\cdot\left|\int \gamma(\omega) e^{i2\sigma\omega \sum_{p = 1}^{2n-k}(-1)^p s_p - i2\sigma\omega \sum_{p= 1}^k (-1)^pt_p}\mathrm{d}\omega\right|\mathrm{d}s_1 \cdots\mathrm{d}s_{2n-k}\mathrm{d}t_1 \cdots \mathrm{d}t_k
\leq  \frac{(\Theta(1))^{n}}{\sigma\sqrt{n}}\mathcal{O}\left( \log(\sqrt{2n}\sigma)\|\gamma'\|_{L^1}\right)
\end{aligned}
\]
where we use~\eqref{eq:finalbdTerm12} in the last inequality. Plugging this back into $\widetilde{\mc{M}}_{\geq 2} = \frac{\sigma}{\Gamma^2}\sum_{n\geq 2} \Gamma^n\widetilde{\mc N}_n$ gives
\[
\begin{aligned}
&\|\widetilde{\mc{M}}_{{\geq 2}}\|_{2\leftrightarrow 2} =\mathcal{O}\left(\sum_{n\geq 2}(\Theta(\Gamma^2))^{n-1}\sqrt{n}\mathcal{O}\left( \log(\sqrt{2n}\sigma)\|\gamma'\|_{L^1} \right)\right)
=\Gamma^2\underbrace{\mathcal{O}\left(\sum_{n\geq 2}(\Theta(\Gamma^2))^{n-2}\sqrt{n}\mathcal{O}\left( \log(\sqrt{2n}\sigma)\|\gamma'\|_{L^1} \right)\right)}_{=\mathcal{O}(1)\text{ when }\Gamma=\mathcal{O}(1)}
\end{aligned}
\]
This concludes the proof of~\eqref{eqn:M2_bound}.
To prove~\eqref{eqn:M2_bound_1}, we use~\cref{lm:1} to get the formula of $\sigma^{-1/4}_\beta G\sigma^{1/4}_\beta$. The rest of the proof is similar to the above calculation.
\end{proof}

The following perturbation framework is inspired by~\cite{ChenKastoryanoBrandaoEtAl2023,ChenKastoryanoGilyen2023,Fag_2007}.

Define
\[
{\widehat{\Phi}_{\rm part,\Gamma}}=I+\frac{\Gamma^2}{\sigma}{\widehat{\mathcal{M}}},\quad {\widehat{\mathcal{M}}}=\mathcal{L}_{\sigma,\rm KMS}-i[{H_{C}},\cdot]+{\widetilde{\mathcal{M}}_{\geq 2}}\,.
\]
% Comparing with $\Phi_{\rm part,\Gamma}$, $\widetilde{\Phi}_{\rm part,\Gamma}(\rho)$ takes $T=\infty$ and modify the Lindbladian part to the KMS Lindbladian.
According to~\cref{lm:Tinf} and~\cref{thm:KMS_DBC_approx}, we have
\[
\left\|{\widehat{\Phi}_{\rm part,\Gamma}}-\Phi_{\rm part,\Gamma}\right\|_{1\to 1}=\mathcal{O}\left(\frac{\Gamma^2\beta^{{2}}}{\sigma^2}+\frac{\Gamma^2\sigma}{T}\exp(-T^2/(4\sigma^2))\right)\,.
\]
Define ${\widehat{\Phi}_{\Gamma}}=\mc{U}_S(T)\circ {\widehat{\Phi}_{\rm part,\Gamma}}\circ \mc{U}_S(T)$, we have
\[
\left\|{\widehat{\Phi}_{\Gamma}}-\Phi_\Gamma\right\|_{1\to 1}=\mathcal{O}\left(\frac{\Gamma^2\beta^{{2}}}{\sigma^2}+\frac{\Gamma^2\sigma}{T}\exp(-T^2/(4\sigma^2))\right)\,.
\]
Given $\epsilon>0$,
according to~\cite[Theorem 8]{ding2025endtoendefficientquantumthermal},
\begin{equation}\label{eqn:mixing_time_approx}
\frac{\sigma t_{\rm mix,{\widehat{\Phi}_{\Gamma}}}(\epsilon/2)}{\Gamma^2}\left\|{\widehat{\Phi}_{\Gamma}}-\Phi_\Gamma\right\|_{1\to 1}\leq \epsilon/2\Rightarrow t_{\rm mix,\Phi_\Gamma}(2\epsilon)\leq t_{\rm mix,{\widehat{\Phi}_{\Gamma}}}(\epsilon/2)\,.
\end{equation}
Thus, to give an upper bound for $t_{\rm mix,\Phi_\Gamma}(\epsilon)$, it suffices to give an upper bound for $t_{\rm mix,{\widehat{\Phi}_{\Gamma}}}(\epsilon)$. Specifically,
\[
\sigma=\Omega\left(\frac{\beta^{{2}} t_{\rm mix,{\widehat{\Phi}_{\Gamma}}}(\epsilon)}{\epsilon}\right),\quad T=\Omega\left(\sigma\sqrt{\log\left(\frac{\sigma t_{\rm mix,{\widehat{\Phi}_{\Gamma}}}(\epsilon)}{\epsilon}\right)}\right)\Rightarrow t_{\rm mix,\Phi_\Gamma}(4\epsilon)\leq t_{\rm mix,{\widehat{\Phi}_{\Gamma}}}(\epsilon)\,.
\]
Here, the form of $\sigma$ looks similar to that in~\cref{thm:main_rigor_thermal}, thus, no extra $\epsilon$-dependence is introduced in the final choice of $\sigma$ by this approximation step.

Now, we focus on upper bounding $t_{\rm mix,{\widehat{\Phi}_{\Gamma}}}(\epsilon)$. First,
\[
\mathcal{K}\left(\sigma_\beta,\widehat{\mathcal{M}}\right)=\mathcal{K}\left(\sigma_\beta,\mathcal{L}_{\sigma,\rm KMS}\right)+\mathcal{K}\left(\sigma_\beta,-i[H_{C},\cdot]\right)+\mathcal{K}\left(\sigma_\beta,\widetilde{\mathcal{M}}_{{\geq 2}}\right)\,,
\]
we decompose each superoperator into Hermitian and anti-Hermitian parts, $\mathcal{K} = \mathcal{H} + \mathcal{A}$, and analyze the properties of each Hermitian superoperator $\mathcal{H}_1, \mathcal{H}_2, \mathcal{H}_3$.
\begin{itemize}
    \item According to~\cref{thm:KMS_DBC_approx},
\[
\mathcal{K}\left(\sigma_\beta,\mathcal{L}_{\sigma,\rm KMS}\right)=\mathcal{K}^\dagger\left(\sigma_\beta,\mathcal{L}_{\sigma,\rm KMS}\right)=\mathcal{H}\left(\sigma_\beta,\mathcal{L}_{\sigma,\rm KMS}\right):=\mathcal{H}_1\,.
\]
\item The Hermitian part of $\mathcal{K}\left(\sigma_\beta,-i[H_{C},\cdot]\right)$ follows
\[
\left\|\underbrace{\mathcal{H}\left(\sigma_\beta,-i[H_{C},\cdot]\right)}_{:=\mathcal{H}_2}\right\|_{2\leftrightarrow 2}=\left\|-\frac{i}{2}\left\{\sigma^{-1/4}_{\beta}H_{C} \sigma^{1/4}_{\beta}-\sigma^{1/4}_{\beta}H_{C} \sigma^{-1/4}_{\beta},\rho\right\}\right\|_{2\leftrightarrow 2}=\mathcal{O}(\beta^{{2}}/\sigma)\,.
\]
\item Furthermore, according to~\cref{lem:M2_bound}~\cref{eqn:M2_bound_1}, we have \[
\left\|\underbrace{\mathcal{H}\left(\sigma_\beta,\widetilde{\mathcal{M}}_{{\geq 2}}\right)}_{:=\mathcal{H}_3}\right\|_{2\leftrightarrow 2}\leq \left\|\mathcal{K}\left(\sigma_\beta,\widetilde{\mathcal{M}}_{{\geq 2}}\right)\right\|_{2\leftrightarrow 2}=\mathcal{O}({\Gamma^2})\,.
\]
\end{itemize}

% and

Because $\mathcal{H}_1$ is Hermitian with respect to Hilbert–Schmidt inner product and $\mathcal{H}_1(\sqrt{\sigma_\beta})=0$, we can define the spectral gap of $\mathcal{H}_1$ as follows: $\lambda_{\rm gap}(\mc{H}_1):=\inf_{\mathrm{Tr}(X\sqrt{\sigma_\beta})=0,X\neq 0}\frac{-\left\langle X,\mathcal{H}_1(X)\right\rangle_2}{\left\langle X,X\right\rangle_2}\,.$
It is straightforward to check that this gap matches with the spectral gap of $\mathcal{L}_{\rm KMS}$ in the KMS inner product:
\[
\begin{aligned}
&{\lambda_{\rm gap}}(\mc{L}_{\sigma, \rm KMS})=\inf_{\mathrm{Tr}(A\sigma_\beta)=0,A\neq 0}\frac{-\left\langle A,\mathcal{L}^\dagger_{\rm KMS}(A)\right\rangle_{1/2,\sigma_\beta}}{\left\langle A,A\right\rangle_{1/2,\sigma_\beta}}
=\inf_{\mathrm{Tr}(A\sigma_\beta)=0,A\neq 0}\frac{-\left\langle \sigma^{1/4}_\beta A\sigma^{1/4}_\beta,\mathcal{K}^\dagger(\sigma_\beta,\mathcal{L}_{\rm KMS})(\sigma^{1/4}_\beta A\sigma^{1/4}_\beta)\right\rangle}{\left\langle \sigma^{1/4}_\beta A\sigma^{1/4}_\beta,\sigma^{1/4}_\beta A\sigma^{1/4}_\beta\right\rangle}\\
=&\inf_{\mathrm{Tr}(\sigma^{1/4}_\beta A\sigma^{1/4}_\beta\sqrt{\sigma_\beta})=0,A\neq 0}\frac{-\left\langle \sigma^{1/4}_\beta A \sigma^{1/4}_\beta,\mathcal{H}_1(\sigma^{1/4}_\beta A \sigma^{-1/4}_\beta)\right\rangle}{\left\langle \sigma^{1/4}_\beta A \sigma^{1/4}_\beta,\sigma^{1/4}_\beta A \sigma^{1/4}_\beta\right\rangle}=\lambda_{\rm gap}(\mc{H}_1)\,,
\end{aligned}
\]
where $\langle A,B\rangle=\mathrm{Tr}\left(A^\dagger B\right)$ is the Hilbert-Schemitz inner product. In the third equality, we use that, if $\mathrm{Tr}(A\sigma_\beta)=0$, then $\sigma^{-1/4}_\beta A \sigma^{-1/4}_\beta$ is orthogonal to $\sqrt{\sigma_\beta}$ under Hilbert Schemitz inner product. The spectral gap of $\mathcal{H}_1$ is same as the spectral gap of $\mathcal{L}_{\rm KMS}$, which has been well studied in the framework in~\cite{rouz2024,tong2024fast,smid2025polynomial,smid2025rapidmixingquantumgibbs,rouze2024optimal,bergamaschi2026fastmixingquantumspin}.

Now, we are ready to prove~\cref{thm:main_result_gap}.

\begin{proof}[Proof of~\cref{thm:main_result_gap}]

  Given any density operator $\rho_1,\rho_2$, we define $\mathcal{E}=\rho_1-\rho_2$. We consider the change of $\left\|\sigma_\beta^{-1/4}\mathcal{E}\sigma_\beta^{-1/4}\right\|_2$ after applying ${\widehat{\Phi}_{\Gamma}}$. First,
  {since $[H, \sigma_\beta] = 0$, we have}
\begin{equation}\label{eqn:Phi_infty_norm}
\left\|\sigma_\beta^{-1/4}{\widehat{\Phi}_{\Gamma}}(\mathcal{E})\sigma_\beta^{-1/4}\right\|_2=\left\|\sigma_\beta^{-1/4}(I+\widehat{\mc{M}}{\frac{\Gamma^2}{\sigma}})(\mc{U}_S(\mathcal{E}))\sigma_\beta^{-1/4}\right\|_2=\left\|(I+\mc{K}(\sigma_\beta,\widehat{\mc{M}}){\frac{\Gamma^2}{\sigma}})\left[\sigma^{-1/4}_\beta \mc{U}_S(\mathcal{E})\sigma^{-1/4}_\beta\right]\right\|_2\,.
\end{equation}
Let ${\mc{E}_1}=\mc{U}_S(\mathcal{E})$. Then with $\mathcal{K} = \mc{K}(\sigma_\beta,\widehat{\mc{M}})$,
\[
\left\|(I+\mc{K}(\sigma_\beta,\widehat{\mc{M}}){\frac{\Gamma^2}{\sigma}})\left[\sigma^{-1/4}_\beta {\mc{E}_1}\sigma^{-1/4}_\beta\right]\right\|^2_2
=\left\langle \sigma^{-1/4}_\beta {\mc{E}_1}\sigma^{-1/4}_\beta,\left(I+2(\mc{H}_1+\mc{H}_2+\mc{H}_3){\frac{\Gamma^2}{\sigma}}+\mc{K}^\dagger \mc{K}{\frac{\Gamma^4}{\sigma^2}}\right)\sigma^{-1/4}_\beta {\mc{E}_1}\sigma^{-1/4}_\beta\right\rangle_2\,.
\]
Because ${\mc{E}_1}$ is traceless, we have $\sigma^{-1/4}_\beta {\mc{E}_1}\sigma^{-1/4}_\beta$ is orthogonal to $\sqrt{\sigma_\beta}$ under Hilbert Schemitz inner product. This implies that
\[
\left\langle \sigma^{-1/4}_\beta {\mc{E}_1}\sigma^{-1/4}_\beta,\left(I+2\mc{H}_1{\frac{\Gamma^2}{\sigma}}\right)\sigma^{-1/4}_\beta {\mc{E}_1}\sigma^{-1/4}_\beta\right\rangle_2\leq (1-2\lambda_{\rm gap}(\mc{H}_1){\frac{\Gamma^2}{\sigma}})\left\|\sigma^{-1/4}_\beta {\mc{E}_1}\sigma^{-1/4}_\beta\right\|^2_2\,.
\]
For the other terms, we have
\[
\begin{aligned}
&\left\langle \sigma^{-1/4}_\beta {\mc{E}_1}\sigma^{-1/4}_\beta,\left(2(\mc{H}_2+\mc{H}_3){\frac{\Gamma^2}{\sigma}}+\mc{K}^\dagger \mc{K}{\frac{\Gamma^4}{\sigma^2}}\right)\sigma^{-1/4}_\beta {\mc{E}_1}\sigma^{-1/4}_\beta\right\rangle_2\\
\leq &\left(2\|\mc{H}_2\|_{2\to 2}{\frac{\Gamma^2}{\sigma}}+2\|\mc{H}_3\|_{2\to 2}{\frac{\Gamma^2}{\sigma}}+\left\|\mc{K}\right\|^{{2}}_{2\to 2}{\frac{\Gamma^4}{\sigma^2}}\right)\left\|\sigma^{-1/4}_\beta {\mc{E}_1}\sigma^{-1/4}_\beta\right\|^2_2\\
=&\mathcal{O}\left(\left(\frac{\beta^{{2}}}{\sigma}+{\frac{\Gamma^2}{\sigma}}\sigma+{\Gamma^4}\right){\frac{\Gamma^2}{\sigma}}\left\|\sigma^{-1/4}_\beta {\mc{E}_1}\sigma^{-1/4}_\beta\right\|^2_2\right)\,.
\end{aligned}
\]
Putting the above two bounds together, we have
\[
\left\|(I+{\mc{K}(\sigma_\beta,\widehat{\mc{M}})}{\frac{\Gamma^2}{\sigma}})\left[\sigma^{-1/4}_\beta {\mc{E}_1}\sigma^{-1/4}_\beta\right]\right\|^2_2=\left(1-2\lambda_{\rm gap}(\mc{H}_1){\frac{\Gamma^2}{\sigma}}+\mathcal{O}\left(\left(\frac{\beta^{{2}}}{\sigma}+{\Gamma^2}\right){\frac{\Gamma^2}{\sigma}}\right)\right)\left\|\sigma^{-1/4}_\beta {\mc{E}_1}\sigma^{-1/4}_\beta\right\|^2_2\,.
\]
Plugging this back into~\eqref{eqn:Phi_infty_norm}, we have
\[
\begin{aligned}
&\left\|\sigma_\beta^{-1/4}\widehat{\Phi}_{\Gamma}(\mathcal{E})\sigma_\beta^{-1/4}\right\|_2 = \left(1-\lambda_{\rm gap}(\mc{H}_1){\frac{\Gamma^2}{\sigma}}+\mathcal{O}\left(\left(\frac{\beta^{{2}}}{\sigma}+{\Gamma^2}\right){\frac{\Gamma^2}{\sigma}}\right)\right)\left\|\sigma^{-1/4}_\beta {\mc{E}_1}\sigma^{-1/4}_\beta\right\|_2\\
=&\left(1-\lambda_{\rm gap}(\mc{H}_1){\frac{\Gamma^2}{\sigma}}+\mathcal{O}\left(\left(\frac{\beta^{{2}}}{\sigma}+{\Gamma^2}\right){\frac{\Gamma^2}{\sigma}}\right)\right)\left\|\sigma^{-1/4}_\beta \mc{E}\sigma^{-1/4}_\beta\right\|_2
\leq \left(1-\frac{\lambda_{\rm gap}(\mc{H}_1)}{2}{\frac{\Gamma^2}{\sigma}}\right)\left\|\sigma^{-1/4}_\beta \mc{E}\sigma^{-1/4}_\beta\right\|_2\,,
\end{aligned}
\]
where we use the condition of $\sigma=\Omega(\beta^{{2}}/\lambda_{\rm gap}(\mc{H}_1))$ and ${\Gamma^2}=\mc{O}\left(\lambda_{\rm gap}(\mc{H}_1)\right)$ in the last inequality.

This implies that after $k$ iterations of ${\widehat{\Phi}_{\Gamma}}$,
\[
\left\|\sigma_\beta^{-1/4}{\widehat{\Phi}_{\Gamma}}^k(\mathcal{E})\sigma_\beta^{-1/4}\right\|_2\leq \left(1-\frac{\lambda_{\rm gap}(\mc{H}_1)}{2}{\frac{\Gamma^2}{\sigma}}\right)^k\left\|\sigma^{-1/4}_\beta \mc{E}\sigma^{-1/4}_\beta\right\|_2\,.
\]
{Finally, using the H\"older's inequality of Schatten norm, $\|BAB\|_1\leq \|B\|^2_4\|A\|_2$, we have
\[
\|\mathcal{E}\|_1\leq \left\|\sigma_\beta^{-1/4}\mathcal{E}\sigma_\beta^{-1/4}\right\|_2\leq \left\|\sigma^{-1/2}_\beta\right\|\|\mathcal{E}\|_1\,.
\]}
% Finally, using $\|BAB\|_1\leq \|B\|^2_4\|A\|_2$, we have
% \[
% \|\mathcal{E}\|_1\leq \left\|{\sigma}^{1/4}_\beta\right\|^2_4\left\|\sigma_\beta^{-1/4}\mathcal{E}\sigma_\beta^{-1/4}\right\|_2=\left\|\sigma_\beta^{-1/4}\mathcal{E}\sigma_\beta^{-1/4}\right\|_2\leq \left\|\sigma^{-1/4}_\beta\right\|^2\|\mathcal{E}\|_2\leq \left\|\sigma^{-1/2}_\beta\right\|\|\mathcal{E}\|_1\,.
% \]
This implies
\[
\left\|{\widehat{\Phi}_{\Gamma}}^k(\mathcal{E})\right\|_1\leq \left(1-\frac{\lambda_{\rm gap}(\mc{H}_1)}{2}{\frac{\Gamma^2}{\sigma}}\right)^k\left\|\sigma_\beta^{-1/4}\mathcal{E}\sigma_\beta^{-1/4}\right\|_2\leq \left(1-\frac{\lambda_{\rm gap}(\mc{H}_1)}{2}{\frac{\Gamma^2}{\sigma}}\right)^k\left\|\sigma^{-1/2}_\beta\right\|\|\mathcal{E}\|_1\,.
\]
This implies that to ensure $\left\|{\widehat{\Phi}_{\Gamma}}^k(\mathcal{E})\right\|_1\leq \epsilon$,
it suffices to choose $k=\mathcal{O}\left(\frac{1}{{\frac{\Gamma^2}{\sigma}}\lambda_{\rm gap}(\mc{H}_1)}\log\left(\frac{\left\|\sigma^{-1/2}_\beta\right\|}{\epsilon}\right)\right)\,.$
Thus, the rescaled mixing time of ${\widehat{\Phi}_{\Gamma}}$ is $t_{\rm mix,{\widehat{\Phi}_{\Gamma}}}(\epsilon)={\frac{\Gamma^2}{\sigma}} k=\mathcal{O}\left(\frac{1}{\lambda_{\rm gap}(\mc{H}_1)}\log\left(\frac{\left\|\sigma^{-1/2}_\beta\right\|}{\epsilon}\right)\right)\,.$
Putting this back to the earlier approximation step in~\eqref{eqn:mixing_time_approx}, we have
\[
t_{\rm mix,\Phi_\Gamma}(4\epsilon)=\mathcal{O}\left(\frac{1}{\lambda_{\rm gap}(\mc{H}_1)}\log\left(\frac{\left\|\sigma^{-1/2}_\beta\right\|}{\epsilon}\right)\right)
\]
with $\sigma=\Omega\left(\frac{\beta^{{2}}}{\lambda_{\rm gap}(\mc{H}_1)}\right), T=\Omega\left(\sigma\sqrt{\log\left(\frac{\sigma}{\epsilon}\cdot \frac{1}{\lambda_{\rm gap}(\mc{H}_1)}\right)}\right), {\Gamma}=\mathcal{O}\left(\sqrt{\lambda_{\rm gap}(\mc{H}_1)}\right).$
This concludes the proof.
\end{proof}
\section{{Mixing time analysis for ground state preparation}}
\label{sec:mixing_time_ground}
We consider a free fermionic system on the $D$-dimensional lattice $[0,L]^D$. The Hamiltonian is
\begin{equation}
    H = \sum_{i,j= 1}^N h_{ij} a_i^\dag a_j,
\end{equation}
where $N = (L+1)^D$ is the number of fermionic modes, $(h_{ij})$ is a Hermitian matrix and $a_j^\dag, a_j$ are fermionic creation and annihilation operators.

Since $h$ is Hermitian, there exists a unitary matrix \(U\) such that
 $\diag(\mu_1,\cdots,\mu_N) = U^\dag h U$. Defining the new fermionic creation and annihilation operators $c_k
    =
    \sum_{j=1}^N U_{jk}^* a_j,
    c_k^\dagger
    =
    \sum_{j=1}^N U_{jk} a_j^\dagger,$
the Hamiltonian becomes
\begin{equation}
    H = \sum_{k=1}^N \mu_k c_k^\dagger c_k .
\end{equation}
The ground state energy is $\lambda_0
    =
    \sum_{k:\mu_k<0}\mu_k,$
corresponding to filling all negative energy modes. The spectral gap $\Delta = \min_{1\leq k\leq N: \mu_k\neq0}|\mu_k|$. Assuming for simplicity that there are no zero modes, the ground state is
unique. The excitation number operator relative to the ground state is $\hat N
    =
    \sum_{k:\mu_k>0} c_k^\dagger c_k
    +
    \sum_{k:\mu_k<0} c_k c_k^\dagger .$
Let $\ket{\psi_0}$ be the ground state such that $ \hat N\ket{\psi_0}=0$ and every excited eigenstate has excitation number at least one. Hence $I-\ketbra{\psi_0}
    \preceq
    \hat N.$

 By Fuchs-van de Graaf inequality, for any density matrix \(\rho\),
\begin{equation}
    \|\rho - \ketbra{\psi_0}\|_1 \leq 2\sqrt{1-\bra{\psi_0}\rho \ket{\psi_0}}\leq 2\sqrt{\Tr(\rho\hat{N})}.
\end{equation}
It is thus sufficient to prove a contraction estimate for the expected
excitation number, namely, $\Tr\left(\Phi_\Gamma(\rho)\hat N\right)
    = \Tr\left(\rho\Phi_\Gamma^\dagger(\hat N)\right)\leq(1-\delta)\Tr(\rho\hat N)$.

By applying the particle-hole transformation
\(d_k=c_k\) for \(\mu_k>0\) and \(d_k=c_k^\dagger\) for \(\mu_k<0\), the Hamiltonian can be written as $H=\lambda_0 I+\sum_k |\mu_k| d_k^\dagger d_k, \hat N=\sum_k d_k^\dagger d_k$,
% \[
% H=\lambda_0 I+\sum_k |\mu_k| d_k^\dagger d_k,
% \qquad
% \hat N=\sum_k d_k^\dagger d_k,
% \]
and since the scalar shift \(\lambda_0 I\) does not affect the Heisenberg evolution or the adjoint channel \(\Phi_\Gamma^\dagger\), it suffices to prove the contraction estimate in the positive energy case \(\mu_k>0\).

% We now reduce the general diagonal free fermion Hamiltonian to the positive energy case.
% The hamiltonian $H = \lambda_0I + \sum_{k:\mu_k>0} \mu_kc_k^\dag c_k- \sum_{k:\mu_k<0} \mu_kc_k c_k^\dag = \lambda_0 + H^+$ by the canonical anticommutation relations. Since the difference between \(H\) and \(H^+\) is only the scalar shift \(\lambda_0 I\), the Heisenberg evolution is unchanged, $e^{itH}\cdot e^{-itH} = e^{itH^+}\cdot e^{-itH^+}$. Thus, the expression of $\Phi_\Gamma^\dag$ is the same whether we use $H^+$ or $H$. We define the  particle-hole transformed quasiparticle modes $ d_k
%     \coloneqq
%     \begin{cases}
%     c_k, & \mu_k>0,\\
%     c_k^\dagger, & \mu_k<0.
%     \end{cases}$, then the number operator can be rewritten as $\hat{N} = \sum_k d_k^\dag d_k$.  Using these variables, $H  = \lambda_0 + \sum_k |\mu_k|d_k^\dag d_k$. The problem is reduced to prove the contraction estimate when all $\mu_k>0$ without loss of generality.

    Furthermore, from~\cref{lm:Tinf}, the finite time channel is close to its infinte window approximation as $\|\Phi_\Gamma - \widetilde{\Phi}_\Gamma\|_{1\rightarrow 1} = \mathcal{O}(\frac{\Gamma^2\sigma}{T}e^{-T^2/(4\sigma^2)})$. It is therefore sufficient to analyze the channel with the even function $\zeta(\omega) = g(\omega)+g(-\omega)$
\begin{equation}\label{eq:Phi_Nhat}
    \widetilde{\Phi}_\Gamma^\dagger(\hat{N})
         = \hat{N} + \mathbb{E}_{A_S}\Bigg(
\sum_{n\geq 1}\Gamma^{2n}(-1)^n
\sum_{k=0}^{2n}(-1)^k
\int_{-\infty}^{0}
\zeta(\omega)
U_S^\dagger(T)\widetilde{G}_{k,A_S}(\omega)\hat{N}\widetilde{G}_{2n-k,A_S}^\dagger(\omega)U_S(T)
\,d\omega
\Bigg)
\end{equation}
Here $\widetilde{G}_{k,A_S}(\omega)$ denotes the operator ${G}_{k,A_S}(\omega)$ in~\cref{eq:GF} by taking $T = \infty$.
\subsection{When $A_S \in\{\pm c_k^\dag, \pm c_k\}$}\label{sec:ground_mixing_time_1}
Throughout this subsection, assume that $\mu_j\geq \Delta>0$ for all $j$.
Then the ground state is the vacuum of the $c$-modes and $\hat N=\sum_{j=1}^N c_j^\dagger c_j .$
We show that one step of the channel decreases the expected excitation
number.
\begin{thm}
    Let $g(\omega) = \frac{1}{2\|h\|}1_{-2\|h\|, 0}$ then $\zeta(\omega) = \frac{1}{2\|h\|}1_{[-2\|h\|, 2\|h\|]}$.  Consider the gapped free fermionic system with spectral gap $\Delta>0$ and
    $A_S \in\{\pm c_k^\dag, \pm c_k\}$. For any $\epsilon>0$, if $\sigma = \widetilde{\Omega}(\Delta^{-1}\log(N/\epsilon))$, $T = \widetilde{\Omega}(\sigma) $ and coupling strength $\Gamma = \Theta(1)$, we have
    \begin{equation}
    t_{\rm mix, \Phi_\Gamma}(\epsilon) = \mathcal{O}\left(\frac{\|h\|\Gamma^2 N\log(2N/\epsilon)}{|\sin(\|f\|_{L_1}\Gamma)|^2}\right) = \mathcal{O}\left(\|h\|N\log^2(2N/\epsilon)\right).
    \end{equation}
\end{thm}

The adjoint quantum channel with $A_S \in\{\pm c_k^\dag, \pm c_k\}$ is applied to observable $O$
\begin{equation}\label{eq:adjointObserve}
    \begin{aligned}
        \widetilde{\Phi}_\Gamma^\dagger(O)
& = O + \mathbb{E}_{c_j}\Bigg(
\sum_{n\geq 1}\Gamma^{2n}(-1)^n
\sum_{k\text{ is even}, 0\leq k \leq 2n}(-1)^k
\int_{-\infty}^{0}
\zeta(\omega)
F_k(\omega-\mu_j) F_{2n-k}^*(\omega-\mu_j)\,d\omega
 U_S^\dagger(T)c_j c_j^\dag O  c_j c_j^\dag U_S(T)\Bigg)
\\
& + \mathbb{E}_{c_j}\Bigg(
\sum_{n\geq 1}\Gamma^{2n}(-1)^n
\sum_{k\text{ is odd}, 0\leq k \leq 2n}(-1)^k
\int_{-\infty}^{0}
\zeta(\omega)
F_k(\omega-\mu_j) F_{2n-k}^*(\omega-\mu_j) \,d\omega
U_S^\dagger(T)c_j O   c_j^\dag U_S(T)\Bigg)
\\
&+  \mathbb{E}_{c_j^\dag}\Bigg(
\sum_{n\geq 1}\Gamma^{2n}(-1)^n
\sum_{k\text{ is even}, 0\leq k \leq 2n}(-1)^k
\int_{-\infty}^{0}
\zeta(\omega)
F_k(\omega+\mu_j) F_{2n-k}^*(\omega+\mu_j)\,d\omega
 U_S^\dagger(T)c_j^\dag c_j O c_j^\dag c_j U_S(T)\Bigg)
\\
& + \mathbb{E}_{c_j^\dag}\Bigg(
\sum_{n\geq 1}\Gamma^{2n}(-1)^n
\sum_{k\text{ is odd}, 0\leq k \leq 2n}(-1)^k
\int_{-\infty}^{0}
\zeta(\omega)
F_k(\omega+\mu_j) F_{2n-k}^*(\omega+\mu_j) \,d\omega
U_S^\dagger(T)c_j^\dag O   c_j U_S(T)\Bigg)
\\
    \end{aligned}
\end{equation}
where $F_k(\omega)= \int_{-\infty< t_1 \leq t_2 \leq \cdots \leq t_{k} <\infty}e^{-i\omega\sum_{r=1}^k(-1)^r t_r} \prod_{r=1}^k f(t_r)\mathrm{d}t_1\cdots \mathrm{d}t_k.$ In this case,
\begin{equation}
    \widetilde{G}_{k, c_j} = F_k(\omega-\mu_j) \times \begin{cases}
        c_j c_j^\dag & k\text{ is even}\\
         c_j  & k\text{ is odd}\\
    \end{cases}, \qquad \widetilde{G}_{k, c_j^\dag} = F_k(\omega+\mu_j) \times\begin{cases}
        c_j^\dag c_j & k\text{ is even}\\
         c_j^\dag  & k\text{ is odd}\\
    \end{cases}.
\end{equation}
\begin{itemize}
    \item When $O = \hat{N}$, we obtain the expression of $\widetilde{\Phi}_\Gamma^\dag(\hat{N})$ in~\cref{eq:Phi_Nhat}. To simplify this expression, we use the canonical
commutation relations with the number operator,
    \begin{equation}
        [c_j c_j^\dag, \hat{N}]=0,  \quad  [c_j,\hat{N}]  = c_j , \quad [c_j^\dag c_j, \hat{N}]=0,\quad  [c_j^\dag, \hat{N}] =-c_j^\dag ,
    \end{equation}
    and the projection identities $(c_j c_j^\dag)^n = c_j c_j^\dag$, $(c_j^\dag c_j )^n =c_j^\dag c_j $. Furthermore, the system Hamiltonian $H$ commutes with $\hat{N}c_jc_j^\dag, \hat{N}c_j^\dag c_j, c_jc_j^\dag$ and $c_j^\dag c_j$.
    \item When $O = I$, the adjoint quantum channel is unital,  $\widetilde{\Phi}_\Gamma^\dag (I) = I $. It gives
    \begin{equation}\label{eq:unital}
\begin{aligned}
   & \mathbb{E}_{j}\Bigg(
\sum_{n\geq 1}\Gamma^{2n}(-1)^n
\sum_{ 0\leq k \leq 2n}(-1)^k
\int_{-\infty}^{0}
\zeta(\omega)
F_k(\omega-\mu_j) F_{2n-k}^*(\omega-\mu_j)\,d\omega
  c_j c_j^\dag \Bigg)
\\
&+  \mathbb{E}_{j}\Bigg(
\sum_{n\geq 1}\Gamma^{2n}(-1)^n
\sum_{ 0\leq k \leq 2n}(-1)^k
\int_{-\infty}^{0}
\zeta(\omega)
F_k(\omega+\mu_j) F_{2n-k}^*(\omega+\mu_j)\,d\omega
   c_j^\dag c_j \Bigg)=0.
\\
\end{aligned}
\end{equation}
\end{itemize}
Combining commutation relations with the unital property, $\widetilde{\Phi}_\Gamma^\dag(\hat{N})$ in~\cref{eq:adjointObserve} reduces to
\begin{equation}
    \begin{aligned}
        \widetilde{\Phi}_\Gamma^\dag (\hat{N}) &= \hat{N} + \hat{N}\times (\text{left hand side of \cref{eq:unital}})  \\
        & + \mathbb{E}_{c_j}\Bigg(
\sum_{n\geq 1}\Gamma^{2n}(-1)^n
\sum_{k\text{ is odd}, 0\leq k \leq 2n}(-1)^k
\int_{-\infty}^{0}
\zeta(\omega)
F_k(\omega-\mu_j) F_{2n-k}^*(\omega-\mu_j) \,d\omega
 c_j c_j^\dag\Bigg)
\\
& - \mathbb{E}_{c_j^\dag}\Bigg(
\sum_{n\geq 1}\Gamma^{2n}(-1)^n
\sum_{k\text{ is odd}, 0\leq k \leq 2n}(-1)^k
\int_{-\infty}^{0}
\zeta(\omega)
F_k(\omega+\mu_j) F_{2n-k}^*(\omega+\mu_j) \,d\omega
c_j^\dag    c_j \Bigg).
\\
    \end{aligned}
\end{equation}
Define $C(\omega) = \left| \sum_{k\text{ odd}} (i\Gamma)^k F_k(\omega)\right|$, then  $\Tr(\rho \widetilde{\Phi}^\dag_\Gamma(\hat{N}))$ takes the form
\begin{equation}
    \Tr(\rho \widetilde{\Phi}^\dag_\Gamma(\hat{N})) =\Tr(\rho\hat{N})-\frac{1}{N}\sum_{j} \underbrace{\int_{-\infty}^0 \zeta(\omega)C^2(\omega+\mu_j)\mathrm{d}\omega}_{\text{cooling coefficient}}\Tr(\rho c_j^\dag c_j)+\frac{1}{N}\sum_{j} \underbrace{\int_{-\infty}^0 \zeta(\omega)C^2(\omega-\mu_j)\mathrm{d}\omega}_{\text{heating coefficient}}\Tr(\rho c_j c_j^\dag)
\end{equation}
Since the integral is restricted to $\omega\leq 0$, and $\mu_j\geq \Delta >0$, the heating coefficient depends on variable $\omega-\mu_j\leq -\Delta$. Thus the heating term is suppressed by the spectral gap $\Delta$. In contrast, the cooling coefficient includes $\omega+\mu_j \in[-2\|h\|+\mu_j, \mu_j]$ which contains the resonant frequency $0$. Consequently, the cooling part contains the resonant region of $C(\omega)$ and yields a non-negligible decay rate.
\begin{proof}
In this proof, we consider the regime when $\sigma \gg 1$.
 \begin{itemize}
    \item We first bound the heating coefficient. For $\omega\leq -\Delta$, and $\Gamma = \Theta(1)$, by~\cref{lm:multiFourier},
\begin{equation}
\begin{aligned}
   &  C(\omega) \leq \sum_{m =1}^{\infty} \Gamma^{2m-1}\left|F_{2m-1}(\omega)\right|
    \leq  \sum_{m=1}^\infty  \frac{(2^{3/4}\pi^{1/4}\Gamma)^{2m-1}}{(2m-1)!} e^{-\frac{\sigma^2\Delta^2}{2m-1}}=:S(\sigma\Delta)\\
\end{aligned}
\end{equation}
 Therefore, the heating coefficients $\int_{-\infty}^0 \zeta(\omega)C^2(\omega-\mu_j)\mathrm{d}\omega =\frac{1}{2\|h\|} \int_{-2\|h\|-|\mu_j|}^{-|\mu_j|} C^2(\omega)\mathrm{d}\omega\leq S^2(\sigma\Delta).$
Combining it with $\Tr(\rho c_j c_j^\dag) = 1-\Tr(\rho c_j^\dag c_j)$, the heating term is bounded by
\begin{equation}
    \frac{1}{N}\sum_{j} \int_{-\infty}^0 \zeta(\omega)C^2(\omega-\mu_j)\mathrm{d}\omega\Tr(\rho c_j c_j^\dag) \leq S^2(\Delta)(1-\frac{1}{N}\Tr(\rho \hat{N}))\leq 2S^2(\sigma\Delta).
\end{equation}
    \item We next lower bound the cooling coefficient. Since $ 0<\Delta \leq\mu_j\leq \|h\|$, the interval $[-2\|h\|+\mu_j, \mu_j]$ contains $[-\|h\|,\Delta]$. Therefore, $\int_{-\infty}^{0}
\zeta(\omega)C^2(\omega+\mu_j)\,d\omega
\geq \frac{1}{2\|h\|}
\int_{-\Delta}^{\Delta}
C^2(\nu)\,d\nu \eqqcolon\kappa.$
It remains to lower bound $\kappa$. Notice that 
\begin{equation}
    \begin{aligned}
        &\int_{-\Delta}^\Delta C^2(\nu)\mathrm{d}\nu 
        % = \int_{-\Delta}^\Delta \left|\sum_{m=1}^\infty (-1)^m \Gamma^{2m-1}F_{2m-1}(\omega) \right|^2\mathrm{d}\nu\\&
        = \int_{-\Delta}^\Delta \left|\sum_{m=1}^\infty (-1)^m \Gamma^{2m-1}\int_{t_1\leq \cdots \leq t_{2m-1}}e^{-i\omega\sum_{r=1}^{2m-1}(-1)^r t_r}\prod_{r=1}^{2m-1}f(t_r)\mathrm{d}t_1 \cdots\mathrm{d}t_{2m-1}\right|^2\mathrm{d}\nu. 
    \end{aligned}
\end{equation}
At frequency $\omega = 0$, the oscillatory phase vanishes, $C(0) = \left|\sum_{m=1}^\infty (-1)^m \frac{\Gamma^{2m-1}}{(2m-1)!} \left(\int_{-\infty}^{\infty}f(t)\mathrm{d}t\right)^{2m-1} \right| = |\sin(\|f\|_{L_1}\Gamma)|.$
Thus, we should choose the coupling strength $\Gamma$ so that $\|f\|_{L_1}\Gamma$ is bounded away from the zeros of the sine function.
Moreover, we bound the Lipschitz constant as
\begin{equation}
\begin{aligned}
    &\left|\sum_{m =1}^{\infty}(-1)^m \Gamma^{2m-1}\partial_\omega F_{2m-1}(\omega)\right|
    \leq  \sum_{m =1}^{\infty}\Gamma^{2m-1}\left|\int_{-\infty< t_1 \leq t_2 \leq \cdots \leq t_{2m-1} <\infty} \sum_{r=1}^{2m-1} |t_r| \prod_{r=1}^{2m-1}f(t_r)\mathrm{d}t_1\cdots \mathrm{d}t_{2m-1}\right|\\
    \leq& \Gamma\|tf\|_{L_1}\cosh(\|f\|_{L_1}\Gamma) =\frac{4\Gamma\sigma}{(2\pi)^{1/4}}\cosh(\|f\|_{L_1}\Gamma).
\end{aligned}
\end{equation}
Consequently, $C(\omega)>|\sin(\|f\|_{L_1}\Gamma)|/2$ for all frequency $\omega$ such that $|\omega|\leq
\frac{(2\pi)^{1/4}|\sin(\|f\|_{L_1}\Gamma)|}
{8\Gamma\sigma\cosh(\Gamma\|f\|_{L^1})}.$
Therefore, in the large $\sigma$ regime, $\kappa\geq
\Omega\left(\frac{|\sin(\|f\|_{L_1}\Gamma)|^2}{\|h\|\sigma}\right)$.
    The cooling term is lower bounded by
    \begin{equation}
        \frac{1}{N}\sum_{j} \int_{-\infty}^0 \zeta(\omega)C^2(\omega+\mu_j)\mathrm{d}\omega\Tr(\rho c_j^\dag c_j)\geq \frac{\kappa}{N}\Tr(\rho \hat{N}) =\Omega (\frac{|\sin(\|f\|_{L_1}\Gamma)|^2}{N\sigma\|h\|})\Tr(\rho \hat{N}).
    \end{equation}
\end{itemize}
In conclusion, we have
\begin{equation}
    \Tr(\rho \widetilde{\Phi}^\dag_{\Gamma}(\hat{N}))\leq (1 - \frac{|\sin(\|f\|_{L_1}\Gamma)|^2}{\sigma N\|h\|})\Tr(\rho\hat{N})+2S^2(\sigma\Delta).
\end{equation}
This implies that
\begin{equation}
    \|\widetilde{\Phi}_{\Gamma}^m(\rho) -\ketbra{\psi_0}\|\leq 2\sqrt{\Tr(\widetilde{\Phi}_{\Gamma}^m(\rho)\hat{N})} \leq\sqrt{\left(1-\frac{|\sin(\|f\|_{L_1}\Gamma)|^2}{\sigma N\|h\|}\right)^{m}N+\frac{\sigma N\|h\|}{|\sin(\|f\|_{L_1}\Gamma)|^2}S^2(\sigma\Delta)}
\end{equation}
Then for $m = \tau_{\rm{mix},\Phi_{\Gamma}}(\epsilon/2)$,
\begin{equation}\label{eq:c_k_fixed_point}
\begin{aligned}
    &\|\Phi_\Gamma^m - \rho_{\rm fix}(\Phi_\Gamma)\|\leq  \|\Phi_\Gamma^m - \widetilde{\Phi}_\Gamma^m\|+  \|\widetilde{\Phi}_\Gamma^m - \ketbra{\psi_0}\| +  \|\ketbra{\psi_0} - \rho_{\rm fix}(\Phi_\Gamma)\|\\
    & \leq  \sqrt{\left(1-\frac{|\sin(\|f\|_{L_1}\Gamma)|^2}{\sigma N\|h\|}\right)^{m}N+\frac{\sigma N\|h\|}{|\sin(\|f\|_{L_1}\Gamma)|^2}S^2(\sigma\Delta)}
    +\mathcal{O}\left(m\frac{\Gamma^2\sigma}{T}\exp(-\frac{T^2}{4\sigma^2})\right) \\
    &+ m\sum_{n=N+1}^\infty\frac{\left(\mathcal{O}\left(\Gamma\|A_S\|\right)\right)^{2n}}{(2n)!}
+m\exp(-\frac{\sigma^2\Delta^2}{2N})\sum_{n=1}^{N}(2n+1)\frac{\mathcal{O}(\Gamma\|A_S\|\|H\|)^{2n}}{(2n)!}+\epsilon/2
\end{aligned}
\end{equation}
Given $\epsilon>0$, we split the series $S(\sigma \Delta)$ at $R$, then
\begin{equation}
    S(\sigma \Delta)
    \leq
    \sum_{m=1}^R
\frac{{(2^{3/4}\pi^{1/4}\Gamma)^{2m-1}}}{(2m-1)!}\exp\left(-\frac{\sigma^2\Delta^2}{2R-1}\right)
    +
    \sum_{m>R}
\frac{{(2^{3/4}\pi^{1/4}\Gamma)^{2m-1}}}{(2m-1)!}.
\end{equation}
Thus for arbitrary $\epsilon'$,  $S(\Delta)\leq \epsilon'$ by choosing $R=
    O\left(
    \frac{\log(1/\epsilon')}{\log\log(1/\epsilon')}
    \right),
    \sigma=
    \Omega\!\left(
    \frac{\sqrt{2R-1}}{\Delta}
    \sqrt{\log(1/\epsilon')}
    \right) .$
Consequently, for $\sigma = \widetilde{\Omega}(\Delta^{-1}\log(N/\epsilon))$ and $T = \widetilde{\Omega}(\sigma) = \widetilde{\Omega}(\Delta^{-1}\log(N/\epsilon))$,
we have $\tau_{\rm mix,\Phi_\Gamma} = \mathcal{O}(\frac{\|h\|\sigma N\log(N/\epsilon^2)}{|\sin(\|f\|_{L_1}\Gamma)|^2}).$
The rescaled mixing time $t_{\rm mix,\Phi_\Gamma}= \mathcal{O}(\frac{\|h\|\Gamma^2 N\log(N/\epsilon^2)}{|\sin(\|f\|_{L_1}\Gamma)|^2})$.
\end{proof}

\subsection{When $A_S \in\{\pm a_k^\dag, \pm a_k\}$}\label{sec:ground_mixing_time_2}
When the system coupling
operator \(A_S\) is chosen from the set $\{\pm a_j,\pm a_j^\dagger:1\leq j\leq N\}.$ We obtain the following theorem for the mixing times.
\begin{thm}\label{thm:S16}
    Let $g(\omega) = \frac{1}{2\|h\|}1_{[-2\|h\|, 0]}$ then $\zeta(\omega) = \frac{1}{2\|h\|}1_{[-2\|h\|, 2\|h\|]}$. With system operators $A_S \in \{\pm a_k^\dag, \pm a_k\}$, for any $\epsilon>0$, if $\sigma = \widetilde{\Theta}(\Delta^{-1}),  T = \widetilde{\Theta}(\sigma)$ and the coupling strength $\Gamma = {\mathcal{O}}(\frac{1}{N^{3/2}})$
    then we have
    \begin{equation}
    t_{\rm{mix},\Phi_{\Gamma}}(\epsilon) = \mathcal{O}(\|h\|N\log(N/\epsilon)).
    \end{equation}
\end{thm}
\begin{proof}
We prove the result by controlling the expected excitation number. 
Let $P_0 = \ketbra{0}$ be the vacuum projector and set $P_\perp = I - P_0$. We decompose the operator space according to  
\begin{equation}
    \Pi_1\rho \coloneqq P_\perp\rho P_\perp\eqqcolon \rho_1, \quad \Pi_0\rho \coloneqq \rho - \Pi_1\rho = P_0\rho P_0 + P_0\rho P_\perp + P_\perp \rho P_0\eqqcolon\rho_0.
\end{equation}
Using these projections, we decompose the quantum channel as well
\begin{equation}
    \widetilde{\Phi}_\Gamma (\rho) = \underbrace{\Pi_1 \widetilde{\Phi}_\Gamma (\Pi_1 \rho) + \Pi_0 \widetilde{\Phi}_\Gamma (\Pi_0 \rho)+ \Pi_0 \widetilde{\Phi}_\Gamma (\Pi_1 \rho)}_{\widetilde{\Phi}_{\Gamma,1}(\rho)} +\underbrace{\Pi_1 \widetilde{\Phi}_\Gamma (\Pi_0 \rho) }_{\widetilde{\Phi}_{\Gamma,2}(\rho)}
\end{equation}
We denote the first three terms by \(\widetilde{\Phi}_{\Gamma,1}\) and the last term by
\(\widetilde{\Phi}_{\Gamma,2}\).
 Since $\hat{N}P_0 =P_0 \hat{N}= 0$,  for every $\rho\in\mathcal{B}(\mathcal{H})$, $\Tr(\Pi_0(X)\hat N)=0.$
Therefore, $\Tr(\widetilde{\Phi}_{\Gamma,1}(\rho)\hat{N}) = \Tr(\Pi_1 \widetilde{\Phi}_\Gamma ( \rho_1)\hat{N}) = \Tr(\widetilde{\Phi}_\Gamma (\rho_1)\hat{N}) $.
The proof has two parts. First, we show that
\(\widetilde{\Phi}_{\Gamma,1}\) contracts the excitation number on the
excited sector. Second, we show that the leakage term
\(\widetilde{\Phi}_{\Gamma,2}\), which maps the vacuum sector into
the excited sector, is exponentially small in the filter width \(\sigma\).
% there exists a constant $\delta >0$ such that  $\Tr(\widetilde{\Phi}_{\Gamma,1}(\rho)\hat{N})\leq (1-\delta)\Tr(\rho\hat{N})$ by
% \begin{equation}
%     \Tr(\widetilde{\Phi}_\Gamma (\rho_1)\hat{N}) \leq (1-\delta)\Tr(\rho_1\hat{N}).
% \end{equation}
% Moreover, while $\widetilde{\Phi}_{\Gamma,2}$ is not of the order $\mathcal{O}(\Tr(\rho\hat{N}))$, the value is relatively small in the large $\sigma$-regime.

We first compute the expression of $\widetilde{\Phi}_\Gamma (\hat{N})$. $\widetilde{G}_{k, a_j}$ and $\widetilde{G}_{k, a_j^\dag}$ of $\widetilde{\Phi}_\Gamma$ in~\cref{eq:Phi_Nhat} are
the expressions for
\begin{equation}
\begin{aligned}
    &\widetilde{G}_{k, a_j}(\omega) = \sum_{\ell_1,\ldots,\ell_{k}=1}^N \prod_{\substack{1\le r\le k\\ r\ {\rm odd}}}
U_{j,\ell_r}
\prod_{\substack{1\le r\le k\\ r\ {\rm even}}}
U_{j,\ell_r}^*  F_k(\omega-\mu_{\boldsymbol{\ell}_{1:k}}) \times \begin{cases}
        c_{\ell_1}c_{\ell_2}^\dagger\cdots
c_{\ell_{k}}^\dagger & k\text{ is even}\\
         c_{\ell_1}c_{\ell_2}^\dagger\cdots
c_{\ell_{k}} & k\text{ is odd}\\
    \end{cases}, \\
    &\widetilde{G}_{k, a_j^\dag}(\omega) = \sum_{\ell_1,\ldots,\ell_{k}=1}^N
    \prod_{\substack{1\le r\le k\\ r\ {\rm odd}}}
U_{j,\ell_r}^*
\prod_{\substack{1\le r\le k\\ r\ {\rm even}}}
U_{j,\ell_r} F_k(\omega+\mu_{\boldsymbol{\ell}_{1:k}})\times\begin{cases}
        c_{\ell_1}^\dagger c_{\ell_2}\cdots
c_{\ell_{k}} & k\text{ is even}\\
         c_{\ell_1}^\dagger c_{\ell_2}\cdots
c_{\ell_{k}}^{^\dagger} & k\text{ is odd}\\
    \end{cases}.
\end{aligned}
\end{equation}
Here $F_k(\omega_1, \cdots, \omega_k)= \int_{ t_1 \leq t_2 \leq \cdots \leq t_{k} }e^{-i\sum_{r=1}^k(-1)^r\omega_r t_r} \prod_{r=1}^k f(t_r)\mathrm{d}t_1\cdots \mathrm{d}t_k.$
For odd $k$, since $f$ is an even function, it satisfies $F_k(\omega_1, \cdots, \omega_k)^* = F_k(\omega_k,\cdots, \omega_1)$. Moreover, the canonical anticommutation
relations imply
$[c_\ell,\hat N]=c_\ell, [c_\ell^\dagger,\hat N]=-c_\ell^\dagger$. For the multi index $\boldsymbol{\ell}_{2n}$, we write $\boldsymbol{\ell}_{2n} = \underbrace{(\ell_1, \cdots,\ell_k}_{\boldsymbol{\ell}_{1:k}},\underbrace{\ell_{k+1},\cdots \ell_{2n})}_{\boldsymbol{\ell}_{k+1:2n}}\in[N]^{2n}.$
% we write the first $k$ terms as $\boldsymbol{\ell}_{1:k}$ and the rest $2n-k$ terms  as $\boldsymbol{\ell}_{k+1:2n}$.
Thus, the commutator of a fermionic monomial with $\hat N$ is determined by
its net particle number change. In particular, the even-$k$ strings appearing
in $\widetilde G_{k,a_j}$ and $\widetilde G_{k,a_j^\dagger}$ are number
preserving, whereas the odd-$k$ strings contain one more annihilation operator
or one more creation operator, respectively.
Therefore,
    \begin{equation}
        [\widetilde{G}_{k, a_j}(\omega), \hat{N}] = \begin{cases}
            0 & \text{if } k \text{ is even}\\
            \widetilde{G}_{k, a_j}(\omega) & \text{if } k \text{ is odd}\\
        \end{cases},\quad [\widetilde{G}_{k, a_j^\dag}(\omega), \hat{N}] = \begin{cases}
            0 & \text{if } k \text{ is even}\\
            -\widetilde{G}_{k, a_j^\dag}(\omega) & \text{if } k \text{ is odd}.\\
        \end{cases}
    \end{equation}
        Substituting the expression of $\widetilde{G}_{k,A_S}(\omega )\hat{N} = \hat{N}\widetilde{G}_{k,A_S}(\omega ) + [\widetilde{G}_{k,A_S}(\omega ),\hat{N}]$ into the expression in~\cref{eq:Phi_Nhat},
    \begin{equation}
    \begin{aligned}
        \widetilde{\Phi}_\Gamma^\dagger(\hat{N})  = \hat{N} + \hat{N}(\widetilde{\Phi}_\Gamma^\dagger(I)-I)
        - &\frac{1}{2} \mathbb{E}_{a_j}\Bigg(
\sum_{n\geq 1}\Gamma^{2n}(-1)^n
\sum_{0\leq k \leq 2n, \rm odd}
\int_{-\infty}^{0}
\zeta(\omega)
U_S^\dagger(T)\widetilde{G}_{k,a_j}(\omega)\widetilde{G}_{2n-k,a_j}^\dagger(\omega)U_S(T)
\,d\omega
\Bigg)\\
+& \frac{1}{2} \mathbb{E}_{a_j^\dagger}\Bigg(
\sum_{n\geq 1}\Gamma^{2n}(-1)^n
\sum_{0\leq k \leq 2n, \rm odd}
\int_{-\infty}^{0}
\zeta(\omega)
U_S^\dagger(T)\widetilde{G}_{k,a_j^\dagger}(\omega)\widetilde{G}_{2n-k,a_j^\dagger}^\dagger(\omega)U_S(T)
\,d\omega
\Bigg).\\
    \end{aligned}
\end{equation}
In this case,
the signs
$\pm$ cancel in the products appearing in the channel.
By the unital property of the adjoint channel, $\widetilde{\Phi}_\Gamma^\dag(I) = I$.

We now estimate the contribution from the excited sector.  Write $\Tr(\rho_1 \widetilde{\Phi}_\Gamma^\dagger (\hat{N}))
 \eqqcolon \Tr(\rho_1\hat{N} ) + M_1  + D_1 + R_{\geq 2}.$
Here $M_1$ denotes the leading-order contribution from the $a_j$ term, while
$D_1$ denotes the leading-order contribution from the $a_j^\dagger$ term,
both corresponding to $n=1$,
 \begin{equation}
        M_1\coloneqq  \frac{\Gamma^2}{N}\int_{-\infty}^0 \zeta(\omega) \sum_{\ell=1}^N \left|F_1(\omega-\mu_\ell)\right|^2\mathrm{d}\omega \Tr(\rho_1 c_\ell c_\ell^\dag),\quad D_1 \coloneqq  -\frac{\Gamma^2}{N}\int_{-\infty}^0 \zeta(\omega) \sum_{\ell=1}^N \left|F_1(\omega+\mu_\ell)\right|^2\mathrm{d}\omega\Tr(\rho_1 c_\ell^\dag c_\ell ).
    \end{equation}
The terms $R_{\geq 2}$ collect
the remaining higher-order contributions,
\begin{equation}
    \begin{aligned}
        R_{\geq 2} & \coloneqq  \frac{1}{2}\mathbb{E}_{j}\Bigg(
\sum_{n\geq 2}\Gamma^{2n}(-1)^n
\int_{-\infty}^{0}
\zeta(\omega)
\sum_{\ell_1,\ldots,\ell_{2n}=1}^N \prod_{r=1}^n U_{j,\ell_{2r-1}}\prod_{r=1}^n U_{j,\ell_{2r}}^*  \sum_{k\text{ is odd}, 0\leq k \leq 2n}F_k(\omega-\mu_{\boldsymbol{\ell}_{1:k}})F_{2n-k}(\omega-{\mu}_{\boldsymbol{\ell}_{k+1:2n}})\,d\omega \\
&\times \Tr(\rho_1 U_S^\dagger(T)c_{\ell_1} c_{\ell_2}^\dagger\cdots
c_{\ell_{2n}}^\dagger U_S(T))
\Bigg)\\
-& \frac{1}{2}\mathbb{E}_{j}\Bigg(
\sum_{n\geq 2}\Gamma^{2n}(-1)^n
\int_{-\infty}^{0}
\zeta(\omega)
\sum_{\ell_1,\ldots,\ell_{2n}=1}^N \prod_{r=1}^n U_{j,\ell_{2r-1}}^*\prod_{r=1}^n U_{j,\ell_{2r}}\sum_{k\text{ is odd}, 0\leq k \leq 2n}F_k(\omega+{\mu}_{\boldsymbol{\ell}_{1:k}})F_{2n-k}(\omega+{\mu}_{\boldsymbol{\ell}_{k+1:2n}})\,d\omega \\
&\times \Tr(\rho_1 U_S^\dagger(T)c_{\ell_1}^\dagger c_{\ell_2}\cdots
c_{\ell_{2n}}U_S(T))
\Bigg).\\
    \end{aligned}
\end{equation}
We first consider the leading order term, which is of order $\mathcal{O}(\Gamma^2)$. Since $\omega \in[-2\|h\|,0]$ and $\Delta\leq \mu_j\leq \|h\|$, we have $|\omega-\mu|\geq \Delta$ and $\omega +\mu_\ell \in[-2\|h\|+\mu_\ell, \mu_\ell]\supseteq[-\Delta,\Delta]$. Therefore, by~\cref{lm:multiFourier}, for $\sigma\Delta >1$,
    \begin{equation}
        M_1\leq \frac{\Gamma^2}{N}(2^{3/4}\pi^{1/4})^2 \exp(-2\sigma^2\Delta^2)(N-1)\Tr(\rho_1\hat{N}),\quad D_1 \leq - \frac{\Gamma^2\pi\operatorname{erf}(\sqrt{2})}{\sigma N\|h\| }
        \Tr(\rho_1\hat{N}) .
    \end{equation}
It remains to bound the higher-order remainder. We notice that $\Tr(\rho_1 U_S^\dagger(T)c_{\ell_1}^\dagger c_{\ell_2}\cdots
c_{\ell_{2n}}U_S(T))\leq \Tr(\rho_1\hat{N})$, 
and the analogous bound for
$c_{\ell_1}c_{\ell_2}^\dagger\cdots c_{\ell_{2n}}^\dagger$. Together with ~\cref{lm:multiFourier}, we get
\begin{equation}
\begin{aligned}
    &\int_{-\infty}^0\zeta(\omega)|F_k(\omega-\mu_{\boldsymbol{\ell}_{1:k}})|\,|F_{2n-k}(\omega-{\mu}_{\boldsymbol{\ell}_{k+1:2n}})|\,d\omega\\
    % &\leq\frac{(2^{3/4}\pi^{1/4})^{2n}}{k!(2n-k)!}\frac{1}{2\|h\|}\int_{\mathbb{R}}\exp\left(-\frac{\sigma^2}{k}\left(\omega+\sum_{p=1}^k(-1)^p\mu_{\ell_p}\right)^2-\frac{\sigma^2}{2n-k}\left(\omega+\sum_{p=k+1}^{2n}(-1)^p\mu_{\ell_p}\right)^2\right)\,d\omega\\
    &\leq \frac{(2^{3/4}\pi^{1/4})^{2n}}{k!(2n-k)!}\frac{1}{2\sigma\|h\|}\sqrt{\frac{\pi k(2n-k)}{2n}}\exp\left(-\frac{\sigma^2}{2n}\left(\sum_{p=1}^k(-1)^p\mu_{\ell_p} -\sum_{p=k+1}^{2n}(-1)^p\mu_{\ell_p}\right)^2\right)
    \leq \frac{(2^{3/4}\pi^{1/4})^{2n}}{k!(2n-k)!}\frac{1}{2\sigma\|h\|}\sqrt{\frac{\pi k(2n-k)}{2n}}.\\
\end{aligned}
\end{equation}
Moreover $\mathbb{E}_j \sum_{\ell_1,\ldots,\ell_{2n}=1}^N \prod_{r=1}^n
    |U_{j,\ell_{2r-1}}|\prod_{r=1}^n
    |U_{j,\ell_{2r}}| = \mathbb{E}_j \left(\sum_{\ell = 1}^N |U_{j,\ell}|\right)^{2n}\leq N^{n}.$
Therefore,
\begin{equation}
    \begin{aligned}
        |R_{\geq 2}|&\leq\sum_{n\geq 2}\Gamma^{2n}N^n\sum_{k\text{ is odd}, 0\leq k \leq 2n}\frac{(2^{3/4}\pi^{1/4})^{2n}}{k! (2n-k)!}\sqrt{\frac{\pi k(2n-k)}{2n}}\frac{1}{2\sigma\|h\|}\Tr(\rho_1\hat{N})\leq \frac{1}{\sigma \|h\|}\sum_{n\geq 2}\frac{(\mathcal{O}(\sqrt{N}\Gamma))^{2n}}{(2n)!}\sqrt{n}\Tr(\rho_1 \hat{N}). \\
    \end{aligned}
\end{equation}
Combining the leading-order bounds with the remainder estimate gives,
\begin{equation}
\begin{aligned}
    \Tr(\widetilde{\Phi}_{\Gamma,1}(\rho)\hat{N})
    &\leq \left(1- \frac{\Gamma^2\pi}{\sigma N\|h\|}+\Gamma^2(2^{3/4}\pi^{1/4})^2\exp(-2\sigma^2\Delta^2)+\frac{1}{\sigma \|h\|}\sum_{n\geq 2}\frac{(\mathcal{O}(\sqrt{N}\Gamma))^{2n}}{(2n)!}\sqrt{n}\right)\Tr(\rho_1 \hat{N}) \\
        &\leq (1-\frac{\Gamma^2\pi}{2\sigma N\|h\|})\Tr(\rho \hat{N})
\end{aligned}
\end{equation}
where the last inequality is by choosing  $\Gamma = \mathcal{O}(\frac{1}{N^{3/2}})$ and $\sigma = \widetilde{\Theta}(\Delta^{-1}\sqrt{\log(\frac{1}{N\|h\|})})$.

We next bound the contribution from $\widetilde{\Phi}_{\Gamma,2}$ using $\Tr(\widetilde{\Phi}_{\Gamma,2}(\rho)\hat{N}) = \Tr(\widetilde{\Phi}_{\Gamma}(\rho_0)\hat{N})  $.
By the commutative property $[P_\perp, H]=[P_0, H]=0$, $P_0\hat{N} = \hat{N}P_0 = 0$,  we obtain
\begin{equation}
    \begin{aligned}
    &\Tr(\rho_0\widetilde{\Phi}_{\Gamma}^\dagger (\hat{N}))
         =-\frac{1}{2} \mathbb{E}_{j}\Bigg(
\sum_{n\geq 1}\Gamma^{2n}(-1)^n
\int_{-\infty}^{0}
\zeta(\omega)
\sum_{\ell_1,\ldots,\ell_{2n}=1}^N \prod_{r=1}^n U_{j,\ell_{2r-1}}\prod_{r=1}^n U_{j,\ell_{2r}}^*  \\
&\times\sum_{k\text{ is odd}, 0\leq k \leq 2n}F_k(\omega-\mu_{\boldsymbol{\ell}_{1:k}})F_{2n-k}(\omega-{\mu}_{\boldsymbol{\ell}_{k+1:2n}})\,d\omega e^{i\sum_{p=1}^{2n}(-1)^p\mu_{\ell_p}T}\Tr\left(\rho\Pi_0 (c_{\ell_1}c_{\ell_2}^\dagger\cdots
c_{\ell_{2n}}^{\dagger})\right)
\Bigg)\\
    \end{aligned}
\end{equation}
The vacuum matrix element is nonzero only when the creationa and annihilation operators are paired. More precisely, $\Tr\left(\rho\Pi_0 (c_{\ell_1}c_{\ell_2}^\dagger\cdots
c_{\ell_{2n}}^{\dagger})\right)
\Bigg) =\bra{0}\rho\ket{0} \prod_{j=1}^n\delta_{\ell_{2j-1},\ell_{2j}}$. Hence, after summing over the indices, we have
\begin{equation}
    \begin{aligned}
        \Tr(\widetilde{\Phi}_{\Gamma,2}(\rho) \hat{N}) &=\Tr(\rho_0\widetilde{\Phi}_{\Gamma}^\dagger \hat{N})= -\frac{1}{2} \mathbb{E}_j\sum_{n\geq 1} \Gamma^{2n}(-1)^n\int \zeta(\omega) \mathrm{d}\omega \sum_{\substack{\ell_2,\ell_4,\cdots, \ell_{2n-2},\ell_{2n}=1}}^N\prod_{r=1}^{n}|U_{j,\ell_{2r}}|^2 \\
\times &\prod_{j=1}^n\delta_{\ell_{2j-1},\ell_{2j}}\sum_{k = 0, odd}^{2n} F_{2n-k}^*(\omega-\mu_{\ell_{k+1:2n}})F_{k}^*(\omega-\mu_{\ell_{1:k}})\bra{0}\rho \ket{0}\\
    \end{aligned}
\end{equation}
Here $\ell_{2r-1} =\ell_{2r}$. For odd $k$, the pairing leaves one unpaired frequency in each $F$ factor, by~\cref{lm:multiFourier}, we have
\begin{equation}
\begin{aligned}
    |F_k^*(\omega-\mu_{\boldsymbol{\ell}_{1:k}})|&\leq \frac{(2^{3/4}\pi^{1/4})^k}{k!}\exp\left(-\frac{\sigma^2}{k}(\omega-\mu_{\ell_k})^2\right).\\
    % |F_{2n-k}^*(\omega-{\mu}_{\boldsymbol{\ell}_{k+1:2n}}) |&\leq \frac{(2^{3/4}\pi^{1/4})^{2n-k}}{(2n-k)!}\exp\left(-\frac{\sigma^2}{2n-k}(\omega-\mu_{\ell_{k+1}})^2\right).
\end{aligned}
\end{equation}
Using also $|\omega-\mu_\ell|\geq \Delta$, we obtain
\begin{equation}
\begin{aligned}
    |\Tr(\widetilde{\Phi}_{\Gamma,2}(\rho) \hat{N})|  & \leq \frac{1}{2}\sum_{n\geq 1}(2^{3/4}\pi^{1/4}\Gamma)^{2n}\sum_{k=0,\rm odd}^{2n}\frac{1}{k!(2n-k)!}\exp(-\frac{\sigma^2\Delta^2}{2n-k})\exp(-\frac{\sigma^2\Delta^2}{k})|\bra{0}\rho \ket{0}|\\
    &\leq \sum_{n\geq 1} \frac{(2^{7/4}\pi^{1/4}\Gamma)^{2n}}{(2n)!}\exp(-\frac{2}{n}\sigma^2\Delta^2)|\bra{0}\rho \ket{0}|.
\end{aligned}
\end{equation}

This implies that
\begin{equation}
\begin{aligned}
    &\|\widetilde{\Phi}_{\Gamma}^m(\rho) -\ketbra{\psi_0}\|\leq 2\sqrt{\Tr(\widetilde{\Phi}_{\Gamma}^m(\rho)\hat{N})} \\
    & \leq \mathcal{O}\left(\sqrt{\left(1-\frac{\Gamma^2\pi}{2\sigma N\|h\|}\right)^mN+\underbrace{\frac{\sigma N\|h\|}{\Gamma^2}\sum_{n\geq 1} \frac{(\mathcal{O}(\Gamma))^{2n}}{(2n)!}\exp(-\frac{2}{n}\sigma^2\Delta^2)|\bra{0}\rho \ket{0}|}_{S(\sigma,\Delta,\Gamma,N)}}\right)
\end{aligned}
\end{equation}
Here for any $\epsilon>0$,
\begin{equation}
    S(\sigma,\Delta,\Gamma,N)\leq \frac{\sigma N\|h\|}{\Gamma^2}\exp(-\frac{2}{R}\sigma^2\Delta^2)\sum_{n=1}^R \frac{(\mathcal{O}(\Gamma))^{2n}}{(2n)!} + \frac{\sigma N\|h\|}{\Gamma^2}\sum_{n\geq R+1} \frac{(\mathcal{O}(\Gamma))^{2n}}{(2n)!}\leq \epsilon^2,
\end{equation}
by choosing $R = \widetilde{\Theta}\left(\log\left(\frac{\sigma N }{\epsilon\Gamma}\right)\right), \sigma = \widetilde{\Theta}(\Delta^{-1}\sqrt{R}) = \widetilde{\Theta}(\Delta^{-1}).$
Since $\|{\Phi}_\Gamma \|_{1\to 1} = \|\widetilde{\Phi}_\Gamma \|_{1\to 1} = 1$, we have $\|\Phi_\Gamma^m - \widetilde{\Phi}_\Gamma^m\|\leq m \|\Phi_\Gamma - \widetilde{\Phi}_\Gamma\|$. Then, for $m = \tau_{\rm{mix},\Phi_{\Gamma}}(\epsilon/2)$, using~\cref{eq:c_k_fixed_point},
\begin{equation}
\begin{aligned}
    &\|\Phi_\Gamma^m - \rho_{\rm fix}(\Phi_\Gamma)\|%\leq  \|\Phi_\Gamma^m - \widetilde{\Phi}_\Gamma^m\|+  \|\widetilde{\Phi}_\Gamma^m - \ketbra{\psi_0}\| +  \|\ketbra{\psi_0} - \rho_{\rm fix}(\Phi_\Gamma)\|\\
    \leq  \mathcal{O}\left(\sqrt{\left(1-\mathcal{O}\left(\frac{\Gamma^2\|h\|}{N\sigma}\right)\right)^mN+S(\sigma,\Delta,\Gamma,N)}\right)
    +\mathcal{O}\left(m\frac{\Gamma^2\sigma}{T}\exp(-\frac{T^2}{4\sigma^2})\right) \\
    &+ m\sum_{n=N+1}^\infty\frac{\left(\mathcal{O}\left(\Gamma\|A_S\|\right)\right)^{2n}}{(2n)!}
+m\exp(-\frac{\sigma^2\Delta^2}{2N})\sum_{n=1}^{N}(2n+1)\frac{\mathcal{O}(\Gamma\|A_S\|\|H\|)^{2n}}{(2n)!}+\epsilon/2
\end{aligned}
\end{equation}
By choosing $\sigma = \widetilde{\Theta}(\Delta^{-1}),  T = \widetilde{\Theta}(\sigma)$,
we have $\tau_{\rm{mix},\Phi_{\Gamma}}(\epsilon) = \frac{\|h\|N\sigma}{\Gamma^2}\log(N/\epsilon),$
and the rescaled mixing time becomes $ t_{\rm{mix},\Phi_{\Gamma}}(\epsilon) = \|h\|N\log(N/\epsilon)$.
We conclude the proof.
\end{proof}

% In particular, if $A_S \in\{\pm c_k^\dag, \pm c_k\}$, the constrains on the coupling strength $\Gamma = \mathcal{O}(\frac{1}{N^{3/2}})$ can be improved to $\Gamma = {\Theta}(1)$ due to better approximation of $R_{\geq 2}$.

\section{Numerical implementation}\label{sec:numerics_appendix}
We consider the transverse field Ising model (TFIM), the Hubbard model, and 1D axial next-nearest-neighbor Ising (ANNNI) model to verify our analysis results and explore the performance of the algorithm in even stronger coupling parameter.  We investigate both thermal state and ground state preparation for the first two models and only ground state for the last model, varying the coupling parameter ${\Gamma}$ to examine its impact on convergence behavior and accuracy.
\subsection{Thermal state preparation}\label{sec:TFIM_numerics}
\paragraph{Transverse field Ising model (TFIM)} Consider the transverse field Ising model (TFIM) in~\eqref{eqn:TFIM} with $ J = 1, g = 1.2$.
% \[
%     H = -J\sum_{i=1}^{L-1} Z_i Z_{i+1} - g \sum_{i = 1}^L X_i\,,
% \]
% where we set
We now provide more detailed numerical results for thermal state preparation with different choices of parameters.
\begin{itemize}
\item The thermal state setting with $L=4$ in the regime $\Gamma = \Theta(1)$: The additional numerical results are presented in~\cref{Fig:TFIM_4_thermal}. As shown in~\cref{fig1:dynamics_2}, the state converges to the thermal state with very high accuracy, and the convergence rate improves as ${\frac{\Gamma}{\sqrt{\sigma}}}$ increases. Moreover,~\cref{fig1:error} illustrates that the steady state of the channel approaches the target thermal state as $\sigma$ increases, while remaining nearly constant with respect to variations in ${\frac{\Gamma}{\sqrt{\sigma}}}$. By computing the spectral gap after constructing the discrete quantum channel $\Phi_\Gamma $ explicitly,
~\cref{fig1:spectral2} shows the spectral gap remains independent of $\sigma$ and scales as ${\frac{\Gamma^2}{\sigma}}$, which is similar to the result we showed for Hubbard model ~\cref{fig3:spectral_gap}.

\item Thermal State preparation with $L=4$ in the strong coupling regime $\Gamma = \Theta(\sigma)$: We extend our numerical experiments to the strong coupling regime. For thermal state preparation, we keep $\sigma$, $T$, and $\omega$ the same as before and vary ${\frac{\Gamma}{\sqrt{\sigma}}}$ as ${\frac{\Gamma}{\sqrt{2\sigma}}} = 1, 0.5, 0.25, 0.1, 0.05, 0.01$, as shown in~\cref{fig_large:dynamics_1} and ~\cref{fig_large:dynamics_2}.
We observe that when ${\frac{\Gamma}{{\sigma}}} \leq 1/2$, the fixed point remains close to the target state, and the spectral gap continues to scale as ${\frac{\Gamma^2}{\sigma}}$ in this regime. These observations suggest that accurate thermal state preparation remains feasible even under strong system--bath interactions, with a correspondingly faster convergence rate. Furthermore, the apparent coupling threshold at ${\frac{\Gamma}{{\sigma}}}  \approx 0.5$ seems to be independent of system size, as we observe similar behavior in both 4-qubit and 8-qubit systems. This phenomenon lies beyond our current theoretical guarantees and suggests that the robustness of system–bath interaction models may be even greater than what our analysis presently establishes.

\end{itemize}

\begin{figure}[!htbp]
    \centering
    \begin{subfigure}[b]{0.3\textwidth}
        \centering        \includegraphics[width=\textwidth]{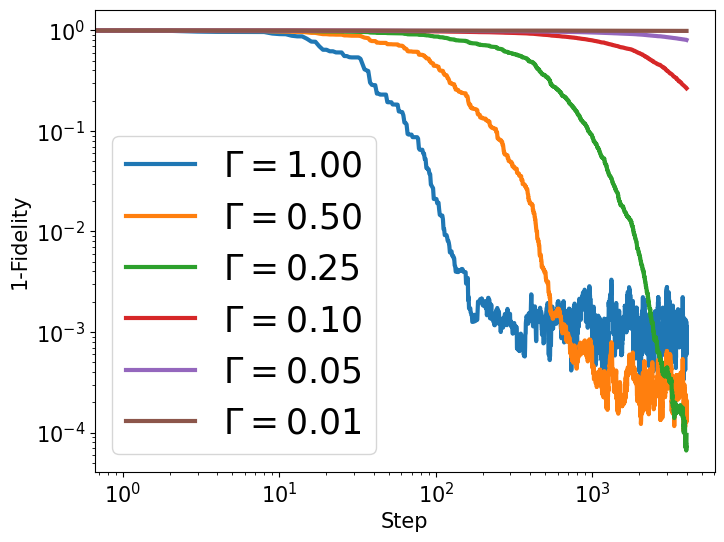}
        \caption{}
        \label{fig1:dynamics_2}
    \end{subfigure}
    \begin{subfigure}[b]{0.3\textwidth}
        \centering
        \includegraphics[width=\textwidth]{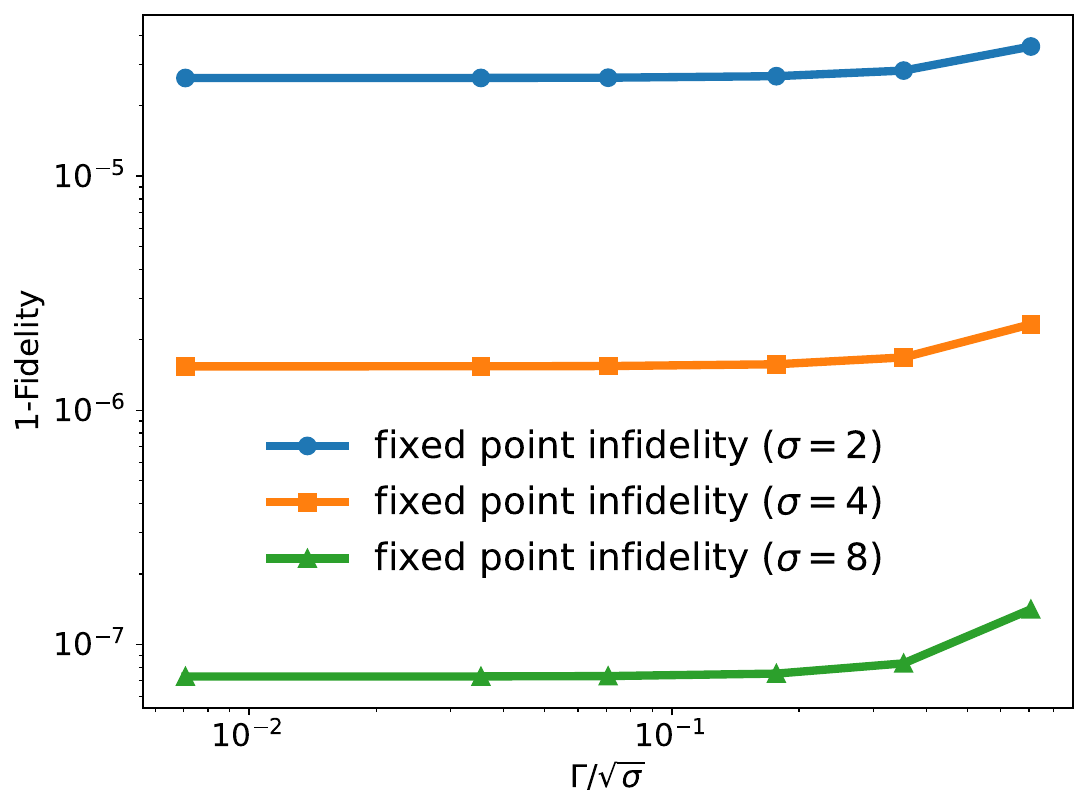}
        \caption{}
        \label{fig1:error}
    \end{subfigure}
    \begin{subfigure}[b]{0.3\textwidth}
        \centering
        \includegraphics[width=\textwidth]{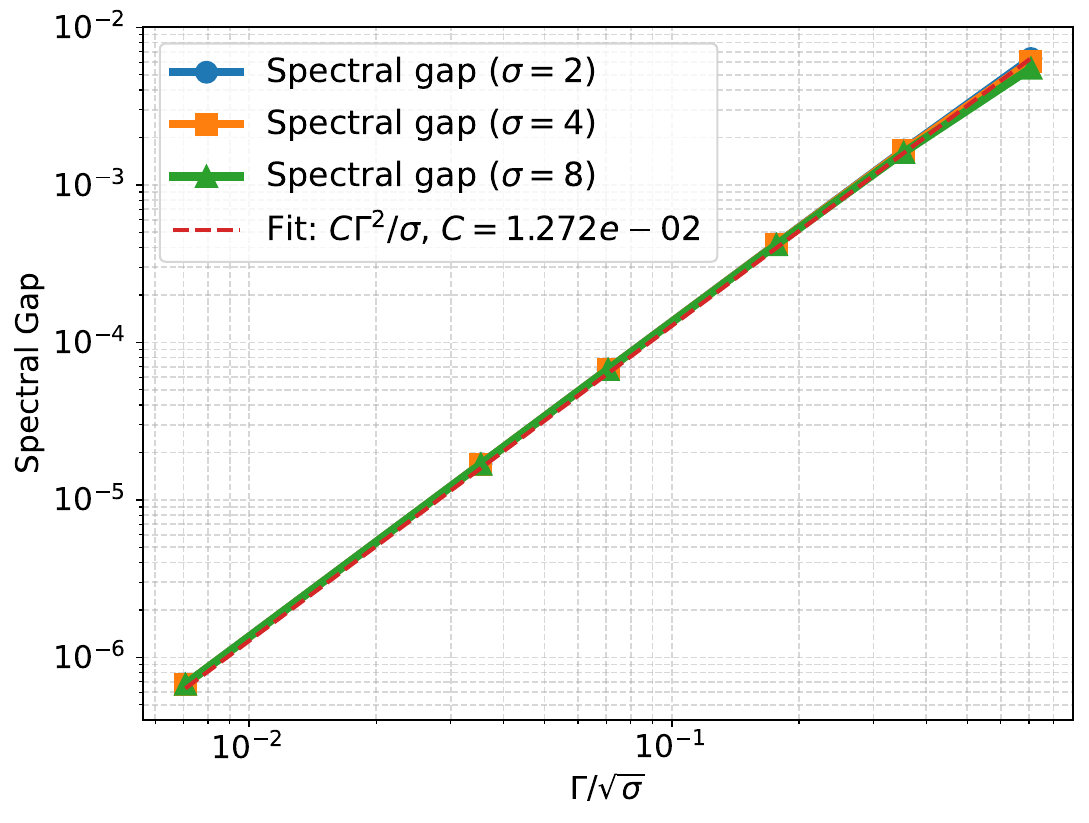}
        \caption{}
        \label{fig1:spectral2}
    \end{subfigure}
    \caption{Thermal state preparation for TFIM with $L=4$ sites in the regime $\Gamma=\Theta(1)$. In (a)--(c), we use $\sigma=2,4,8$, choose the coupling parameter ${\Gamma} = 1.0, 0.5, 0.25, 0.10, 0.05, 0.01$ and set the interaction time $T=5\sigma$; the frequency $\omega$ is sampled uniformly from the interval $[0,5]$. (a)The evolution of infidelity, i.e., $1 - F$, versus the iteration steps. (b) The fix point infidelity between the target thermal state and the stationary state of $\Phi_\Gamma$, shown versus different $\sigma$. (c) The spectral gap of $\Phi_\Gamma$.
    }
    \label{Fig:TFIM_4_thermal}
\end{figure}

\begin{figure}[!htbp]
    \centering
    \begin{subfigure}[b]{0.24\textwidth}
        \centering        \includegraphics[width=\textwidth]{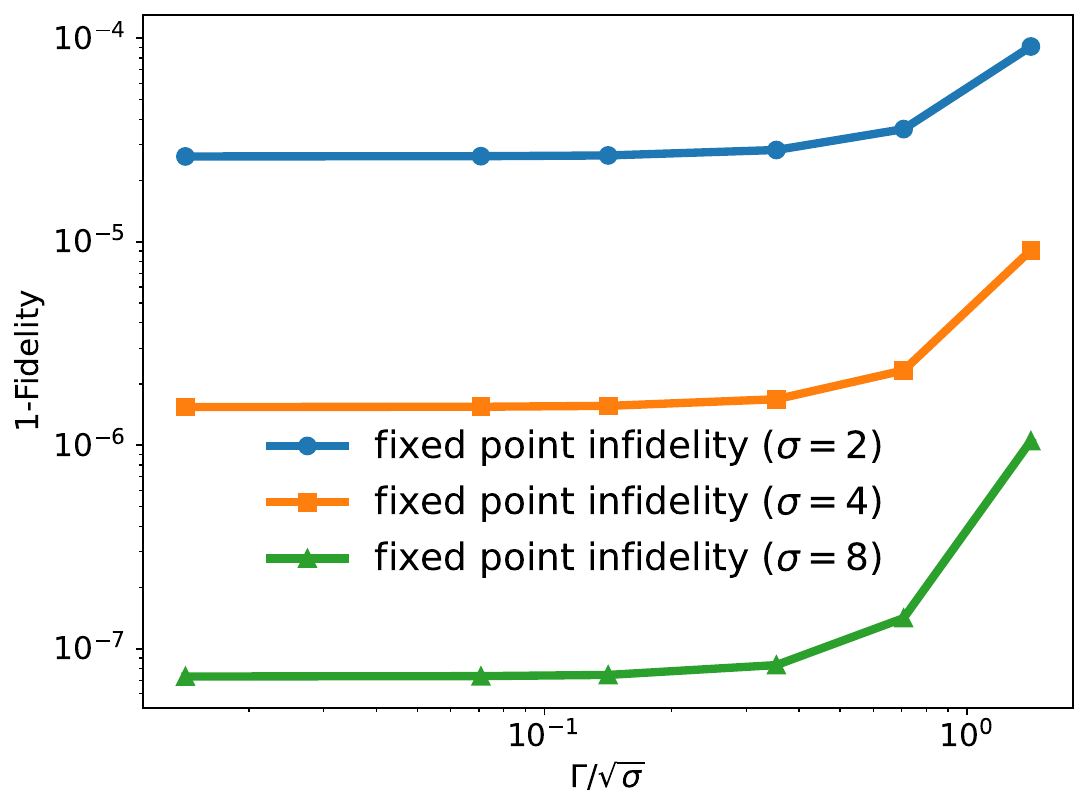}
        \caption{}
        \label{fig_large:dynamics_1}
    \end{subfigure}
    \hfill
    \begin{subfigure}[b]{0.24\textwidth}
        \centering        \includegraphics[width=\textwidth]{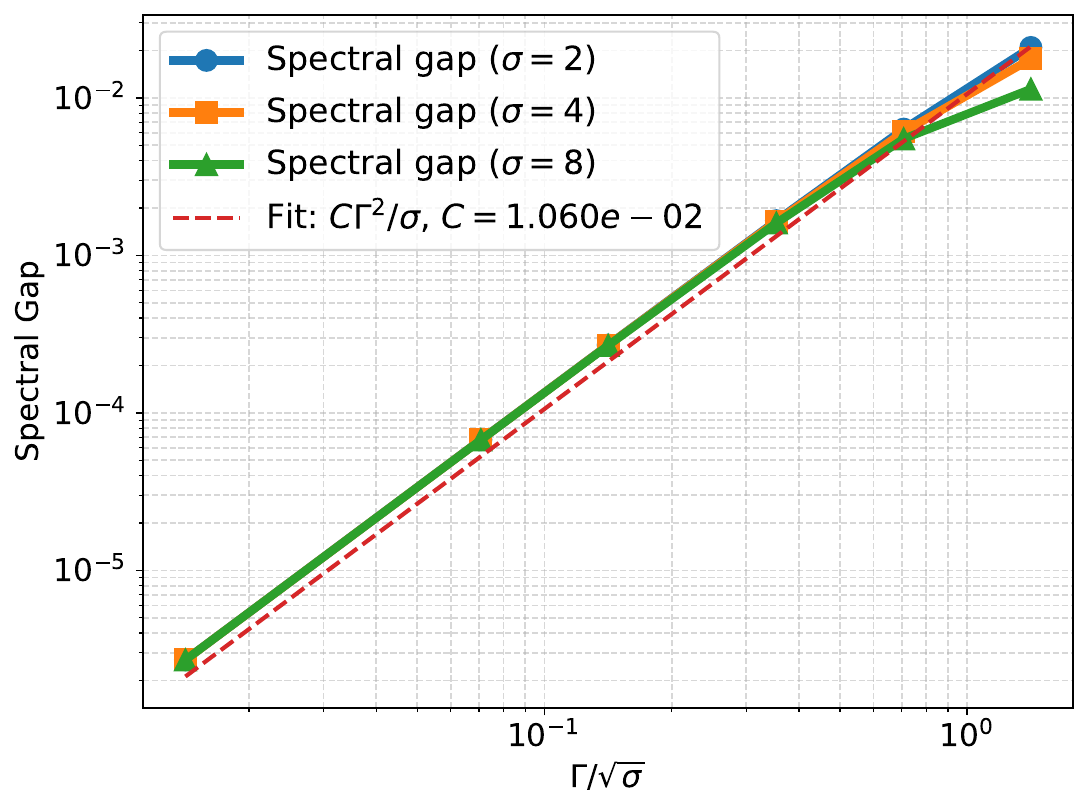}
        \caption{}
        \label{fig_large:dynamics_2}
    \end{subfigure}
    \begin{subfigure}[b]{0.24\textwidth}
        \centering
        \includegraphics[width=\textwidth]{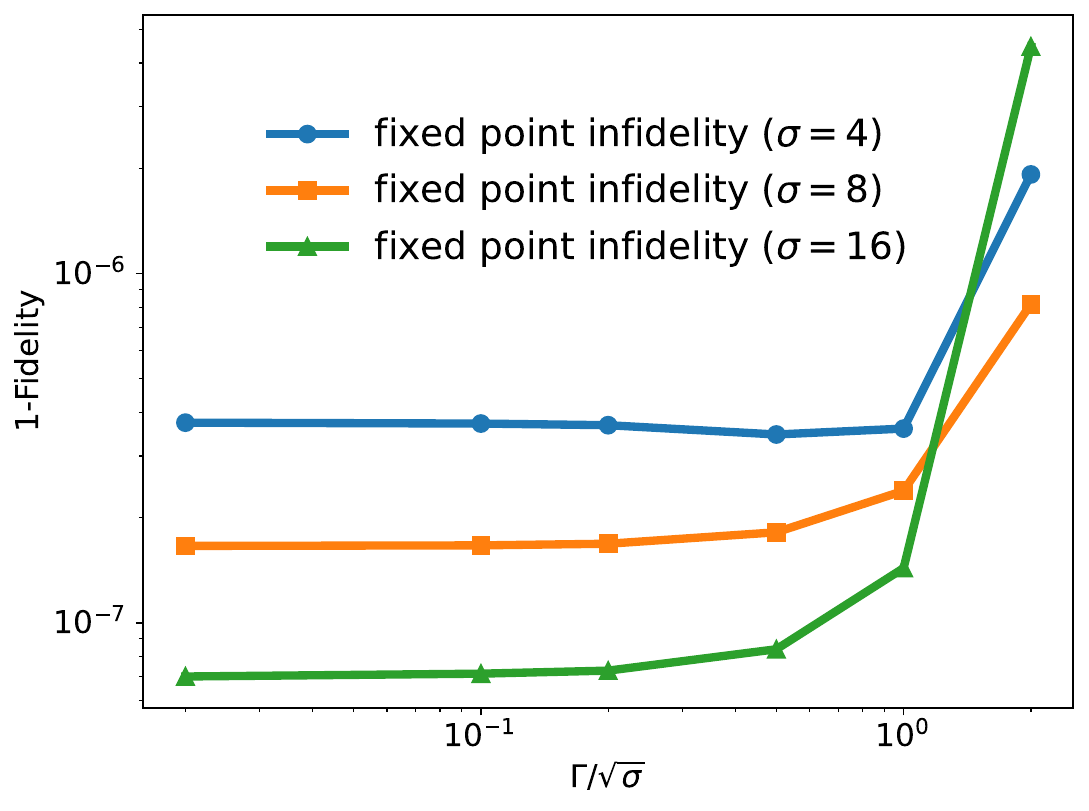}
        \caption{}
        \label{fig_large:error}
    \end{subfigure}
    \hfill
    \begin{subfigure}[b]{0.24\textwidth}
        \centering
        \includegraphics[width=\textwidth]{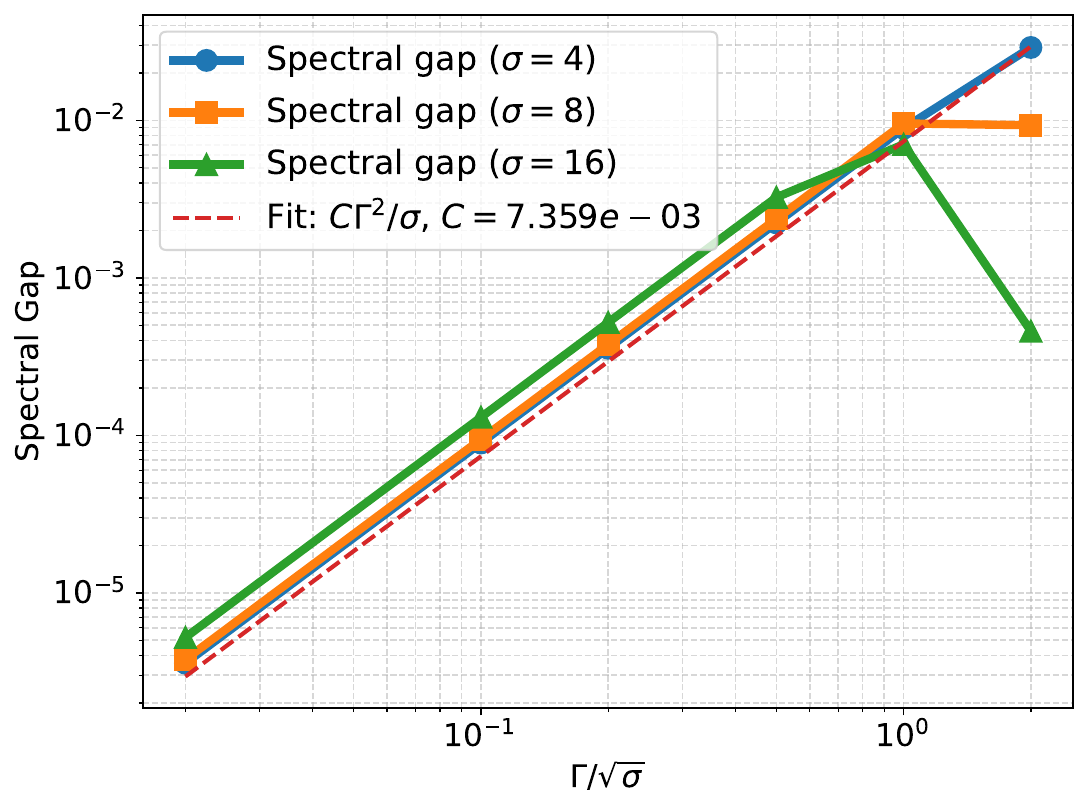}
        \caption{}
        \label{fig_large:spectral_gap}
    \end{subfigure}
    \caption{Thermal and groud state preparation for TFIM with $4$ sites in the regime $\Gamma = \Theta(\sigma)$. (a) Infidelity between target thermal state ($\beta=1$) and the stationary state of $\Phi_\Gamma$ for varying $\sigma$ and ${\frac{\Gamma}{\sqrt{2\sigma}}} = 1, 0.5, 0.25, 0.1, 0.05, 0.01$.(b) Spectral gap of $\Phi_\Gamma$ with different
        ${\frac{\Gamma}{\sqrt{\sigma}}}$, $\sigma$, and $\beta=1$ and the same values of  ${\frac{\Gamma}{\sqrt{2\sigma}}}$. (c) Infidelity between target ground state ($\beta=\infty$) and the stationary state of $\Phi_\Gamma$ for varying $\sigma$ and  ${\frac{\Gamma}{\sqrt{\sigma}}}=2, 1, 0.5, 0.2, 0.1, 0.02$. (d) Spectral gap of $\Phi_\Gamma$ with different
        ${\frac{\Gamma}{\sqrt{\sigma}}}$, $\sigma$, $\beta=\infty$, and the same values of ${\frac{\Gamma}{\sqrt{\sigma}}}$.}
    \label{Fig:TFIM_4_large_alpha}
\end{figure}

\paragraph{1-D Hubbard model}\label{sec:Hubbard_numerics}
Consider the 1-D Hubbard model defined on $L=2,4$ spinful sites with open boundary conditions in~\eqref{eqn:Hubbard} with $t = 1, U = -4$.
% \[
%     H = -t\sum_{j=1}^{L-1} \sum_{\sigma\in\{\uparrow,\downarrow\}} c_{j,\sigma}^\dag c_{j+1,\sigma}+U\sum_{j=1}^L (n_{j,\uparrow}-\frac{1}{2})(n_{j,\downarrow}-\frac{1}{2})
% \]
% where the number operator $n_{j,\sigma}= c_{j,\sigma}^\dag c_{j,\sigma}$ and the dimension is $2^{2L}$. Here we choose $t = 1, U = -4$.
Similar to the TFIM case, we consider the thermal state case with the same choice of parameters ${\Gamma}, \sigma$. We observe very similar results in~\cref{Fig:Hubbard_2_thermal} and~\cref{Fig:Hubbard_2_ground}. Furthermore, the numerical experiments can be extended to the regime where $\Gamma = \Theta(\sigma)$, as shown in~\cref{Fig:Hubbard_2_large_alpha}. Similar to the TFIM example, when ${\frac{\Gamma}{{\sigma}}}  \leq 1/2$, the fixed point closely matches the target state, and the spectral gap increases as ${\frac{\Gamma^2}{\sigma}}$.

\begin{figure}[!htbp]
    \centering
    \begin{subfigure}[b]{0.19\textwidth}
        \centering
        \includegraphics[ width=\textwidth]{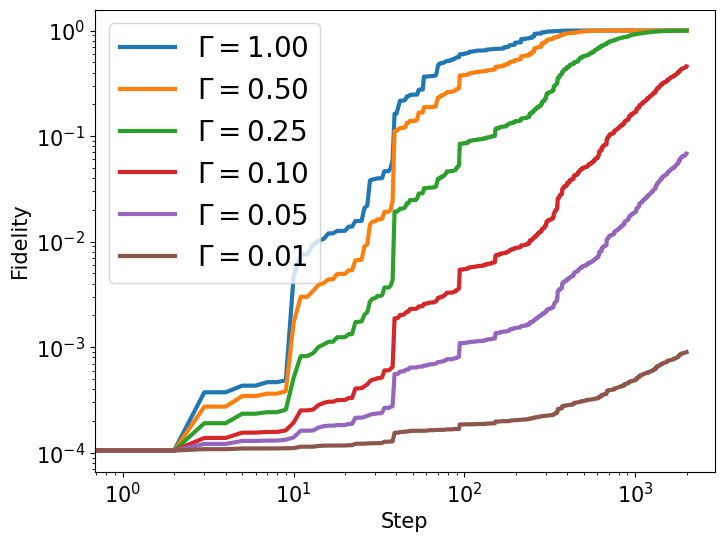}
        \caption{}
        \label{fig3:dynamics_1}
    \end{subfigure}
    % \hfill
    \begin{subfigure}[b]{0.19\textwidth}
        \centering        \includegraphics[width=\textwidth]{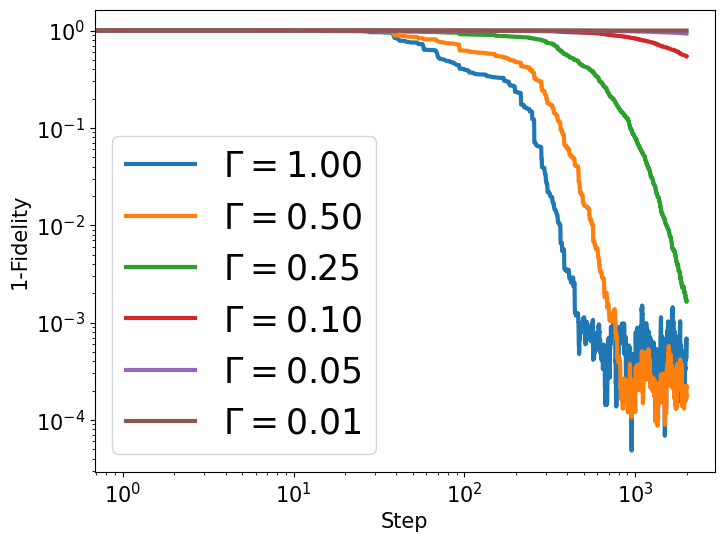}
        \caption{}
        \label{fig3:dynamics_2}
    \end{subfigure}
    \begin{subfigure}[b]{0.19\textwidth}
        \centering
        \includegraphics[width=\textwidth]{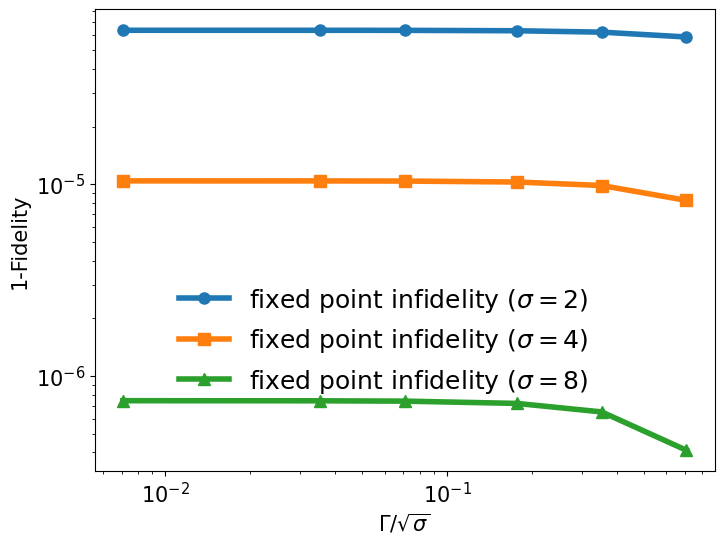}
        \caption{}
        \label{fig3:error}
    \end{subfigure}
    \begin{subfigure}[b]{0.19\textwidth}
        \centering        \includegraphics[width=\textwidth]{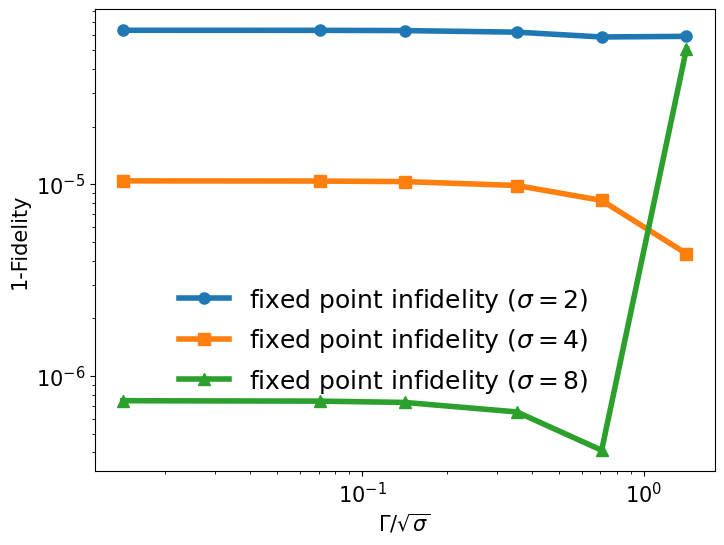}
        \caption{}
        \label{fig2_large:dynamics_1}
    \end{subfigure}
        \begin{subfigure}[b]{0.19\textwidth}
        \centering        \includegraphics[width=\textwidth]{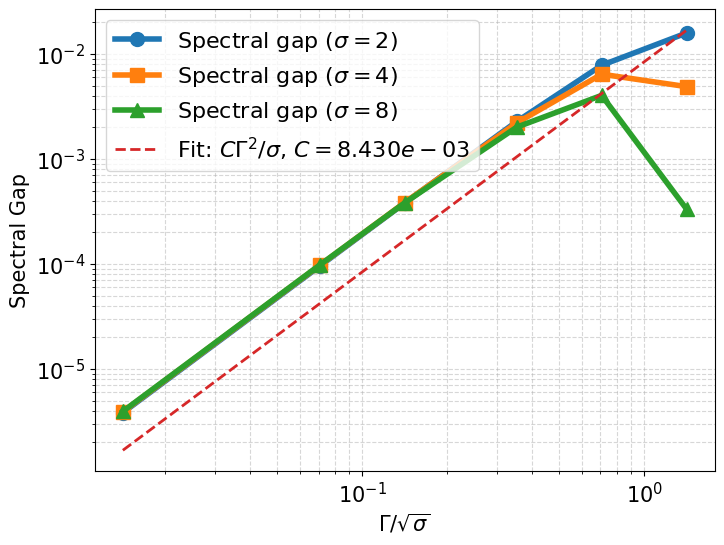}
        \caption{}
        \label{fig2_large_spectraldensity}
    \end{subfigure}
    % \hfill
    \caption{Thermal state preparation of Hubbard with $L = 2$ sites. (a)--(c) In the regime $\Gamma=\Theta(1)$.  (a) The evolution of fidelity. We set $\sigma=2$ and ${\Gamma}$. (b) The evolution of infidelity. (c) Infidelity between target thermal state and the stationary state of $\Phi_\Gamma$ with different $\sigma$. Fidelity increases with $\sigma$. (d)--(e) Strong coupling regime $\Gamma = \Theta(\sigma)$.  (d) Infidelity between target thermal state and the stationary state of $\Phi_\Gamma$ with different $\sigma$ and $\beta=1$. Here ${\frac{\Gamma}{\sqrt{2\sigma}}}=1, 0.5, 0.25, 0.1, 0.05, 0.01$. (e) Spectral gap.}
    \label{Fig:Hubbard_2_thermal}
\end{figure}

\subsection{Ground state preparation}\label{sec:ground_state_8qubits}
We also consider the ground state preparation of TFIM, Hubbard and ANNNI model.
For those large scale numerical experiments, including TFIM-8, Hubbard-4, we simulate the state vector instead of density operator to reduce the computational cost. Thus, we report the evolution of energy along a single trajectory. Unless stated otherwise, the remaining parameters are chosen as in the previous experiments. Similar to the thermal state case, even in the strong-coupling regime (${\frac{\Gamma}{{\sigma}}} \approx 0.5$), the algorithm still converges to the ground state with high accuracy, and the convergence speed increases as ${\frac{\Gamma}{\sqrt{\sigma}}}$ grows.
\paragraph{TFIM} We study the ground state preparation for the TFIM model with $L = 4, 8$ sites. In the regime $\Gamma=\Omega(1)$ with $L= 4$ (\cref{Fig:TFIM_4_ground}), we set ${\frac{\Gamma}{\sqrt{\sigma}}}=0.5, 0.25, 0.125, 0.05, 0.025, 0.005$, $\sigma=4,8,16$, $T=5\sigma$, and also sample $\omega$ uniformly from $[0,5]$. Similar to the thermal state case, figures in~\cref{fig2:dynamics_1} and~\cref{fig2:dynamics_2} justifies the convergence behavior of the fidelity and energy in the algorithm. In the strong coupling regime  $\Gamma=\Omega(\sigma)$ with size $L = 4$ (\cref{fig_large:error}, \cref{fig_large:spectral_gap}) and $L= 8$ (\cref{fig9:TFIM_8_enery}), we can still observe the convergence of the energy. Moreover, once the coupling parameter ${\frac{\Gamma}{\sqrt{\sigma}}}$ is below a moderate large constant,
the fidelity and spectral gap are essentially insensitive to the change of $\sigma$.
\begin{figure}[!htbp]
    \centering
    \begin{subfigure}[b]{0.24\textwidth}
        \centering        \includegraphics[width=\textwidth]{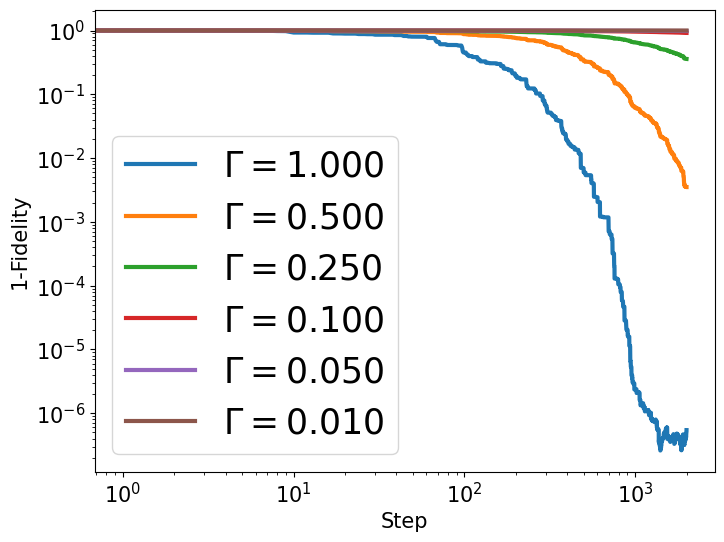}
        \caption{}
        \label{fig2:dynamics_1}
    \end{subfigure}
    % \hfill
    \begin{subfigure}[b]{0.24\textwidth}
        \centering        \includegraphics[width=\textwidth]{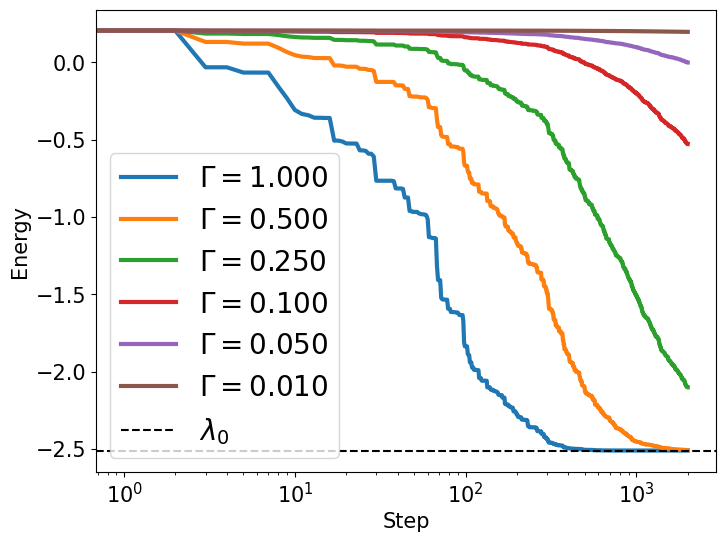}
        \caption{}
        \label{fig2:dynamics_2}
    \end{subfigure}
        \begin{subfigure}[b]{0.24\textwidth}
        \centering
        \includegraphics[width=\textwidth]{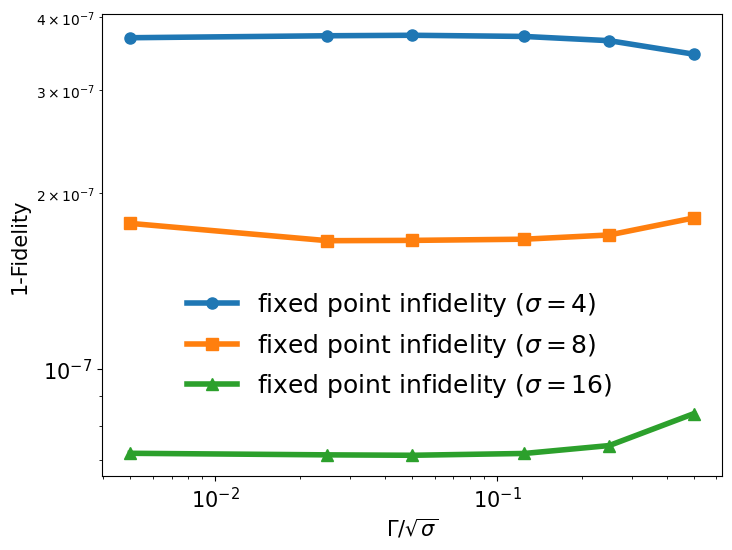}
        \caption{}
        \label{fig2:spectral_gap}
    \end{subfigure}
    \begin{subfigure}[b]{0.24\textwidth}
        \centering
        \includegraphics[width=\textwidth]{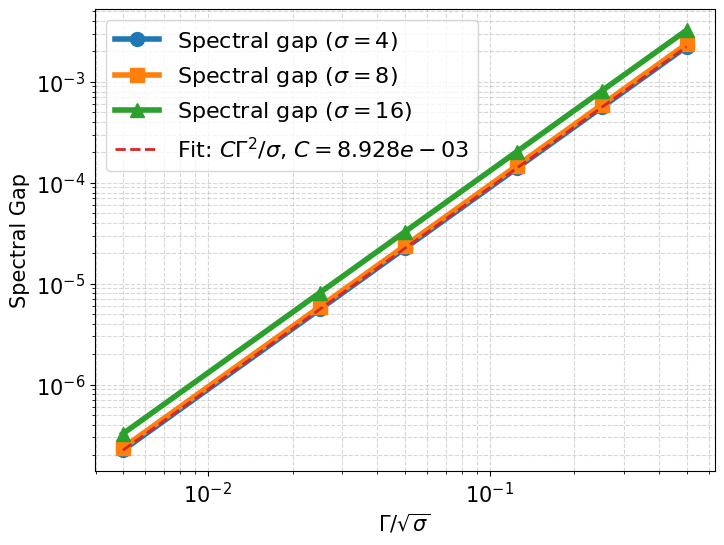}
        \caption{}
        \label{fig2:error}
    \end{subfigure}
    \caption{Ground state preparation for the TFIM with $L=4$ sites with ${\Gamma}=1.00, 0.50, 0.25, 0.10, 0.05, 0.01$ and $\sigma=4$. (a) infidelity dynamics,(b) energy dynamics, (c) spectral gap and (d) spectral gap. The ground state energy is denoted by $\lambda_0$.
    }
    \label{Fig:TFIM_4_ground}
\end{figure}
\paragraph{Hubbard model} We consider the ground state preparation for the Hubbard model with $L = 2, 4$ sites. In the regime  $\Gamma = \Theta(1)$ with sites $L = 2$ (\cref{Fig:Hubbard_2_ground}), the energy trajectories relax toward the ground state energy across a wide range of ${\frac{\Gamma}{\sqrt{\sigma}}}$ and converges faster as ${\frac{\Gamma}{\sqrt{\sigma}}}$ increases. The spectral gap scales approximately quadratically in ${\frac{\Gamma}{\sqrt{\sigma}}}$ while only mild depend on $\sigma$. The infidelity exhibit rapid decay over iterations. Specifically, the fixed point infidelity remains small and is approximately unchanged over different choice of $\sigma$. In the stronger coupling case with $L = 2$ (\cref{fig2_large:error,fig2_large:spectral_gap}) and $L = 4$ (\cref{fig9:Hubbard4_energy}), the fixed point infidelity and spectral gap curves indicate that: up to moderately large coupling, the preparation performance is essentially insensitive to $\sigma$, with noticeable deviation only as ${\frac{\Gamma}{\sqrt{\sigma}}}$ becomes large enough to slow mixing and reduce the achievable final accuracy.
\begin{figure}[!htbp]
    \centering
    \begin{subfigure}[b]{0.24\textwidth}
        \centering
        \includegraphics[width=\textwidth]{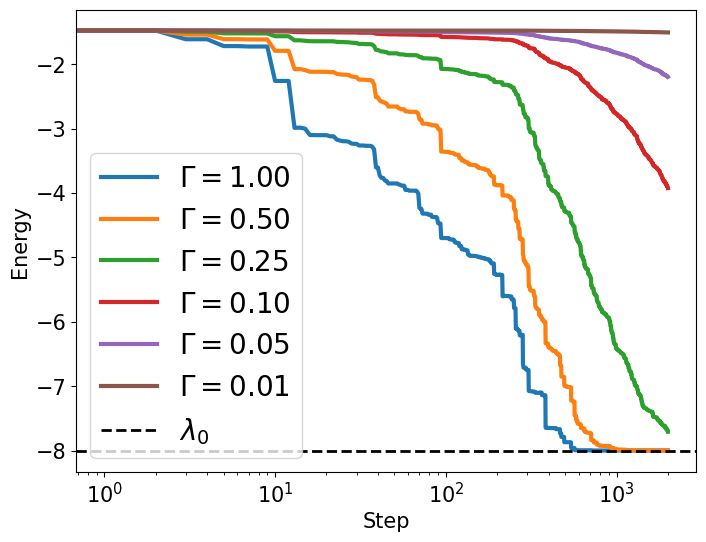}
        \caption{}
        \label{fig4:error}
    \end{subfigure}
    \begin{subfigure}[b]{0.24\textwidth}
        \centering        \includegraphics[width=\textwidth]{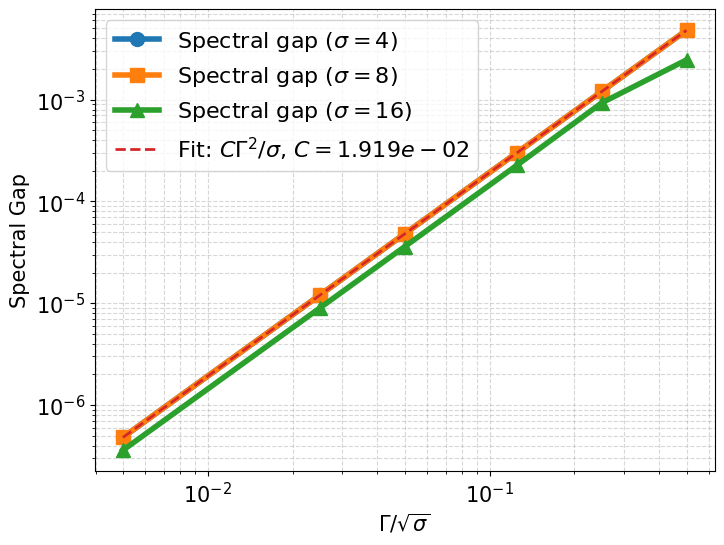}
        \caption{}
        \label{fig4:dynamics_2}
    \end{subfigure}
    \begin{subfigure}[b]{0.24\textwidth}
        \centering        \includegraphics[width=\textwidth]{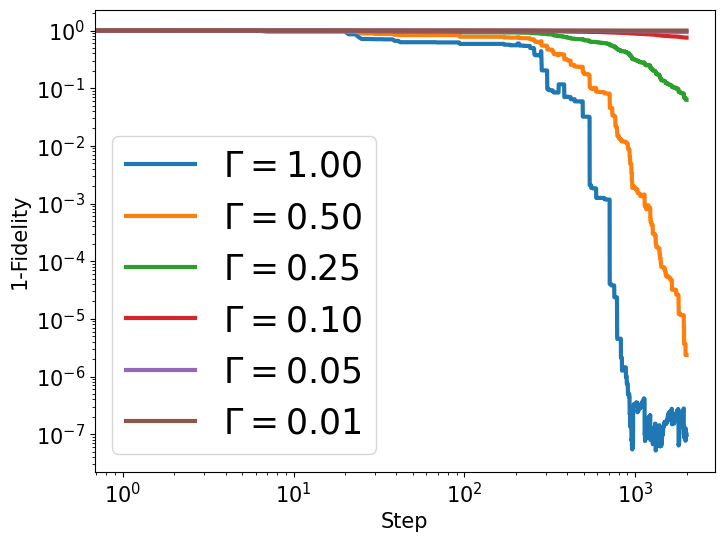}
        \caption{}
        \label{fig4:dynamics_3}
    \end{subfigure}
% \hfill
    \begin{subfigure}[b]{0.24\textwidth}
        \centering
        \includegraphics[width=\textwidth]{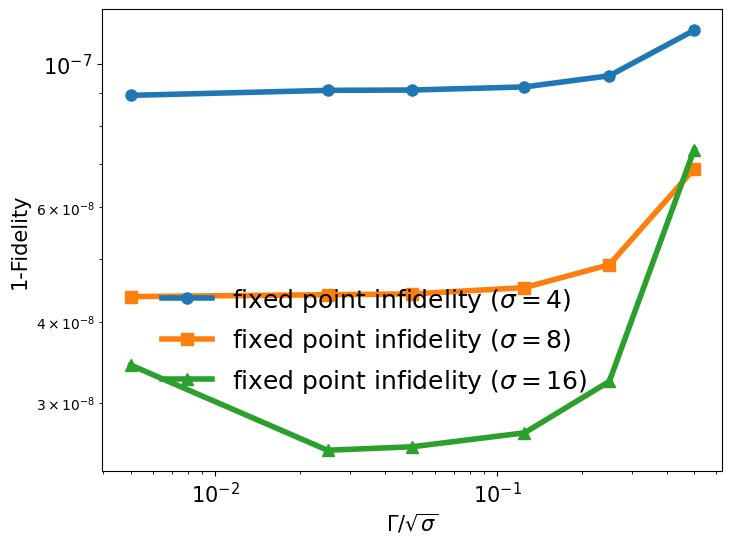}
        \caption{}
        \label{fig4:error2}
    \end{subfigure}
    \caption{Ground state preparation of the Hubbard model with $L=2$ sites in the regime $\Gamma=\Theta(1)$. We set ${\frac{\Gamma}{\sqrt{\sigma}}}=0.5, 0.25, 0.125, 0.05, 0.025, 0.005$, $\sigma=4,8,16$, $T=5\sigma$, and also sample $\omega$ uniformly from $[0,5]$. (a) The evolution of energy, $\lambda_0$ is the ground state energy. (b) Spectral gap of $\Phi_\Gamma$.
    (c) The evolution of infidelity. Here, we set $\sigma=4$ and ${\Gamma}$. (d) Infidelity between the target thermal state and the stationary state of $\Phi_\Gamma$ with different $\sigma$. Fidelity increases with $\sigma$.}
    \label{Fig:Hubbard_2_ground}
\end{figure}
\begin{figure}[!htbp]
    \centering
    \begin{subfigure}[b]{0.3\textwidth}
        \centering
        \includegraphics[width=\textwidth]{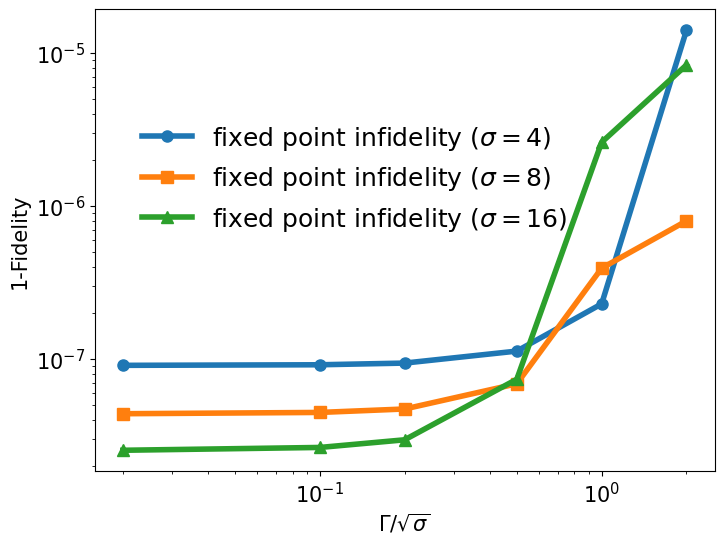}
        \caption{}
        \label{fig2_large:error}
    \end{subfigure}
    \begin{subfigure}[b]{0.3\textwidth}
        \centering
        \includegraphics[width=\textwidth]{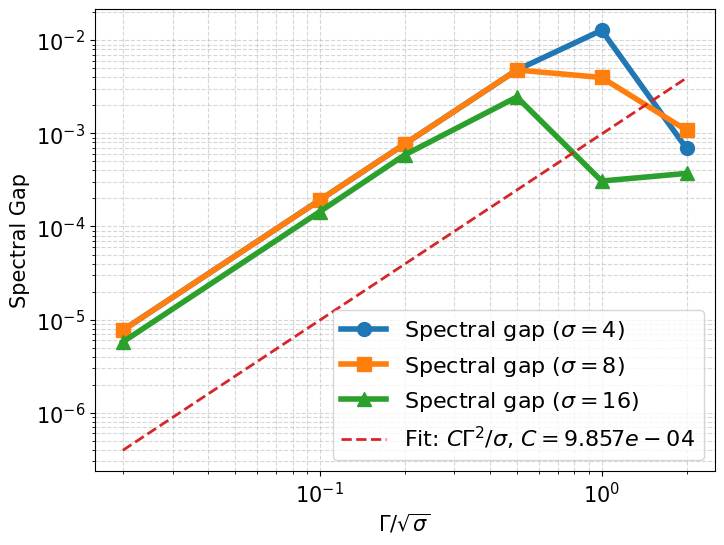}
        \caption{}
        \label{fig2_large:spectral_gap}
    \end{subfigure}
    \begin{subfigure}[b]{0.3\textwidth}
        \centering
        \includegraphics[width=\textwidth]{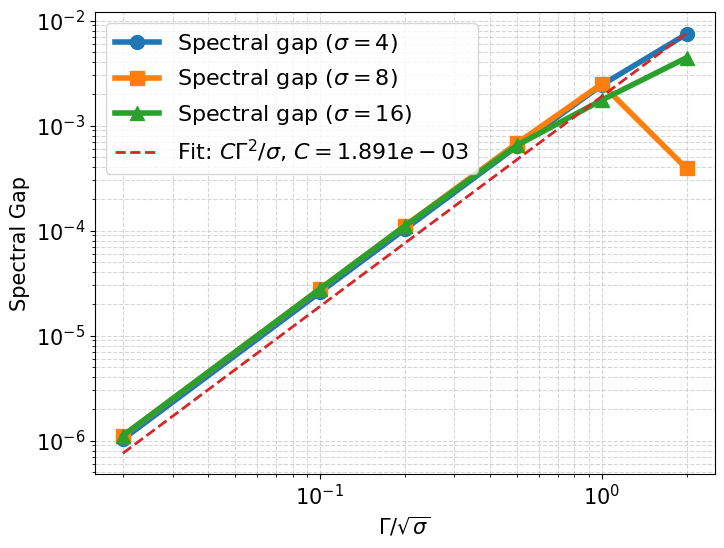}
        \caption{}
        \label{Fig:ANNNI}
    \end{subfigure}
    \caption{Ground state preparation  in the regime $\Gamma = \Theta(\sigma)$. (a)-(b) Hubbard model with $L = 2$ sites. The parameters are ${\frac{\Gamma}{\sqrt{\sigma}}}=2, 1, 0.5, 0.2, 0.1, 0.02$, $\sigma = 4,8,16$. (a) Infidelity between target ground state and the stationary state of $\Phi_\Gamma$. (b) Spectral gap of $\Phi_\Gamma$. (c) Ground state preparation for ANNNI model with $L = 4$ sites in the regime $\Gamma = \Theta(\sigma)$. Spectral gap with different $\sigma$.}
    \label{Fig:Hubbard_2_large_alpha}
\end{figure}

\paragraph{ANNNI model}\label{sec:ANNNI}
Consider the 1-D axial next-nearest-neighbor Ising (ANNNI) model with $L = 4$ sites defined as
\begin{equation}\label{eq:annni}
H_{\mathrm{ANNNI}}
=\frac{J_1}{4} \sum_{i} Z_i Z_{i+1}
+\frac{J_2}{4} \sum_{i} Z_i Z_{i+2}
-\frac{G}{2} \sum_{i} X_i,
\end{equation}
with $J_1=2$, $J_2=0.6$, $G=0.2$, and $\rho_0=\ket{0}\bra{0}$. In our test, we set ${\frac{\Gamma}{\sqrt{\sigma}}}=2, 1, 0.5, 0.2, 0.1, 0.02$, $\sigma =4, 8, 16$, and $T = 5\sigma$. The resulting spectral gaps are shown in~\cref{Fig:ANNNI}. Similar to the TFIM and Hubbard models, our algorithm converges to the correct ground state whenever ${\frac{\Gamma}{{\sigma}}} \le \tfrac{1}{2}$, and the convergence rate increases proportionally to ${\frac{\Gamma^2}{\sigma}}$.

\begin{figure}
    \centering
    \begin{subfigure}[b]{0.24\textwidth}
        \centering
        \includegraphics[width=\textwidth]{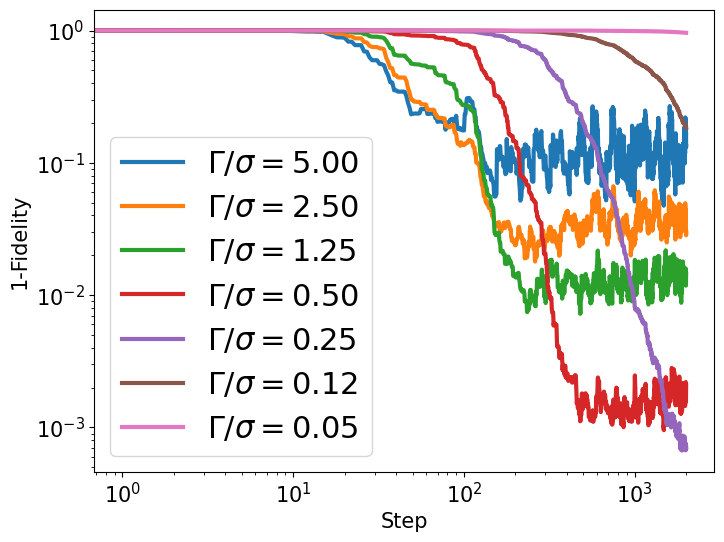}
        \caption{}
        \label{fig9:TFIM_8_infidelity}
    \end{subfigure}
    \begin{subfigure}[b]{0.24\textwidth}
        \centering
        \includegraphics[width=\textwidth]{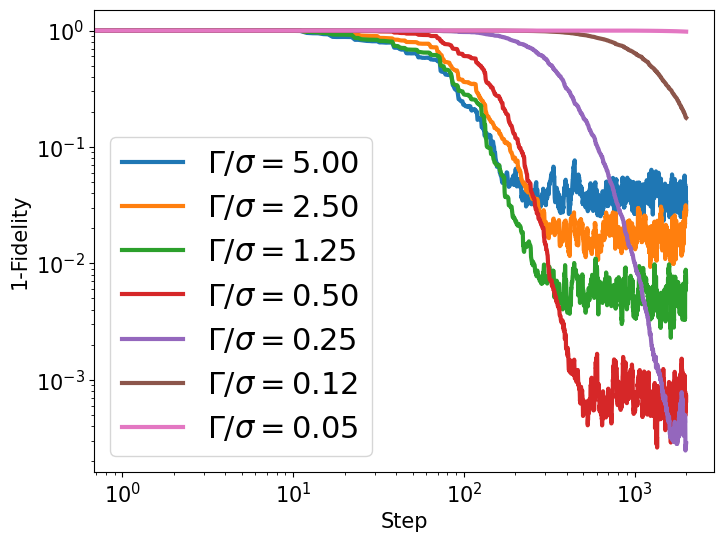}
        \caption{}
        \label{fig9:Hubbard4_infidelity}
    \end{subfigure}
    \begin{subfigure}[b]{0.24\textwidth}
        \centering        \includegraphics[width=\textwidth]{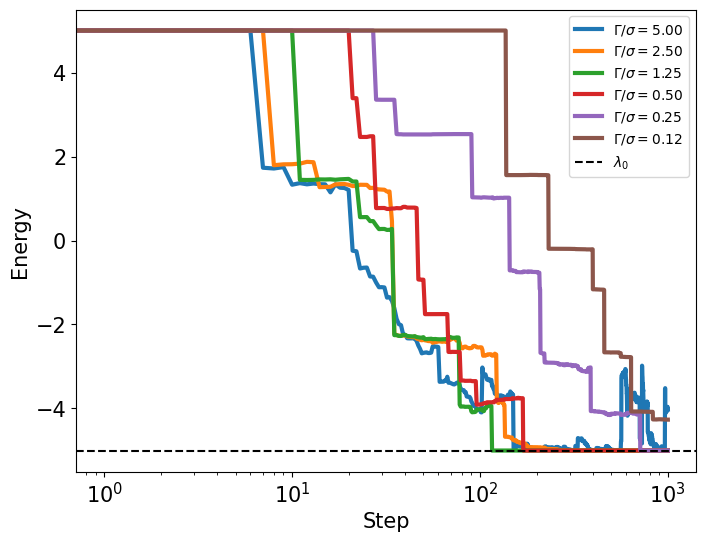}
        \caption{}
        \label{fig9:TFIM_8_enery}
    \end{subfigure}
    \begin{subfigure}[b]{0.24\textwidth}
        \centering
        \includegraphics[width=\textwidth]{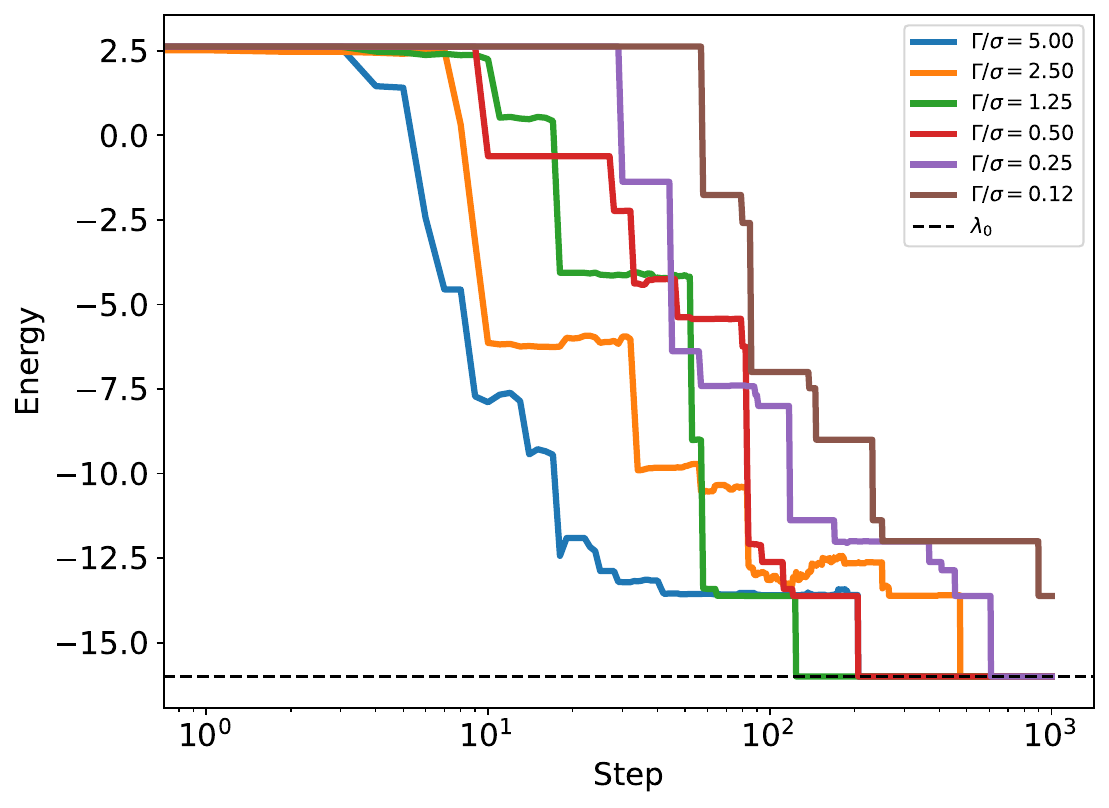}
        \caption{}
        \label{fig9:Hubbard4_energy}
    \end{subfigure}
    \caption{ State preparation for larger systems in the strong coupling regime $\Gamma = \Theta(\sigma)$.
    $\lambda_0$ denotes the ground state energy. (a) Infidelity of the thermal state preparation for TFIM with $L=8$ sites. (b) Infidelity of the thermal state preparation for the Hubbard model with $L=4$ sites. (c) Evolution of energy of ground state preparation for TFIM with $L=8$ sites. (d) Evolution of energy of ground state preparation for the Hubbard model with $L=4$ sites.}
    \label{Fig:large_alpha_ground}
\end{figure}